\documentclass[11pt]{amsart}

\usepackage{amsmath, amsthm, amssymb, amsfonts, enumerate, bbm}
\usepackage[colorlinks=true,linkcolor=blue,urlcolor=blue]{hyperref}
\usepackage{dsfont}
\usepackage{color}
\usepackage{geometry}

\usepackage{graphicx}
\usepackage{caption}
\usepackage{subcaption}
\usepackage{xcolor}
\usepackage{soul}

\def \cC{\mathcal{C}}
\def \cD{\mathcal{D}}

\def \cF{\mathcal{F}}
\def \cG{\mathcal{G}}
\def \cI{\mathcal{I}}

\def \cL{\mathcal{L}}
\def \cM{\mathcal{M}}
\def \cP{\mathcal{P}}

\def \P{\mathsf P}

\def \Q{\mathsf Q}
\def \E{\mathsf E}

\def \EQ{\mathsf E^{\mathsf Q}}
\def \EEQ{\mathsf E^{\tilde{\mathsf Q}}}
\def \R{\mathbb{R}}

\def \WW{\widetilde{W}}
\def \ud{\mathrm{d}}

\newcommand\bF{\mathbb{F}}

\newcommand\bN{\mathbb{N}}

\newcommand{\eps}{\varepsilon}
\newcommand{\ind}{\mathbbm{1}}

\newtheorem{theorem}{Theorem}[section]
\newtheorem{lemma}[theorem]{Lemma}
\newtheorem{corollary}[theorem]{Corollary}
\newtheorem{proposition}[theorem]{Proposition}

\newtheorem{remark}[theorem]{Remark}
\newtheorem{assumption}[theorem]{Assumption}

\DeclareMathOperator*{\esssup}{ess\,sup}

\makeatletter
\@namedef{subjclassname@2020}{\textup{2020}Mathematics Subject Classification}
\makeatother

\title[Variable annuities with surrender]{On variable annuities with surrender charges}
\author{Tiziano De Angelis \and Alessandro Milazzo \and Gabriele Stabile}
\address{T.\ De Angelis: (Corresponding author) School of Management and Economics, Dept.\ ESOMAS, University of Torino, Corso Unione Sovietica, 218 Bis, 10134, Torino, Italy; Collegio Carlo Alberto, Piazza Arbarello 8, 10122, Torino, Italy.}
\email{\href{mailto:tiziano.deangelis@unito.it}{tiziano.deangelis@unito.it}}
\address{A.\ Milazzo: School of Management and Economics, Dept.\ ESOMAS, University of Torino, Corso Unione Sovietica, 218 Bis, 10134, Torino, Italy.}
\email{\href{mailto:alessandro.milazzo@unito.it}{alessandro.milazzo@unito.it}}
\address{G.\ Stabile: Dipartimento di Metodi e Modelli per l'Economia, il Territorio e la Finanza, Sapienza-Universit\`{a} di Roma, Roma, Italy}
\email{\href{mailto:gabriele.stabile@uniroma1.it}{gabriele.stabile@uniroma1.it}}
\date{\today}

\numberwithin{equation}{section}

\begin{document}

\begin{abstract}
In this paper we provide a theoretical analysis of Variable Annuities (VAs) with a focus on the holder's right to an early termination of the contract. 
Motivated by risk-management considerations, and in line with existing literature, we assume that the surrender option is optimally exercised from a purely financial perspective (i.e., we consider the worst-case scenario for the insurer). In this context, we rigorously derive the pricing formula for the VA and characterise the optimal surrender time. 
We also illustrate our theoretical results with extensive numerical experiments.

The pricing problem is formulated as an optimal stopping problem with a time-dependent payoff which is discontinuous at the maturity of the contract and non-smooth. This structure leads to non-monotonic optimal stopping boundaries which we prove nevertheless to be continuous and regular in the sense of diffusions for the stopping set. The lack of monotonicity of the boundary makes it impossible to use classical methods from optimal stopping. Also more recent results about Lipschitz continuous boundaries are not applicable in our setup. Thus, we contribute a new methodology for non-monotone stopping boundaries.  
\end{abstract}

\maketitle

\section{Introduction}
We provide a theoretical analysis of Variable Annuities (VAs) with a focus on the holder's right to an early termination of the contract. VAs are insurance contracts designed to meet retirement and long-term investment goals, that combine the participation in equity performance with insurance coverage. 

The subscriber of a VA (policyholder) enters the contract by paying a premium (either single or periodic) which is invested in a selection of investment funds. Life-insurance protection is provided to the policyholder along with a number of financial guarantees and options. A minimum interest rate is guaranteed by the insurer in order to shield the policyholder's capital from market downturns. Moreover, the policyholder has the right to an early termination of the contract (so-called {\em surrender option}), which leads to a (partial) reimbursement of the account value. Such option can be exercised by the policyholder at any time prior to the maturity of the contract, hence sharing the financial feature of American options. However, there is a cost for early termination which is deducted from the account value and that may be quite substantial in the early stages of the contract. Such penalty, imposed by the insurer as a deterrent from early surrender, decreases over time and it vanishes at or prior to the maturity.
The guarantees offered by the insurer in the VA are financed by a fee paid periodically by the policyholder, often in the form of a fixed proportion of the account value.
At maturity, the policyholder can choose between claiming the account value or converting it into an annuity stream.

The features of VA contracts described above generate liabilities for the issuer which, if not properly evaluated, may put strain on the insurance company's solvency requirements. Pricing and risk management of those contracts require to take into account a variety of risks such as mortality risk, market risk and surrender risk. In particular, an unforeseen number of surrenders may cause serious liquidity concerns to the insurer together with a loss of future expected gains. In addition, the insurance company faces large upfront costs when issuing new contracts, that cannot be recovered in case of early surrender.

VAs have been studied extensively in the academic literature. In particular, the analysis of surrender risk has been recently the subject of several papers. One strand of the literature studies the influence of economic factors and policyholder's needs on surrender. Several papers adopt an {\em intensity-based framework} for
surrender modelling. For example, Russo et al.\ \cite{russo} consider a stochastic surrender intensity with Vasicek and Cox-Ingersoll-Ross models to derive closed-form expressions for the probability of the policyholder's surrender of the policy.
Ballotta et al.\ \cite{ballotta} model the intensity of surrender as a constant, representing policyholders' personal contingencies, plus a stochastic process based on the spread between the surrender benefit and the value at maturity of the guaranteed amount. De Giovanni \cite{degiovanni} builds a rational expectation model that describes the policyholder's behaviour in abandoning the contract (`lapsing' in the terminology of \cite{degiovanni}). Li and Szimayer \cite{li} propose a model in which surrender may occur because of either {\em exogenous} or {\em endogenous} surrender. Exogenous surrender occurs at a given deterministic intensity, whereas endogenous surrender can be chosen by the policyholder at the jump times of a Poisson process. The pricing problem is addressed via numerical methods for the associated PDEs. Jia et al. \cite{jia} develop a machine learning framework taking into account both human behaviour and rational optimality. Implementing a deep learning algorithm, the optimal surrender decision is estimated from a convex combination of the probability of surrendering based on market conditions and the probability of surrendering due to behavioural factors.

Another strand of the literature assumes that the policyholder is fully rational and the surrender option is exercised optimally exclusively from a financial perspective (i.e., no influence of personal needs is considered). Following this approach, the study of the VA contract with surrender option is set in the American-style contingent claim framework. However, the presence of mortality risk makes the analysis rather different from that performed in the case of American financial derivatives.

Bernard et al.\ \cite{BERNARD2014116} and Jeon and Kwak  \cite{JEON2021113508} consider VAs with guaranteed minimum accumulation benefit (GMAB) and surrender option. Both papers assume an exponentially decreasing surrender charge function and do not consider death benefits. In \cite{BERNARD2014116} the optimal surrender strategy is derived by splitting the value of the VA into a European part and an early exercise premium. 
In contrast to \cite{BERNARD2014116},  in \cite{JEON2021113508} the surrender benefit is linked to the largest between the account value and the GMAB. The authors observe that the VA features two surrender regions, one for large values of the fund and one for low values of the fund (the latter is due to the fact that the policyholder receives at least the minimum guaranteed {\em also} upon surrender). The approach in both \cite{BERNARD2014116} and \cite{JEON2021113508} is mostly numerical.

MacKay et al.\ \cite{MacKay} study VAs combining guaranteed minimum death benefit (GMDB) and GMAB (so-called {\em riders}) with the surrender option. They consider various fee structures, ranging from a constant fee paid continuously to a state-dependent fee, where the fee is charged only if the account value is below a fixed threshold. They also study the design of fees and surrender charges that eliminate the surrender incentive, while keeping the contract attractive to policyholders. 

Laundriault et al.\ \cite{landriault2021high} consider the same VA contract as in \cite{MacKay} but with an additional high-water mark fee structure. They study the related optimal stopping problem via numerical solution of a system of extended Hamilton-Jacobi-Bellmann equations.
Hilpert \cite{Hilpert} takes into account risk preferences in the form of risk-averse and loss-averse policyholders and uses a Least-Squares Monte Carlo method to study the associated optimal stopping problem. Dai et al.\ \cite{dai2008guaranteed} perform a numerical study via finite-difference methods of variable annuities with withdrawal benefits. Finally, Kirkby and Aguilar \cite{kirkby2023valuation} consider a VA contract under a L\'evy-driven equity market and study the contract's valuation and optimal surrender problem in the discrete time case.

Our paper contributes to the second strand of the literature by providing a theoretical analysis of VAs' rational pricing. It is important to emphasise that in reality VAs are not generally priced according to no-arbitrage principles, because the underlying fund may not be traded on the market (e.g., the insurance company has its own fund) and the policyholder may not have access to full market information that would allow an optimal exercise of the surrender option. Moreover, surrender decisions may not be driven purely by the policyholder's financial considerations. Other needs (e.g., unexpected medical expenses) or life events (e.g., unemployment) may result in surrender choices that are sub-optimal from a financial perspective. This feature has already been pointed out by many authors who adopted the rational pricing paradigm (cf.\ \cite{bauer2017policyholder}, \cite{BERNARD2014116}, \cite{JEON2021113508},  \cite{landriault2021high}, \cite{MacKay}). However, the same authors also highlight the importance of rational pricing and the study of optimal surrender behaviour as a tool for risk management of insurance companies. Indeed, rational pricing delivers an upper bound for the price of the VA contract, because it amounts to maximising the surrender value for the policyholder and, hence, to minimising profits for the insurer.

Leveraging on stochastic calculus and optimal stopping theory, we develop a fully theoretical analysis of VA contracts in the Black-Scholes framework with both GMDB and GMAB riders and embedded surrender option. These benefits are financed by a continuously paid fee which is taken as a fixed proportion of the account's value. The account's value is stochastic and it evolves as a geometric Brownian motion with constant risk-free rate and volatility. In the case of early surrender, a surrender charge is deducted from the account's value that is returned to the policyholder. The surrender charge is modelled via a decreasing function of time (similarly to, e.g., \cite{MacKay}) and the policyholder's residual lifetime is modelled with a deterministic force of mortality. We neglect transaction costs.

From an actuarial and financial perspective we contribute: (i) a rigorous analysis of VA contracts and of the optimal surrender strategy (in terms of an optimal surrender boundary), and (ii) an extensive numerical sensitivity analysis of the VA's price and surrender boundary with respect to model parameters. In particular, our analysis shows that the structure of the surrender charges is by far the main driver of VA's price variations. Instead, possible miss-specifications in the (deterministic) mortality force have an impact of smaller order of magnitude. We also argue that marketable VA contracts should have relatively low fees and penalties, in order to be appealing to a risk-neural policyholder when compared to, e.g., bond investment. Our assumption that the VA's account value is invested in a Black-Scholes market, with constant risk-free rate and volatility, is of course rather stylized compared to real financial markets. However, it is a common assumption in the literature (see  \cite{BERNARD2014116}, \cite{MacKay}, \cite{Hilpert}, \cite{JEON2021113508} among others) and, as far as we know, our paper provides the first rigorous theoretical analysis of VA contracts in this setting.

At the mathematical level, we develop new methods for the analysis of optimal stopping boundaries. The arbitrage-free price of the policy is the value function of a suitable finite-time horizon optimal stopping problem with a discounted reward which is {\em time-inhomogeneous} and {\em discontinuous} at the terminal time. Moreover, the payoff at maturity exhibits a kink and it is not continuously differentiable. The study of the optimal stopping boundary (surrender boundary) is made particularly complicated by this problem structure. In particular, it is impossible to establish monotonicity of the optimal boundary, which indeed often fails as illustrated in our numerical examples. Without such monotonicity, much of the existing theory for time-dependent optimal stopping boundaries is not applicable. At the technical level, we lose continuity of the boundary (which is proven under monotonicity assumptions in, e.g., \cite{cdeap2023}, \cite{de2015note} and \cite{peskir2019continuity}) and smoothness of the value function across the boundary, due to the lack of `regularity' of the stopping set in the sense of diffusions (see \cite{deape2020}).  In the past, when we encountered similar difficulties with monotonicity of optimal stopping boundaries, we were able to prove their local Lipschitz continuity (cf.\ \cite{de2019free} and \cite{de2019lip}). That resolved both the issues of continuity and regularity. Unfortunately, those methods seem to fail in the present setting (see detailed discussion in Section \ref{sec:meth}).

Standard PDE results do not apply either. Much of the classical literature on the Stefan problem considers PDEs without a source term and with zero boundary condition at the free boundary (e.g., \cite{cannon1968infinite,cannon1967existence,friedman1968one,friedman1975one}). Instead, we have a time-inhomogeneous running payoff that translates into a time-and-space-dependent source term in the PDE and it is not amenable to simplifications. In \cite{kotlow1973free} the obstacle is time-homogeneous and therefore the optimal boundary is monotonic. A classical reference for the study of parabolic variational inequalities with one-dimensional space variable and non-monotone free boundary is \cite{friedman1975parabolic}. However, the specific assumptions connecting the obstacle to the running payoff are incompatible with our setting (cf.\ Sec.\ 1 and Sec.\ 6 in \cite{friedman1975parabolic}). In the more modern literature as, e.g., \cite{blanchet2006regularity,blanchet2006continuity}, it is assumed that the source term does not change sign and it is strictly separate from zero. That is not the case in our problem, where the running payoff can vanish and it changes its sign\footnote{Prior to proving our Lemma 5.4 we only know that $t\mapsto b(t)$ is upper semi-continuous. Thus, even if we try to apply \cite{blanchet2006regularity,blanchet2006continuity} locally near the boundary, we still have difficulties meeting their sufficient conditions.}. Continuity of the time-derivative of the value is shown in \cite[Thm.\ 1.2]{blanchet2006continuity} only for monotone boundaries, which is not enough in our case. In summary, our problem appears to fall outside the scope of the classical literature on both probabilistic and PDE methods for parabolic free boundary problems.

Our new method goes as follows: first, we find a parametrisation of the optimal stopping problem for which we can prove existence of a stopping boundary $t\mapsto b(t)$ (the stopping set is below the boundary); second, we obtain a (negative) lower bound on the growth rate of the boundary and, third, we use such bound to define a monotonic increasing function $t\mapsto \beta(t)$; the function $t\mapsto\beta(t)$ is a bijective, continuous transformation of the optimal boundary $t\mapsto b(t)$ and it is shown to be an optimal stopping boundary in another parametrisation of our original problem. After these many steps, we have reduced the initial optimal stopping problem to an auxiliary one with a \textit{monotonic} boundary. Then we can continue our analysis taking advantage of existing tools in optimal stopping theory. In particular, we prove continuity of $t\mapsto \beta(t)$ and continuous differentiability of the value function in the new parametrisation. Both properties are then transferred back to the original problem, where we establish continuity of $t\mapsto b(t)$ and continuous differentiability of the VA's price. Finally, we derive an integral equation that allows an efficient numerical computation of the surrender boundary. 
We believe that our paper is the first one to implement such approach for time-dependent optimal stopping boundaries. 

The rest of the paper is organised as follows. In Section \ref{sec:setup} we provide the mathematical formulation of the problem and we state our main results (Theorem \ref{thm:main}). In Section \ref{Sec:NumEvid} we illustrate our results and draw economic conclusions based on a careful sensitivity analysis. The theoretical analysis of the paper starts with Section \ref{sec:contb}, where we obtain continuity of the value function and existence of an optimal stopping boundary $t\mapsto b(t)$. In Section \ref{sec:beta} we provide the transformed boundary $t\mapsto \beta(t)$. In Section \ref{sec:wSmooth} we obtain continuous differentiability of the value function and, finally, in Section \ref{sect:bCont} we prove continuity of the boundary $b$ and obtain the integral equation that characterises the surrender boundary. The paper is completed by Section \ref{sec:meth} where we outline broader applications of the methodology devised in this paper and revolving around one-sided Lipschitz bounds for optimal stopping boundaries.


\section{Actuarial model, problem formulation and main results}\label{sec:setup}

Given $T>0$, we consider a financial market with finite time horizon $[0,T]$ on a complete probability space $(\Omega, \cF, (\cF_t)_{t\in[0,T]}, \Q)$ that carries a
one-dimensional Brownian motion $\WW:=(\WW_{t})_{t\in[0,T]}$. With no loss of generality we assume that the filtration $\bF:=(\cF_t)_{t\in[0,T]}$ is generated by the Brownian motion and it is completed with $\Q$-null sets. 
 
We consider a VA contract with maturity $T$. The VA's account value process is denoted by $X:=(X_{t})_{t\in[0,T]}$. We assume that $X$ evolves as a geometric Brownian motion under $\Q$. That is,
\begin{align}\label{dynX0}
\left \{
\begin{array}{l}
\ud X_{t}= X_{t}\big[(r-c)  \ud t + \sigma  \ud \WW_{t}\big],  \\
X_{0} = x_0>0,
\end{array} \right.
\end{align}
where $r\ge 0$ is the risk-free rate, $\sigma>0$ is the market volatility, $c\ge 0$ denotes a constant fee rate levied by the insurance company (constant percentage of the account value), and $x_0>0$ is the initial account value. Notice that, according to the standard theory, setting $\tilde X_t:= e^{ct}X_t$ we have that $e^{-rt}\tilde X_t$ is a $(\bF,\Q)$-martingale and $\Q$ is the unique risk-neutral measure on the financial market.

Since the VA contract provides death benefits, we must include mortality risk in our model.
We assume that the policyholder is aged $\eta>0$ at time zero and this value is given and fixed throughout the paper. We define a random variable $\tilde \tau_d$
that represents the time of death of the policyholder and we assume $\tilde \tau_d$ independent of the Brownian motion (hence $\tilde \tau_d$ independent of $\cF_\infty$). More precisely, we define $\tilde \tau_d$ on another probability space $(\Omega^d,\cF^d,\Q^d)$ and work with the product space
\[
\big(\tilde \Omega,\tilde\cF,\tilde\Q\big):=\big(\Omega\times\Omega^d,\cF\times\cF^d,\Q\times\Q^d\big).
\]
Random variables $Z_1:\Omega\to \R$ and $Z_2:\Omega^d\to \R$ are redefined on the product space $(\tilde\Omega,\tilde\cF)$ in the usual way: $Z_1(\tilde \omega)=Z_1(\omega)$ and $Z_2(\tilde \omega)=Z_2(\omega^d)$ for $\tilde\omega=(\omega,\omega^d)\in\tilde\Omega$.

We denote by 
$_s \tilde p_{\eta+t}$,
with $s,t\ge 0$, the probability that an individual who is alive at the age $\eta+t$ survives to the age $\eta+t+s$. That is $_s\tilde p_{\eta+t}=\tilde\Q(\tilde \tau_d>\eta\!+\!t\!+\!s|\tilde \tau_d>\eta\!+\!t)=\Q^d(\tilde \tau_d>\eta\!+\!t\!+\!s|\tilde \tau_d>\eta\!+\!t)$ and,
letting $\tilde{\mu}:[0,+\infty)\to(0,+\infty)$ be the mortality force, we have
\begin{align}\label{pS}
_s \tilde{p}_{\eta+t}=\exp \left( -\int_{0}^{s}\tilde{\mu}(\eta+t+u)du\right),\quad\text{for $t,s\ge0$}.
\end{align}
Since the age $\eta>0$ of the policyholder is fixed throughout the paper, we simplify our notation and define $\tau_d:=\tilde\tau_d-\eta$, $\mu(t):=\tilde{\mu}(\eta+t)$ and $_s p_t:=\, _s \tilde{p}_{\eta+t}$ for every $t,s\geq 0$. 

Furthermore, our VA includes minimum rate guarantees on the interest credited to the account value and it allows the policyholder to exercise an early surrender option. The so-called intrinsic value of the policy is the value that the policyholder receives either at the maturity of the policy or at an earlier time, should she die or in case she decides to exercise the surrender option. Now we describe the structure of the intrinsic value.

At maturity, the VA's payoff reads 
\[
\max\big\{x_0e^{gT},X_T\big\},
\]
with $g\ge 0$ denoting the guaranteed minimum interest rate per annum.
In case of the policyholder's death prior to the maturity of the contract, the insurer provides a death benefit equal to 
\[
\max\big\{x_0e^{g \tau_d},X_{\tau_d}\big\}.
\]
The contract neither punishes nor rewards death and the form of the payoff is the same as the one at maturity but with $\tau_d$ in place of $T$ (see, e.g.,  \cite{Hilpert}). Finally, if the policyholder decides to surrender the contract at any time $0\le \tau \le T$, she receives
\begin{align*}
\big(1-k(\tau)\big)X_{\tau},
\end{align*}
where $k(\tau)X_\tau$ is a so-called surrender charge. The structure of $k(t)$, $t\in[0,T]$, is specified in the contract and it is designed to disincentivise early surrender (which carries substantial liquidity risk for the insurer). Typically, $0\le k(t)\le 1$ for all $t\in[0,T]$ and $k(T)=0$. In practice, insurance companies try to prevent early surrenders, especially in the first years of the contract, whereas they are more lenient towards the maturity. Most contracts feature steep penalties (large $k(t)$) at the beginning of the contract and gradually decreasing penalties moving closer to the maturity. In some cases the charges vanish after a certain number of years from the beginning of the contract.

The rational price at time zero of a VA {\em without} early surrender option, is given by
\begin{align}
\label{U0}
U_0:= \EEQ \Big[e^{-r\, (\tau_d \wedge T)}\max\{x_0e^{g (\tau_d \wedge T)},X_{\tau_d \wedge T}\}\Big],
\end{align}
where $\EEQ$ denotes the expectation under the product measure $\tilde\Q$. When the policyholder is given the option to early surrender, this adds value to the contract and it must be taken into account in the pricing formula for the VA. Denoting by $V_0$ the rational price of the VA {\em with} surrender option at time zero, we have
\begin{equation}\label{value0}
\begin{aligned}
V_0= \sup_{0\leq\tau\leq T}\EEQ\Big[\ind_{\{\tau<\tau_d \wedge T \}}e^{-r \tau}\big(1-k(\tau)\big) X_{\tau}+\ind_{\{\tau \ge \tau_d \wedge T \}}e^{-r\, (\tau_d \wedge T)}\max\{x_0e^{g (\tau_d \wedge T)},X_{\tau_d \wedge T}\}\Big],
\end{aligned}
\end{equation}
where the supremum is taken over a class of stopping times $\tau\in[0,T]$ that we are going to specify next. It is well-known that the filtration $\bar\bF=(\bar \cF_t)_{t\in[0,T]}$, with $\bar \cF_t=\cap_{\eps>0}\bar\cF^0_{t+\eps}$ and $\bar\cF^0_t:=\cF_{t}\vee\sigma(t\wedge\tau_d)$, is the smallest enlargement of $\bF$ that makes $\tau_d$ a stopping time and it satisfies the usual assumptions (cf.\  \cite[Sec.\ VI.3, p.\ 378]{protter2012stochastic}). It may seem that we should consider $\bar\bF$-stopping times, but for any $\bar\bF$-stopping time $\bar\tau$ there is a $\bF$-stopping time $\tau$ such that $\bar\tau\wedge\tau_d=\tau\wedge\tau_d$, $\tilde\Q$-a.s. (see lemma in \cite[Sec.\ VI.3, p.378]{protter2012stochastic}). Then, for any $\bar \bF$-stopping time $\bar \tau$, using $\{\bar \tau<\tau_d\wedge T\}=\{\bar \tau\wedge \tau_d<\tau_d\wedge T\}$ we have
\begin{equation*}
\begin{aligned}
&\EEQ\Big[\ind_{\{\bar \tau<\tau_d \wedge T \}}e^{-r \bar \tau}\big(1-k(\bar \tau)\big) X_{\bar \tau}+\ind_{\{\bar \tau \ge \tau_d \wedge T \}}e^{-r\, (\tau_d \wedge T)}\max\{x_0e^{g (\tau_d \wedge T)},X_{\tau_d \wedge T}\}\Big]\\
&=\EEQ\Big[\ind_{\{\bar \tau\wedge\tau_d<\tau_d \wedge T \}}e^{-r (\bar \tau\wedge\tau_d)}\big(1-k(\bar \tau\wedge\tau_d)\big) X_{\bar \tau\wedge\tau_d}\\
&\qquad\quad+\big(1-\ind_{\{\bar \tau\wedge\tau_d < \tau_d \wedge T \}}\big)e^{-r\, (\tau_d \wedge T)}\max\{x_0e^{g (\tau_d \wedge T)},X_{\tau_d \wedge T}\}\Big]\\
&=\EEQ\Big[\ind_{\{\tau\wedge\tau_d<\tau_d \wedge T \}}e^{-r (\tau\wedge\tau_d)}\big(1-k(\tau\wedge\tau_d)\big) X_{\tau\wedge\tau_d}\\
&\qquad\quad+\big(1-\ind_{\{\tau\wedge\tau_d < \tau_d \wedge T \}}\big)e^{-r\, (\tau_d \wedge T)}\max\{x_0e^{g (\tau_d \wedge T)},X_{\tau_d \wedge T}\}\Big]\\
&=\EEQ\Big[\ind_{\{\tau<\tau_d \wedge T \}}e^{-r \tau}\big(1-k(\tau)\big) X_{\tau}+\ind_{\{\tau \ge \tau_d \wedge T \}}e^{-r\, (\tau_d \wedge T)}\max\{x_0e^{g (\tau_d \wedge T)},X_{\tau_d \wedge T}\}\Big],
\end{aligned}
\end{equation*}
where $\tau$ is a $\bF$-stopping time. Thus, with no loss of generality we restrict the class of stopping times in \eqref{value0} to those for the filtration $\bF$.

Throughout the paper we make the following standard assumptions on the penalty function and on the mortality force. 
\begin{assumption}
\label{A1}
The penalty function $t\mapsto k(t)$ is non-increasing and twice continuously differentiable in $[0,T]$ with $k(T)=0$. The mortality force $t\mapsto \mu(t)$ is continuously differentiable in $[0,\infty)$ with 
$\dot\mu(t)\leq \lambda \mu(t)$ for some $\lambda\ge 0$.
\end{assumption} 
Notice that Assumption \ref{A1} imposes (rather naturally) that the lifetime ageing rate $\dot\mu(t)/\mu(t)$ (see, e.g., \cite{horiuchi1990age}) be bounded from above and that there is no surrender charge at maturity of the contract. The latter is compatible with situations in which the surrender charge vanishes at (or strictly prior to) the maturity.  
\begin{remark}
Several papers studying surrender risk consider an exponentially decreasing surrender charge function (see, e.g., \cite{Hilpert} and \cite{MacKay}) of the form
\begin{equation}
\label{kt}
k(t)=1-e^{-K(T-t)},
\end{equation}
where $K>0$ is the surrender charge intensity. We take this function as our benchmark. It will be shown in Remark \ref{remarkt*} that for $K\ge c$ the policyholder should never surrender the contract, because the penalty for an early surrender is too high compared to the contract fees.
\end{remark}

\subsection{Markovian formulation via change of measure}\label{subs:markovian}

In order to obtain the price $V_0$ of the VA, we embed the problem in a Markovian framework. First we use independence of $\tau_d$ from $\tilde\cF_\infty:=\cF_\infty\times\{\varnothing,\Omega^d\}$ and Fubini's theorem to integrate out the force of mortality. This yields an optimal stopping problem on the filtered space $(\Omega,\cF,(\cF_t)_{t\in [0,T]},\Q)$. Then, we extend the formulation of such problem for an arbitrary initial time $t\in[0,T]$ and an arbitrary initial position $x>0$ of the process $X$. Finally, we perform a change of probability measure that leads to a more convenient problem formulation. These steps take us to the formulation of a Markovian optimal stopping problem with finite-time horizon and with a discounted reward which is time-inhomogeneous and discontinuous at the terminal time.

Using independence of $\tau_d$ from $\tilde\cF_\infty$, we obtain for any $\bF$-stopping time $\tau\in[0,T]$
\begin{equation*}
\begin{aligned}
&\EEQ\Big[\ind_{\{\tau<\tau_d \wedge T \}}e^{-r \tau}\big(1\!-\!k(\tau)\big) X_{\tau}\!+\!\ind_{\{\tau \ge \tau_d \wedge T \}}e^{-r\, (\tau_d \wedge T)}\max\{x_0e^{g (\tau_d \wedge T)},X_{\tau_d \wedge T}\}\Big|\tilde\cF_\infty\Big]\\
&=\tilde\Q\big(\tau<\tau_d\big|\tilde\cF_\infty\big)\ind_{\{\tau<T\}}e^{-r \tau}\big(1\!-\!k(\tau)\big) X_{\tau}\\
&\quad+\!\EEQ\Big[\ind_{\{\tau \ge \tau_d\}\cap\{\tau_d\le T\}}e^{-r\, \tau_d}\max\{x_0e^{g \tau_d},X_{\tau_d}\}\Big|\tilde\cF_\infty\Big]\\
&\quad+\!\ind_{\{\tau = T \}}e^{-r\, T}\max\{x_0e^{g T},X_{T}\}\tilde\Q\big(\tau_d>T\big|\tilde\cF_\infty\big)\\
&=e^{-\int_0^\tau\mu(s)\ud s}\ind_{\{\tau<T\}}e^{-r \tau}\big(1\!-\!k(\tau)\big) X_{\tau}+\int_0^{\tau}e^{-\int_0^t\mu(s)\ud s-rt}\mu(t)\max\{x_0e^{gt},X_t\}\ud t\\
&\quad+\!\ind_{\{\tau = T \}}e^{-r\, T}\max\{x_0e^{g T},X_{T}\}e^{-\int_0^T\mu(s)\ud s},
\end{aligned}
\end{equation*}
where in the final expression we used the so-called freezing lemma (see \cite[Lemma 4.1]{baldi2017stochastic}) for the first two terms and $\tau\wedge T=\tau$ for the second term. Moreover, we have implicitly used that $\tilde \Q(\tau_d>0|\tilde\cF_\infty)=\Q^d(\tau_d>0)=1$. Applying the tower property in the formula for $V_0$ (cf.\ \eqref{value0}) and substituting the expression above under expectation, yields
\begin{equation*}
\begin{aligned}
V_0=\sup_{0\le \tau\le T}\EQ&\Big[\ind_{\{\tau<T\}}e^{-r \tau} {_\tau p_0}\big(1\!-\!k(\tau)\big) X_{\tau}+\int_0^{\tau}e^{-rt}{_t p_0}\mu(t)\max\{x_0e^{gt},X_t\}\ud t\\
&+\!\ind_{\{\tau = T \}}e^{-r\, T}{_T p_0}\max\{x_0e^{g T},X_{T}\}\Big],
\end{aligned}
\end{equation*}
where we recall that $_t p_0=\exp(-\int_0^t\mu(s)\ud s)$. Notice that the product measure has now disappeared and we are working on $(\Omega,\cF,(\cF_t)_{t\in[0,T]},\Q)$.

Next, we want to formulate the problem as if a VA was purchased at time zero with initial account value $x_0$ and we are now sitting at time $t\in[0,T]$ where we observe an account value equal to $X_t$. On the event $\{\tau_d\le t\}$ the contract is liquidated at or before time $t$ and we are not interested in this case. Therefore, we assume throughout that we are on the event $\{\tau_d>t\}$. At time $t$, and on the event $\{\tau_d>t\}$, the expected payoff of the VA must be computed conditionally on the information contained in $\cF_t$. That leads us to defining
\begin{align}
\label{eq:Vt}
\begin{aligned}
e^{-rt}V_t:= \esssup_{t\le \tau\le T}\EQ&\Big[\ind_{\{\tau<T\}}e^{-r \tau} {_{\tau-t} p_t}\big(1\!-\!k(\tau)\big) X_{\tau}+\int_t^{\tau}e^{-rs}{_{s-t} p_t}\mu(s)\max\{x_0e^{gs},X_s\}\ud s\\
&+\!\ind_{\{\tau = T \}}e^{-r\, T}{_{T-t} p_t}\max\{x_0e^{g T},X_{T}\}\Big|\cF_t\Big],
\end{aligned}
\end{align}
where we notice that the conditioning on $\{\tau_d>t\}$ is embedded into the survival probabilities ${_s p_t}$ that appear in the three terms under expectation. For any $s\in[t,T]$, we can write the account value as $X_s=X_t\exp[(r-c-\frac{\sigma^2}{2})(s-t)+\sigma (\WW_s-\WW_t)]$ and thanks to the Markovian structure we can write $V_t=v(t,X_t)$, where $v(t,x)$ is a measurable function representing the value of the policy at time $t$ when the account value is $X_t=x$ (an explicit representation for $v(t,x)$ is obtained in \eqref{eq:v(t,x)}).

Problem  \eqref{eq:Vt} can be made more tractable by changing measure. The $\Q$-martingale process 
$M:=(M_t)_{t\in[0,T]}$, given by
\begin{align}\label{Mt}
M_t:=e^{\sigma \WW_t-\frac{\sigma^2}{2}t}=e^{-(r-c)t}X_t/{x_0},
\end{align}
defines a probability measure $\P$, equivalent to $\Q$ on $\cF_T$, via $\ud\P=M_T\, \ud\Q$. We will denote by $\E$ the expectation with respect to $\P$.  By Girsanov's theorem, the process $\bar W:=(\bar W_t)_{t\in[0,T]}$ with
$\bar W_{t} := \WW_{t} - \sigma t$
is a $\P$-Brownian motion. Then, under the new measure $\P$, the dynamics of $X$ is given by $X_t=x_0\exp[(r-c+\frac{\sigma^2}{2})t+\sigma \bar W_t]$. For any $\bF$-stopping time $\tau\in[t,T]$ and any $\cF_\tau$-measurable random variable $\Psi$ it holds 
\begin{align}\label{eq:PQ}
\begin{aligned}
\E^\Q\big[e^{-r\tau}X_\tau\Psi\big|\cF_t\big]&=\E^\Q\big[e^{-c\tau}e^{-(r-c)\tau}X_\tau\Psi\big|\cF_t\big]=e^{-(r-c) t}X_t\E^\Q\Big[e^{-c\tau}\frac{e^{-(r-c)\tau}X_\tau}{e^{-(r-c)t}X_t}\Psi\Big|\cF_t\Big]\\
&=e^{-(r-c)t}X_t\frac{\E^\Q\big[M_T e^{-c\tau}\Psi\big|\cF_t\big]}{\E^\Q\big[M_T\big|\cF_t\big]}=e^{-(r-c) t}X_t\E\big[e^{-c\tau}\Psi\big|\cF_t\big],
\end{aligned}
\end{align}
where we used $M_\tau=\E^\Q[M_T|\cF_\tau]$ and the tower property in the third equality. Then, we can change measure in \eqref{eq:Vt} and obtain 
\begin{align}
\label{eq:VtP}
\begin{aligned}
e^{-rt}V_t= e^{-(r-c)t}X_t\esssup_{t\le \tau\le T}\E\Big[&\ind_{\{\tau<T\}}e^{-c \tau} {_{\tau-t} p_t}\big(1\!-\!k(\tau)\big)\\
&+\int_t^{\tau}e^{-c s}{_{s-t} p_t}\mu(s)\max\{x_0e^{gs}/X_s,1\}\ud s\\
&+\!\ind_{\{\tau = T \}}e^{-c\, T}{_{T-t} p_t}\max\{x_0e^{g T}/X_{T},1\}\Big|\cF_t\Big].
\end{aligned}
\end{align}
In order to further simplify the analysis, we set
$Z_{s} :=x_0 e^{g s}/ X_s$, so that $Z_0=1$ and
\begin{align}\label{eq:Z}
\ud Z_s =\alpha Z_s \ud s+\sigma Z_s\ud W_s,\quad \text{for $s\in[0,T]$},
\end{align} 
where $W_t:=-\bar W_t$ is again a Brownian motion under $\P$ and
\begin{equation}\label{eq:alpha}
	\alpha:=g+c-r.
\end{equation}
That yields
\begin{align}
\label{eq:VtPZ}
\begin{aligned}
e^{-ct}V_t= X_t\esssup_{t\le \tau\le T}\E&\Big[\ind_{\{\tau<T\}}e^{-c \tau} {_{\tau-t} p_t}\big(1\!-\!k(\tau)\big)+\int_t^{\tau}e^{-c s}{_{s-t} p_t}\mu(s)\max\{Z_s,1\}\ud s\\
&+\!\ind_{\{\tau = T \}}e^{-c\, T}{_{T-t} p_t}\max\{Z_{T},1\}\Big|\cF_t\Big]=:e^{-c t}X_t u(t,Z_t),
\end{aligned}
\end{align}
where we use the Markovian nature of the problem and time-homogeneity of the process $Z$ to define
\begin{align}
\begin{aligned}
\label{valueu}
u(t,z)=\sup_{0\le\tau\le T-t}\E\Big[& \int_{0}^{\tau}e^{-cs}  \mu(t+s) \,  _s p_{t}   \max\big\{Z^z_{s},1\big\} \ud s +\ind_{\{\tau< T-t \}} e^{-c \tau}\big(1-k(t+\tau)\big)\, _\tau p_{t}    \\
&+\ind_{\{\tau= T-t \}}e^{-c (T-t)}  \ _{T-t} p_{t} \max\big\{Z^z_{T-t},1\big\} \Big],
\end{aligned}
\end{align}
with $(t,z)\in [0,T]\times (0,\infty)$. Notice that in \eqref{valueu} we have performed a time-shift so that $(Z^z_s)_{s\in[0,T-t]}$ is the solution of \eqref{eq:Z} on $[0,T-t]$ starting from $Z_0=z$, and $\tau\in[0,T-t]$ is any stopping time for the filtration\footnote{The original Brownian motion $\widetilde W$ and the new one $W$ generate the same filtration $(\cF_s)_{s\in[0,T-t]}$.} $(\cF_s)_{s\in[0,T-t]}$. 

Recalling that $V_t=v(t,X_t)$ and $Z_t=x_0 e^{gt}/X_t$, in \eqref{eq:VtPZ} we have established 
\begin{align}\label{eq:uv}
v(t,X_t)=X_t u(t, x_0 e^{gt}/X_t),\quad \text{for $t\in[0,T]$}.
\end{align}
Given a pair $(t,x)$ and setting $z=x_0 e^{gt}/x$ the formula above reads 
\[
v(t,x)=x u(t, x_0 e^{gt}/x)\quad\text{or equivalently}\quad u(t,z)=\frac{v(t,x)}{x}. 
\]
Thus, by undoing the measure change and using the relationship between the processes $Z$ and $X$, we obtain the expression for $v(t,x)$ from \eqref{valueu} as 
\begin{align}\label{eq:v(t,x)}
	v(t,x)=\sup_{0\le \tau\le T-t}\EQ&\Big[\!\int_0^{\tau}e^{-r s}{_{s} p_t}\mu(t\!+\!s)\max\{x_0e^{g(t+s)},X^{x}_s\}\ud s\!+\!\ind_{\{\tau<T-t\}}e^{-r \tau} {_{\tau} p_t}\big(1\!-\!k(t\!+\!\tau)\big) X^{x}_{\tau}\nonumber \\
	&+\!\ind_{\{\tau = T-t \}}e^{-r\, (T-t)}{_{T-t} p_t}\max\{x_0 e^{g T},X^{x}_{T-t}\}\Big],
\end{align}
with $(X_s^{x})_{s\in[0,T-t]}$ the solution of the SDE \eqref{dynX0} with $X_0=x$.

We recover the VA's rational price at time zero by simply taking 
	\begin{equation}\label{eq:V_0u(0,1)}
		V_0=v(0,x_0)=x_0 u(0,1).
	\end{equation}
Such expression highlights how the VA's value is just proportional to the account's value at time zero. Later we will explore numerically conditions on the fee rate $c>0$ under which $V_0=x_0$ (i.e., $u(0,1)=1$). That corresponds to the fee rate that the seller of the VA requires from a buyer who pays a price equal to the initial account value. Such fee rate is calculated in order to compensate for the additional benefits in the contract.

\begin{remark}
The quantity $v(t,x)$ corresponds to the rational price at time $t$ of a VA with an initial account value (at time zero) equal to $x_0$, when the account value at time $t$ is equal to $x$. In other words, this is the amount of money that the policyholder should accept in exchange for the policy at time $t$ having observed a realised account value $X_t(\omega)=x$ (if such a trade was allowed and frictionless). 
\end{remark}

The next sections will be devoted to the analytical study of properties of $u(t,z)$ with the aim of characterising an optimal stopping time $\tau_*$ in \eqref{valueu}. It should be noted that for each initial condition $(t,z)$, we are going to produce an optimal stopping time $\tau_*=\tau_*(t,z)$ for the value function $u(t,z)$. Then, in order to recover the optimal surrender strategy for the policyholder in the VA problem we must compute $\tau_*(0,1)$ (recall \eqref{eq:V_0u(0,1)}).

Before moving to the next section, we point out that we will work mostly with a shifted value function $w(t,z):=u(t,z)-(1-k(t))$. We define the continuously differentiable function (recall Assumption \ref{A1})
\begin{align}\label{ftilde}
	f(t):=k(t) (c+\mu(t))- \dot k(t)-c, \qquad  t\in[0,T],
\end{align}
which will play an important role in the analysis of the problem. 
Notice that
\begin{align}\label{eq:kito}
\begin{aligned}
e^{-c \tau}\big(1\!-\!k(t\!+\!\tau)\big)\, _\tau p_{t}-(1\!-\!k(t))
&= \int_{0}^{\tau}e^{-cs}\! {_s p_{t}}\Big( f(t\!+\!s)\! -\!\mu(t\!+\!s)  \Big) \ud s.
\end{aligned}
\end{align}
Moreover, using $k(T)=0$,
\[
e^{-c \tau}\big(1-k(t+\tau)\big)\ind_{\{\tau<T-t\}}\,_\tau p_t=e^{-c \tau}\big(1-k(t+\tau)\big)\,_\tau p_t-\ind_{\{\tau=T-t\}}e^{-c (T-t)}\,_{T-t}p_t
\] 
and so, from \eqref{valueu}, we obtain
\begin{align}
\label{valuew}
\begin{aligned}
w(t,z)=\sup_{0\le\tau\le T-t}\E\Big[&\int_{0}^{\tau}e^{-cs} \,  _s p_{t}  \Big(f(t+s)+ \mu(t+s) (Z^z_s-1)^{+}  \Big)  \ud s\\
&+\ind_{\{\tau= T-t \}}e^{-c (T-t)}  \ _{T-t} p_{t} (Z^z_{T-t}-1)^{+} \Big],
\end{aligned}
\end{align}
where $(a)^+=\max\{a,0\}$ is the positive part. Finally, in order to simplify the notation we set 
\begin{align}\label{setP}
\cP_{t,z}(\tau):=\int_{0}^{\tau}\!\!e^{-cs}  {_s p_{t}}  \Big(f(t\!+\!s)\!+\! \mu(t\!+\!s) \big(Z^z_s\!-\!1\big)^{+}  \Big)  \ud s
\!+\!\ind_{\{\tau= T-t \}}e^{-c (T-t)} {_{T-t} p_{t}} \big(Z^z_{T-t}\!-\!1\big)^{+},
\end{align}
so that $w(t,z)=\sup_{0\le \tau\le T-t }\E[\cP_{t,z}(\tau)]$.

Let us introduce 
\begin{equation}\label{eq:t*}
	t^*:=\inf\{t\in[0,T): f(t)<0 \}\wedge T,\quad\text{with $\inf\varnothing=\infty$},
\end{equation}
and let us assume the following properties.

\begin{assumption}\label{A2}
The function $f$ changes its sign at most once in $[0,T]$.  If it changes its sign, then $f(t)\geq 0$ for every $t\in[0,t^*]$ and $f(t)<0$ for every $t\in(t^*,T]$ (i.e., $f$ goes from positive to negative). 
Finally, one of the two conditions below must hold
\begin{enumerate}[(i)]
\item $\dot \mu(t)\le 0$ and $\dot f(t)\le 0$ for $t\in(t^*,T]$;
\item $\dot \mu(t)>0$ and $\dot f(t)-\mu(t)f(t)\le 0$ for $t\in(t^*,T]$.
\end{enumerate}
\end{assumption}

The conditions required by Assumption \ref{A2} are of a technical nature but, as we explain in Remark \ref{rem:impact}, they are fully consistent with real-world mortality data (see also Section \ref{Sec:NumEvid}). In particular, (i) and (ii) will be used to prove our key Lemma \ref{lem:b'_delta}. 
The requirement that $f$ only changes its sign once, and it goes from positive to negative values, is consistent with the idea that the interplay of mortality risk and financial penalty will incentivise the policyholder to delay the surrender at least until time $t^*$ (mathematically this is expressed by Lemma \ref{lem:stb}; cf.\ also \eqref{eq:sC}). As explained in Section \ref{sec:t^*} and shown in Figure \ref{fig1}, 
when surrender charges are sufficiently high, early surrender is not optimal at any time  $t<t^*$.
It is perhaps worth noticing that in the absence of a penalty for early surrender (i.e., if $k(t)\equiv 0$) we have $f(t)\equiv-c$ and, if $c>0$, $\dot f(t)-f(t)\mu(t)> 0$ for every $t\in[0,T]$. Then Assumption \ref{A2}--(i) requires $\dot \mu(t)\le 0$ for all $t\in[0,T]$, which does not seem very realistic. However, the problem without surrender charges is not relevant here and it should be treated separately.

\begin{remark}\label{rem:impact}
In order to test our assumptions in real-world situations, we collected data from the Life Table of the resident population in Italy for the year 2024 provided by ISTAT (National Institute for Statistics) \cite{istat}. We constructed from the data the empirical version of the function $\mu(t)$; cf.\ Figure \ref{figIstat}, blue solid line. We computed the functions $f(t)$ and $\dot f(t)-\mu(t)f(t)$ using a penalty function of the form $k(t)=1-\exp(-K(T-t))$, with $K=1.4\%$ per year and an annual fee rate $c=2.5\%$; cf.\ Figure \ref{figIstat}, red dashed and green dashed-dotted lines, respectively. 
As illustrated in Figure \ref{figIstat}, we observe $t^*=0$ and our Assumption \ref{A2}--(ii) is fully satisfied.

The empirical analysis shows
that in our study the predominant source of time-variability is given by the penalty function $k(\cdot)$. 
Indeed, while values of $k(t)$ and $\dot k(t)$ are in the range of $10^{-1}$ and $10^{-2}$, respectively, the range of values of $\mu(t)$ and $\dot \mu(t)$ are closer to $10^{-3}$ and $10^{-4}$, respectively.
Therefore, even choosing a constant force of mortality $\mu(t)\equiv\mu>0$ in our model would not be a very stringent requirement. In that case, Assumption \ref{A2} is satisfied for any $k(\cdot)$ such that $\dot f(t)\le 0$ for $t\in(t^*,T]$, i.e., such that $(c+\mu)\dot k(t)\le \ddot k(t)$ for $t\in(t^*,T]$. For the benchmark penalty $k(t)=1-\exp(-K (T-t))$ (cf.\ \eqref{kt}), that corresponds to $K\le c+\mu$. The latter induces no loss of generality because we will see in Remark \ref{remarkt*} that early surrender only occurs for $K<c$. 
\end{remark}

\begin{figure}[ht!]
\centering
\includegraphics[scale=0.35]{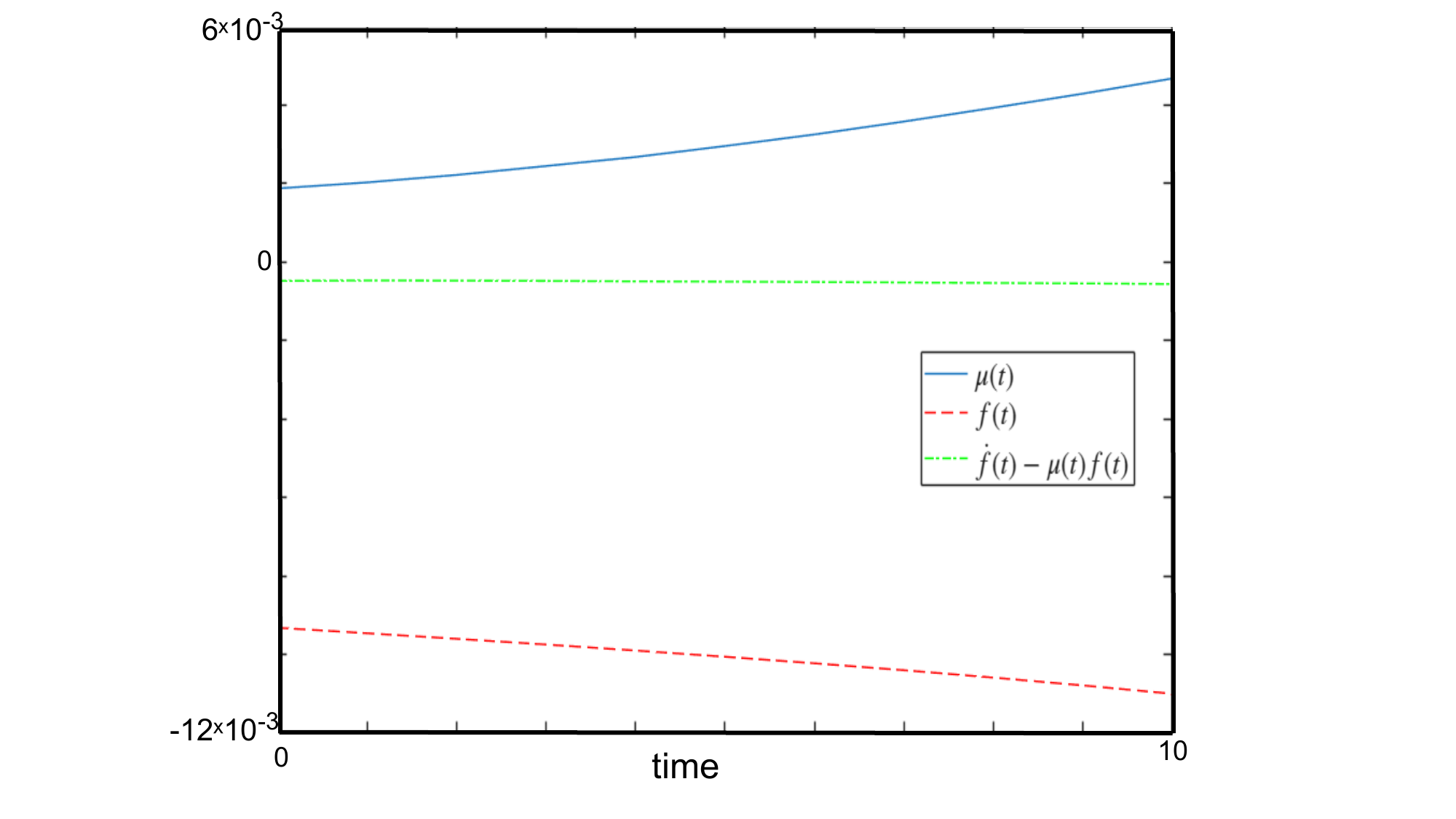}
\caption{The figure shows the empirical force of mortality $\mu(t)$ constructed from the Life Table of the resident population in Italy for the year 2024 provided by ISTAT (blue solid line), the function $f(t)$ (red dashed line) and  the function $\dot f(t)-\mu(t)f(t)$ (green dotted-dashed line).}
\label{figIstat}
\end{figure}
\subsection{Pricing formula and optimal surrender boundary}

The pricing formula and a characterisation of the optimal surrender rule are illustrated in the next theorem. These two results refer to a VA which is purchased at time zero with an initial account value equal to $x_0$. Before stating the theorem, we recall the definition of $t^*\in[0,T]$ in \eqref{eq:t*} and we introduce a special class of functions:
\begin{align}\label{eq:Sigma}
\Sigma:=\left\{\psi:[0,T]\to[0,\infty)\left|
\begin{array}{l}
\text{$0\le \psi(t)\le (1-f(t)/\mu(t))^+$ for $t\in[0,T]$};\\
\text{$\exists\, t_\psi\in[t^*,T)$\, s.t.\ $\psi(t)=0$ for $t\in[0,t_\psi)$}\\
\text{and $\psi\in C([t_\psi,T])$ with $\psi(T)=1$}
\end{array} 
\right.
\right\},
\end{align}
We adopt the following conventions: if $t^*= T$, then $\Sigma$ contains only the function $\psi(t)=\ind_{\{T\}}(t)$; 
for $t^*\in[0,T)$ we have $\Sigma\neq\varnothing$, because $f(t)< 0$ and $\mu(t)>0$ for every $t\in(t^*,T]$; finally,
if $f<0$ on $[0,T]$ (i.e., $t^*=0$), we allow for $t_\psi<0$ in the sense that $\psi(0)$ does not have to be zero.
For $\psi\in\Sigma$ we denote its support by $\Gamma_{\psi}:=\{t\in[0,T): \psi(t)>0\}\subseteq[t_\psi, T]$. 

\begin{theorem}\label{thm:main}
For each $x_0\in(0,\infty)$ given and fixed, there is a function $\ell_{x_0}:[0,T]\to [0,\infty]$, where  $\ell_{x_0}(t)=+\infty$ for $t\in[0,t^*)$ and\footnote{Here we understand continuity of $\ell_{x_0}$ at $t^*$ in the sense that $\ell_{x_0}(t)<\infty$ for all $t\in(t^*,T]$, $\ell_{x_0}\in C((t^*,T])$ and $\ell_{x_0}(t^*)=\lim_{t\downarrow t^*}\ell_{x_0}(t)\in[0,+\infty]$.} $\ell_{x_0}\in C([t^*,T])$ with $\ell_{x_0}(T)=x_0 e^{gT}$, such that the optimal surrender time reads 
\[
\gamma_*=\inf\{s\in[0,T):X^{x_0}_s\ge \ell_{x_0}(s)\}\wedge T.
\] 
The VA's rational price (at time zero) reads
\begin{align}\label{eq:V0int}
\begin{aligned}
V_0&= e^{-rT}{_T p_0}\E^\Q\big[\mathrm{max}\big\{x_0e^{gT},X^{x_0}_T\big\}\big]+\int_0^{T}e^{-rs}{_s p_0}(\mu(s)-f(s))\E^\Q\Big[ X^{x_0}_s\ind_{\{\ell_{x_0}(s)\le  X^{x_0}_s\}}\Big]\ud s\\
&\quad +\int_0^{T}e^{-rs}{_s p_0}\mu(s)\E^\Q\Big[\mathrm{max}\big\{x_0e^{gs},X^{x_0}_s\big\}\ind_{\{\ell_{x_0}(s)> X^{x_0}_s\}}\Big]\ud s.
\end{aligned}
\end{align} 

If $t^*\in[0,T)$ the optimal surrender boundary $\ell_{x_0}$, is obtained as $\ell_{x_0}(t):=x_0 e^{gt}/b(t)$ for $t\in(t^*,T]$, where $b$ is the unique function from the class $\Sigma$ that solves the integral equation
\begin{align}\label{eq:ellint}
\begin{aligned}
&e^{-c(T-t)}{_{T-t} p_t}\int_1^\infty \!\Phi\big(b(t),T\!-\!t,y\big)\big(y\!-\!1\big)\ud y\\
&+\!\int_t^{T}\!\!e^{-c(s-t)}{_{s-t} p_t}\bigg(\int_{b(s)}^\infty \big(\mu(s)\big(y-1\big)^++f(s)\big)\Phi\big(b(t),s-t,y\big)\ud y\bigg)\,\ud s=0,
\end{aligned}
\end{align}
for all $t\in\Gamma_b$, where
\begin{align}\label{eq:Phi}
\Phi(z,s,y)=\frac{1}{\sqrt{2\pi\sigma^2 s}\,y}\exp\bigg(-\frac{\big(\ln\frac{y}{z}-\big(\alpha-\frac{\sigma^2}{2}\big)s\big)^2}{2\sigma^2 s}\bigg).
\end{align}
Moreover, the closure of $\Gamma_b$ coincides\footnote{When $t^*\in(0,T)$, whether $\Gamma_b=(t^*,T]$ or $\Gamma_b=[t^*,T]$ depends on whether $b\in C([0,T])$ or not (hence, whether $b(t^*)=0$ or not). Numerically, we can only observe $\Gamma_b=[t^*,T]$, i.e., $b(t^*)>0$.} with $[t^*,T]$.
\end{theorem}
A formal proof of the theorem is provided at the end of Section \ref{sect:bCont} after a long and delicate analysis of the problem, detailed in the next sections. In order to facilitate the read, here we provide an outline of the proof in which we highlight the key steps and their role towards the full proof of Theorem \ref{thm:main}. Most of the analysis is performed on the function $w$ from \eqref{valuew} and the associated optimal stopping boundary $b$. Only in the final step do we return to the formulation in terms of the VA price $V_0$ with boundary $\ell_{x_0}$.
\smallskip

{\em Step 1}. The first part of the analysis is aimed at showing continuity of the function $w$, which is obtained by combining Lemma \ref{lem:w(z)Lip} (spatial continuity) and Lemma \ref{lem:w(t)Cont} (time continuity). Continuity is needed to rigorously formulate the free-boundary problem \eqref{eq:fbp} and to establish optimality of the stopping time $\tau^*$ (cf.\ \eqref{eq:tau*1}). Lemma \ref{lem:stb} and Lemma \ref{lem:wincrconv} are ancillary to this analysis and clarify the impact of $t^*$ on the form of the function $w$. Lemma \ref{lem:wincrconv} also yields monotonicity and convexity of $z\mapsto w(t,z)$, which in turn guarantee existence of an optimal stopping boundary $t\mapsto b(t)$.
\smallskip

{\em Step 2}. The first result about the optimal boundary is contained in Proposition \ref{prop:limsupT}, which establishes the behaviour of $b$ near the maturity of the contract. Since we do not yet know that $b$ is continuous, the statement in Proposition \ref{prop:limsupT} is given in terms of `$\limsup$'. This is however sufficient to prove in Corollary \ref{lem:Z>1} that if the contract is not surrendered until maturity, then the terminal value of the process $Z$ must be greater than one. Such seemingly irrelevant detail is actually used in the proofs of key Lemmas \ref{lem:w_z} and \ref{lem:b'_delta}.
\smallskip

{\em Step 3}. The lack of monotonicity of $t\mapsto w(t,z)$ points us to a non-monotone behaviour of $t\mapsto b(t)$. Our next step is to find a lower bound on the `rate' of decrease of $b$: we first derive a probabilistic expression for $w_z(t,z)$ in Lemma \ref{lem:w_z} and then an upper bound on $w_t(t,z)$ in Lemma \ref{lem:w_tUpBound}. Combining the latter two, we obtain the desired lower bound on the `rate' of decrease of $b$ (cf.\ Lemma \ref{lem:b'_delta}), which allows us to rewrite the optimal stopping problem in a new parametrisation with a non-decreasing optimal boundary $\beta$. This new parametrisation has two main advantages: (i) we can prove suitable `continuity' of optimal stopping times with respect to the initial point of the optimization (Lemma \ref{lem:tau_nTo0}) thanks to regularity in the sense of diffusions of the stopping set (Lemma \ref{lem:tau*=tau'}); (ii) we can prove continuity of the boundary $t\mapsto \beta(t)$ (Proposition \ref{prop:betalc}), which then translates easily into continuity of the original boundary $t\mapsto b(t)$ (Theorem \ref{thm:bCont}).
\smallskip

{\em Step 4}. We focus on obtaining maximal regularity for the value function $w$. Lemma \ref{lem:EuAmCall} is a technical result which we need to obtain a probabilistic expression for a lower bound for $w_t(t,z)$ in Lemma \ref{lem:w_tLowBound}. Thanks to the above mentioned continuity of stopping times, we use the formula for $w_z$ in Lemma \ref{lem:w_z} and the bounds for $w_t$ in Lemmas \ref{lem:w_tLowBound} and \ref{lem:w_tUpBound} to conclude that $w$ is continuously differentiable in the whole domain $[0,T)\times(0,\infty)$ (Proposition \ref{prop:wC1}). Next, using regularity of $w$ and a generalisation of It\^o's formula we are able to obtain the integral equation for $b$ and a formula for $w(t,z)$ (Proposition \ref{prop:inteq}). Finally, by undoing the transformations that lead from $V_0$ to $w(t,z)$ we obtain the proof of Theorem \ref{thm:main}.
\smallskip

Notice that continuity of the function $b$ is equivalent to continuity of the optimal surrender boundary $\ell_{x_0}$ and it is important both to obtain uniqueness for the solution of the integral equation \eqref{eq:ellint} and for an efficient numerical resolution thereof. We also notice that in the integral equation \eqref{eq:ellint} the function $\Phi$ is the transition density of the process $Z$ under the measure $\P$, which points us to the key role played by the change of measure to obtain our integral equation.

\subsection{Comments on an infinite-time horizon problem}
In this section we outline an adaptation of our framework to an infinite-time horizon problem, i.e., when $T=\infty$. While a complete analysis is outside the scope of the present paper, we point the reader to the main ingredients and to the results that we expect can be obtained, provided that suitable adjustments to the proofs are made.

The analysis can be developed along the same lines as in Section \ref{subs:markovian}, up to imposing suitable conditions at $T=\infty$. In particular, in the problem formulation we should assume that the force of mortality admits a well-defined limit $\mu_\infty:=\lim_{t\to\infty}\mu(t)\in(0,\infty)$, we should replace $k(T)=0$ with $\lim_{T\to\infty} k(T)=0$ and we also need to assume $r\ge g$ in order to guarantee finiteness of the payoff on the event $\{\tau\ge \tau_d\}$. The initial problem formulation then reduces to 
\[
V_0= \sup_{\tau\ge 0}\EEQ\Big[\ind_{\{\tau<\tau_d \}}e^{-r \tau}\big(1-k(\tau)\big) X_{\tau}+\ind_{\{\tau \ge \tau_d \}}e^{-r\, \tau_d }\max\{x_0e^{g \tau_d},X_{\tau_d}\}\Big].
\]
It is also worth noticing that for a positive fee rate $c>0$ we have $\lim_{t\to\infty}e^{-rt}X_t=0$, $\Q$-a.s., and the family $\{e^{-rt}X_t,\,t\ge 0\}$ is uniformly integrable\footnote{Notice that for any $p>1$, $\E^\Q[e^{-rpt}X^p_t]=x_0\exp([(p-1)\sigma^2/2-c]pt)$; thus for sufficiently small $p>1$ the exponent is negative and the family is $L^p$-bounded, hence uniformly integrable.}.

Repeating steps as in Section \ref{subs:markovian} and taking care of subtleties about the change of measure on infinite horizon, the value function $w$ from \eqref{valuew} reduces to
	\begin{equation}\label{eq:wInfHor}
	w(t,z)=\sup_{\tau\geq 0}\E\Big[\int_{0}^{\tau}e^{-cs} \,  _s p_{t}  \Big(f(t+s)+ \mu(t+s) (Z^z_s-1)^{+}  \Big)  \ud s
	\Big].
	\end{equation}
	This is still a time-inhomogeneous optimal stopping problem, since $f$ and $\mu$ depend on the starting time $t\in[0,\infty)$, but it is in some respect simpler than its finite-time horizon version. For example, the reward process in \eqref{setP} has a discontinuity at the terminal time $T$, whereas such discontinuity is removed in \eqref{eq:wInfHor}.
		
Under mild conditions on $f$ that guarantee the legitimate use of convergence theorems (e.g., $f$ bounded and with $f_\infty:=\lim_{t\to\infty}f(t)\in(-\infty,0)$), one can reproduce results like those in Sections \ref{sec:contb}--\ref{sect:bCont} below. A version of Theorem \ref{thm:main} should therefore be available. In particular, the VA's rational price for the infinite-time horizon problem reduces to (cf.\ \eqref{eq:V0int})
	\begin{align*}
		\begin{aligned}
			V_0&= \int_0^{\infty}e^{-rs}{_s p_0}(\mu(s)-f(s))\E^\Q\Big[ X^{x_0}_s\ind_{\{\ell_{x_0}(s)\le  X^{x_0}_s\}}\Big]\ud s\\
			&\quad +\int_0^{\infty}e^{-rs}{_s p_0}\mu(s)\E^\Q\Big[\mathrm{max}\big\{x_0e^{gs},X^{x_0}_s\big\}\ind_{\{\ell_{x_0}(s)> X^{x_0}_s\}}\Big]\ud s,
		\end{aligned}
	\end{align*}
where $\ell_{x_0}(t)=+\infty$ for $t\in[0,t^*)$, $\ell_{x_0}(t):=x_0e^{gt}/b(t)$ for $t\in(t^*,\infty)$. The optimal boundary $b$ solves the integral equation
\[
\int_t^{\infty}\!\!e^{-c(s-t)}{_{s-t} p_t}\bigg(\int_{b(s)}^\infty \big(\mu(s)\big(y-1\big)^++f(s)\big)\Phi\big(b(t),s-t,y\big)\ud y\bigg)\,\ud s=0,\]
with a terminal condition $b_\infty:=\lim_{t\to\infty}b(t)$ that needs to be determined as part of the analysis. The natural conjecture is that $b_\infty$ should be obtained as the optimal boundary of the time-homogeneous problem obtained asymptotically for $t\to\infty$. That is, we expect $b_\infty$ to be the optimal boundary for the problem
\[
w(z):=\sup_{\tau\geq 0}\E\Big[\int_{0}^{\tau}e^{-cs} \,  e^{-\mu_\infty s}\Big(f_\infty+ \mu_\infty (Z^z_s-1)^{+}  \Big)  \ud s
	\Big],
\]
but this requires fuller justification and we prefer to leave it to the interested reader.

We finally notice that uniqueness of the solution for the integral equation is also more delicate in the infinite-time horizon framework and it would need to be addressed in detail.

\section{Sensitivity analysis and economic interpretation}\label{Sec:NumEvid}

Before delving into the technical details of the proof of Theorem \ref{thm:main}, we illustrate the economic conclusions that can be drawn from our results. We perform a detailed sensitivity analysis of the VA's price and of the optimal surrender boundary with respect to various model parameters. We also find numerically the so-called {\em fair fee}, which makes $V_0=x_0$. As a first step, we introduce the numerical scheme used to calculate the optimal boundary.

\subsection{A numerical scheme for the boundary}
As stated in Theorem \ref{thm:main}, the optimal surrender boundary is given by $\ell_{x_0}(t)=x_0 e^{gt}/b(t)$, where $b$ is the unique solution of the integral equation \eqref{eq:ellint}. The numerical solution of this integral equation follows an algorithm based on a Picard iteration scheme that we learned from \cite{yerkin} and that has been used elsewhere in the optimal stopping literature (see, e.g., \cite[Sec.\ 6]{de2020optimal}). 

We take an equally-spaced partition $0=t_0<t_1<\ldots<t_{n-1}<t_n=T$ with $\Delta t:=t_{i+1}-t_i=T/n$. The Picard iteration is initialised with $b^ {0}(t_j)=\frac{1}{2}\big(1+(\frac{t_j}{T})^2\big)$ for $j=0,1,\ldots n$ and we notice that $b^0(T)=1$ in keeping with the theoretical result $b(T)=1$ from Theorem \ref{thm:bCont}. The particular choice of $b^0$ is motivated by numerical efficiency. Other specifications are possible but we observe slower convergence rates of the Picard scheme (always to the same limit). Let $b^ {k}(t_j)$, for $j=0,1,\ldots n$, be the values of the boundary obtained after the $k$-th iteration. Then, the value $b^ {k+1}(t_j)$ for the $(k+1)$-th iteration is computed as the value $0\le \theta\le (1-f(t_j)/\mu(t_j))^+$ that solves the following discretised version of equation \eqref{eq:ellint}: for $j=0,1,\ldots, n-1$
\begin{align*}
e^{-c(T-t_j)}& {_{T-t_j}}p_{t_j}\int_1^\infty \Phi(\theta,T-t_j,y)\ud y\\
&+\Delta t \sum_{i=j+1}^{n}e^{-c(t_i-t_j)} {_{t_i-t_j}}p_{t_j} \int_{b^ {k}(t_i)}^\infty \Phi(b^k(t_i),t_i-t_j,y) H(t_i,y)\ud y=0,
\end{align*}
where we recall the log-normal density $y\mapsto\Phi(z,s,y)$ from \eqref{eq:Phi} and $H(t,x)=\mu(t)(x-1)^++f(t)$ (cf.\ \eqref{H}).
Given a tolerance parameter $\eps>0$, the algorithm stops when 
\[
\max_{j=0,1,\ldots, n}| b^ {k}(t_j)-b^ {k+1}(t_j)|<\varepsilon.
\]
We chose $\eps=10^{-2}$. 

\subsection{A model for the force of mortality}

For the numerical illustrations in this section we use the Gompertz-Makeham mortality model. That is, for the force of mortality we assume 
\begin{equation}
\label{force}
\mu(t)=A+BC^t,
\end{equation}
with $A=0.0001$, $B=0.00035$ and $C=1.075$. These parameter values are taken from  \cite{MacKay} and they yield a life expectancy of $21.7$ years for a $50$ years old individual.
We consider the exponential penalty function from \eqref{kt}.
The benchmark parameters in what follows are: 
\begin{align}\label{eq:benchmk}
\begin{aligned}
\left\{
\begin{array}{c}
\text{$T=10$ years, $\eta=50$ years, surrender charge intensity $K=1.4\%$},\\
\text{annual risk-free rate $r=5\%$, annual volatility $\sigma=20\%$},\\ 
\text{annual fee rate $c=2.5\%$, minimum guaranteed $g=0\%$, initial fund value $x_0=100$}.
\end{array}
\right.
\end{aligned}
\end{align}
All parameters are expressed per annum. When different parameter choices are made, we state it explicitly (e.g., in Table \ref{tab2}). We emphasise that in all our numerical experiments Assumptions \ref{A1} and \ref{A2} are satisfied. 

The particular choice of the Gompertz-Makeham model is not crucial for the main conclusions of this section to hold. Indeed, Figure \ref{figweibullBeard} 
compares the optimal surrender boundary obtained under the  Gompertz-Makeham mortality model with those obtained under the Weibull and Perks models (see, e.g., \cite[Sec.\ 2.5.1]{pitacco}), given respectively by
\[
\mu(t)=\frac{\alpha}{\beta}\left( \frac{t}{\beta} \right)^{\alpha-1}\quad\text{and}\quad \mu(t)=\frac{a  e^{bt}+d}{c e^{b}+1}.
\]
Parameters are estimated so that a 50-year-old individual has a life expectancy of 21.7 years, as in the Gompertz-Makeham model. In particular, we set $\alpha=2.97$, 
$\beta=66.4$ in the Weibull  model and $a=0.000136$, $b=0.10$, $c=0.001405$, $d=0$ in the Perks model.

\begin{figure}[hbt]
\centering
\includegraphics[scale=0.3]{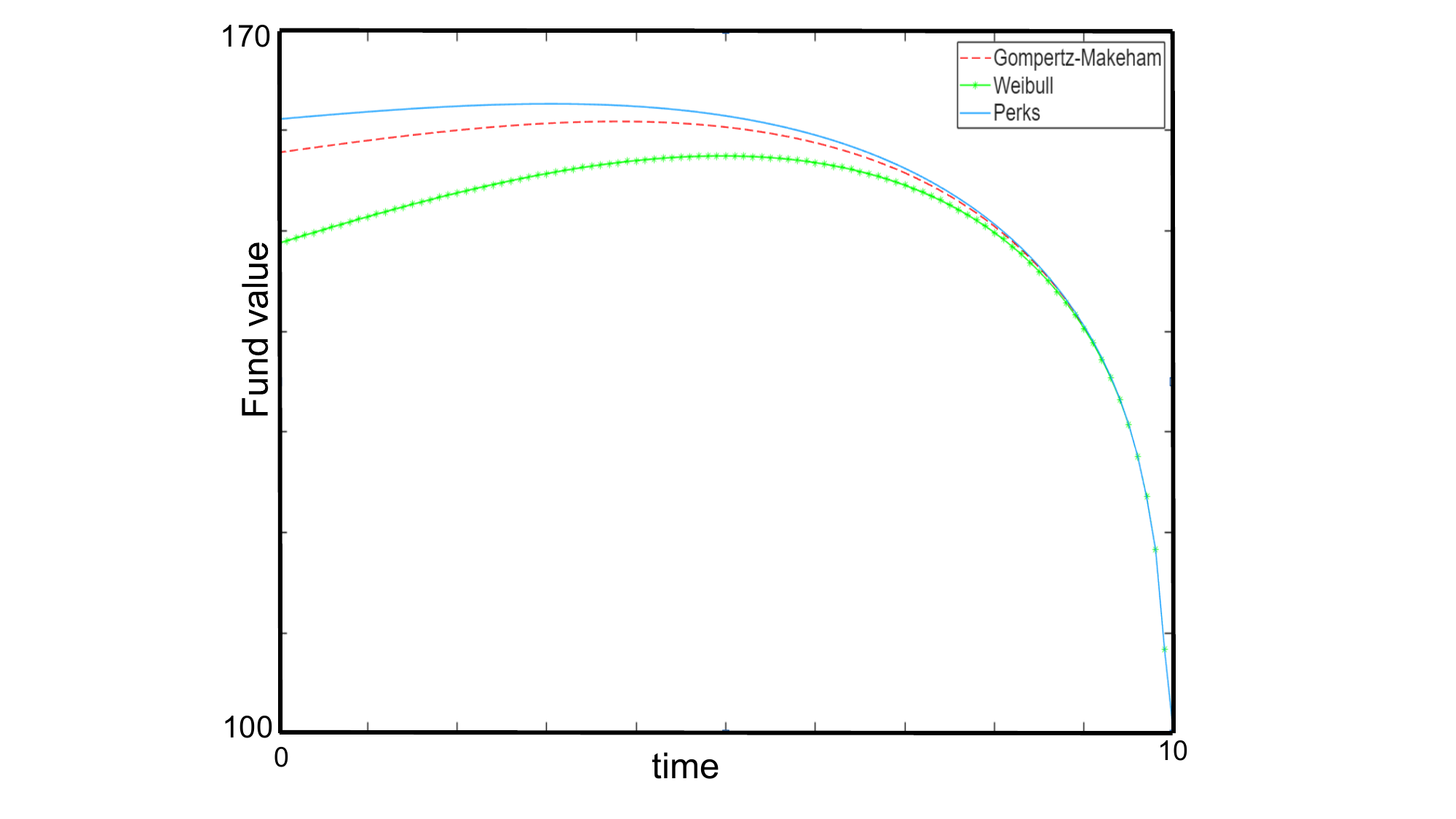}
\caption{The optimal surrender boundary $t\mapsto \ell_{x_0}(t)$ for the Gompertz-Makeham (dashed red line), Weibull (starred green line) and Perks (solid blue line) mortality models.}
\label{figweibullBeard}
\end{figure}
It is observed that as time increases, the three curves get closer to each other. eventually converging to the value of 100 at the contract's maturity, as shown in Theorem 4.5. From a qualitative standpoint, the three curves exhibit substantially the same non-monotonic behavior, suggesting overall robustness of the qualitative results with respect to the exact specification of the mortality model (the same conclusion is also reached in \cite{Hilpert}, p.\ 681). This conclusion can be explained in light of Remark \ref{rem:impact}, where an empirical analysis shows that the main driver in the policyholder's decision to surrender the contract is the surrender charge structure, rather than the mortality specification.
As a result, the sensitivity analysis that we perform here for the Gompertz-Makeham model, yields the same conclusions when performed for the Weibull or Perks models. In the paper, we only focus on the former model to avoid repetitions.
\begin{figure}[ht!]
\centering
\includegraphics[scale=0.5]{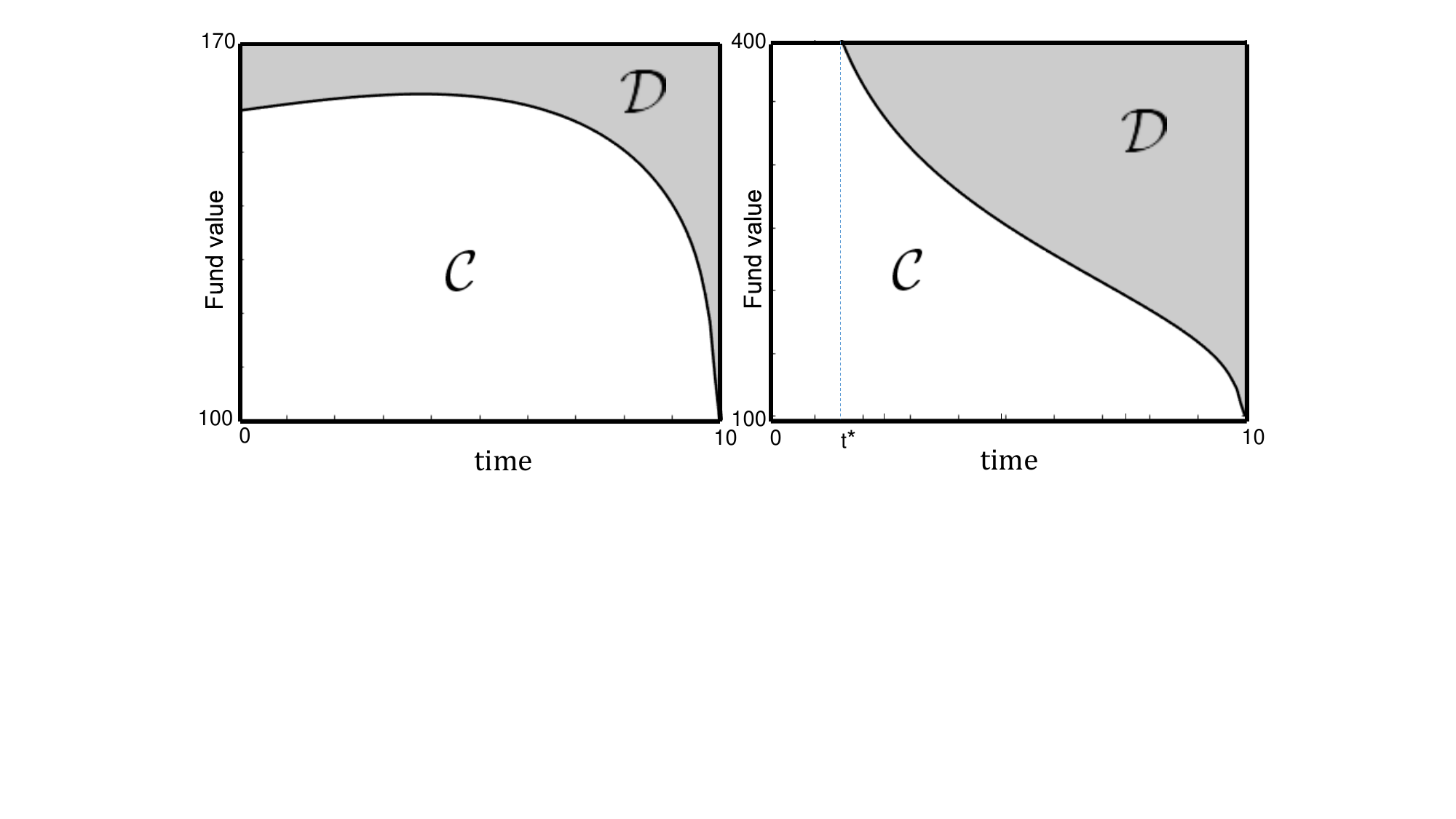}
\vspace{-4.5cm}
\caption{The optimal surrender boundary $t\mapsto \ell_{x_0}(t)$ for $K=1.4\%$ (left plot) and for $K=2.2\%$ (right plot). Notice that in the left plot $t^*=0$ and $b(t^*)>0$; instead, in the right plot $t^*>0$ and $b(t^*)>0$.}
\label{fig1}
\end{figure}
\subsection{Sensitivity to the surrender charge intensity}\label{sec:t^*}
Figure \ref{fig1} shows the optimal surrender boundary  $\ell_{x_0}(\cdot)$ with the associated continuation region $\cC$ and stopping region $\mathcal{D}$ (where it is optimal to continue and to surrender, respectively) for  $K=1.4\%$ (left plot) and $K=2.2\%$ (right plot). In case $K=1.4\%$ (left plot) we observe $t^*=0$, whereas in case $K=2.2\%$ (right plot) we have $t^*=1.5$. In both cases, Assumption \ref{A2} \textit{(ii)} is satisfied.
For $K=1.4\%$, we observe that in the initial years of the contract surrender only occurs for large VA's account values. Approaching the maturity of the contract, surrender is optimal also for VA's account values closer to $x_0$ and, eventually, $\ell_{x_0}(T)=x_0$  (recall that here $g=0$). 
As previously anticipated, this is an example of a non-monotonic surrender boundary. For larger surrender charge intensity ($K=2.2\%$, right plot), the surrender option (SO) is never exercised for $t\in[0,t^*)$, whatever the fund's value.  This shows that large surrender charges in the early years of the contract totally disincentivise surrender. As time approaches the maturity, we observe a decreasing surrender boundary that eventually converges to $\ell_{x_0}(T)=x_0$. As shown in Remark \ref{remarkt*} the contract is never surrendered if $K\ge c$. In particular, $K=c$ is the minimum surrender charge intensity that totally removes the incentive to an early surrender, i.e., it yields $t^*=T$.

\subsection{Comparing effects of surrender charge and mortality}
In Figure \ref{fig3} we study the sensitivity of the optimal surrender boundary $\ell_{x_0}(\cdot)$ to the surrender charge intensity (left plot) and to the mortality specification (right plot). The boundary increases monotonically and rather rapidly with the surrender charge intensity. That shows a significant impact of the surrender charge on the shape of $\ell_{x_0}(\cdot)$ and on the decision of the policyholder to surrender the contract.
For the sensitivity analysis with respect to the mortality specification, we use the so-called \textit{proportional hazard rate transformation}, introduced in actuarial science by Wang \cite{Wang} (see also \cite{MY}). Given a mortality force $t\mapsto \mu(t)$ representative of the whole population, the transformation is defined as $(1+\bar{\mu})\mu(\cdot)$, with $\bar{\mu} \in (-1,+\infty)$. Roughly speaking, if $\bar{\mu}<0$ (resp.\ $\bar \mu>0$) the individual considers herself healthier (resp.\ unhealthier) than the average in the population.
In this way, we test the sensitivity of $\ell_{x_0}(\cdot)$ to perturbations of the mortality force specified in \eqref {force}. Notice that for $\bar{\mu} \in \{ 0, -0.38, 0.38 \}$ the life expectancy of a $50$ years old individual is respectively $21.7$, $26.6$ and $18.5$ years. We observe that the optimal surrender boundary exhibits no substantial change in its behaviour for the considered values of $\bar{\mu}$. More precisely, moving from $\bar{\mu}=0$ to $\bar{\mu}=\pm 0.38$ causes a maximal shift for  $\ell_{x_0}(\cdot)$ of less than $3\%$. Therefore, we conclude that a change in the mortality specification has a small impact on the optimal surrender boundary. In summary, Figure \ref{fig3} shows that the main driver in the policyholder's decision to surrender the contract is the surrender charge structure, rather than demographic considerations (cf.\ Remark \ref{rem:impact}).

\begin{figure}[ht!]
\centering
\includegraphics[scale=0.5]{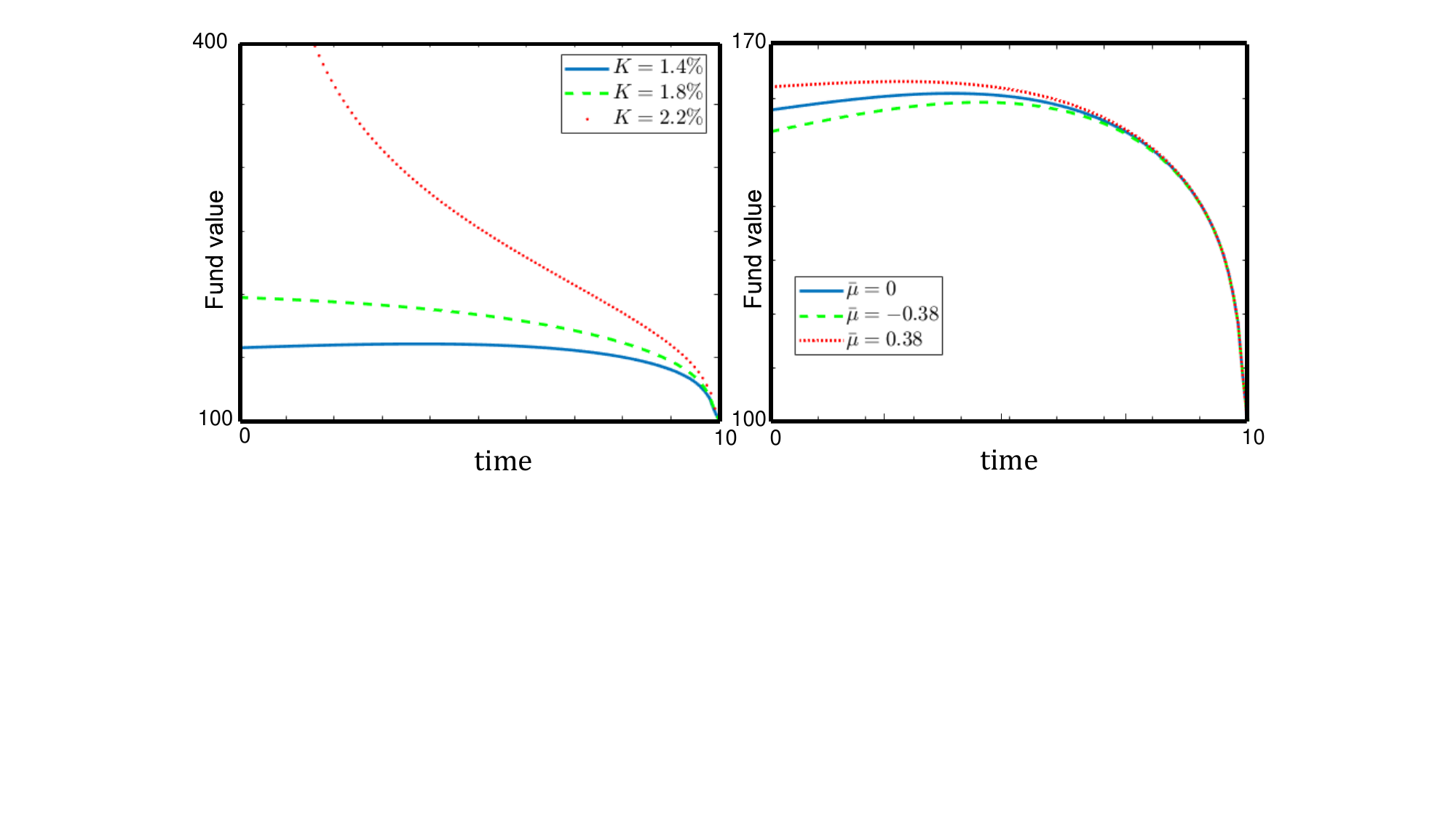}
\vspace{-4.5cm}
\caption{Sensitivity of the optimal surrender boundary $t\mapsto \ell_{x_0}(t)$ with respect to $K$ (left plot) and to $\bar{\mu}$ (right plot).}
\label{fig3}
\end{figure}

\subsection{Sensitivity to fee rate and minimum guaranteed}
In Figure \ref{fig4} we study the sensitivity of the optimal surrender boundary $\ell_{x_0}(\cdot)$ to the constant fee rate $c$ (left plot) and to the minimum guaranteed rate $g$ (right plot).  On the one hand, when the constant fee rate $c$ increases, the policyholder pays a progressively greater percentage of the account value to the insurer. Therefore, the incentive to surrender increases, the stopping region expands and the optimal boundary is pushed downward in the plot. On the other hand, as the guaranteed roll-up rate $g$ increases, the policyholder progressively obtains more benefits from staying in the contract and then the continuation region expands. As a result the optimal boundary $\ell_{x_0}(\cdot)$ is pushed upward in our plot. Both the fee rate and the minimum guaranteed rate have a significant impact on the positioning of the optimal boundary, hence on the policyholder's surrender strategy. We observe that the fee rate also has a strong influence on the curvature of the optimal boundary. In particular, we observe stronger non-monotone behaviour of the boundary for larger values of the fee rate. This is in line with the intuition that, starting at time zero, if the policyholder does not surrender in the very early years (e.g., because of steep surrender charges), the longer she waits the longer she pays a fee that accumulates over time at the rate $c>0$. Large values of $c>0$ make surrender even less appealing as time passes, because the policyholder would like to recover the initial loss incurred with the fee payments. Thus, the boundary increases at the beginning of the contract.  In later years, the effect of the fee payments is more than compensated for by a reducing surrender charge and the approaching maturity of the contract. Those countering effects increase the appeal of the surrender option and invert the monotonicity of the optimal boundary.

\begin{figure}[ht!]
\centering
\includegraphics[scale=0.5]{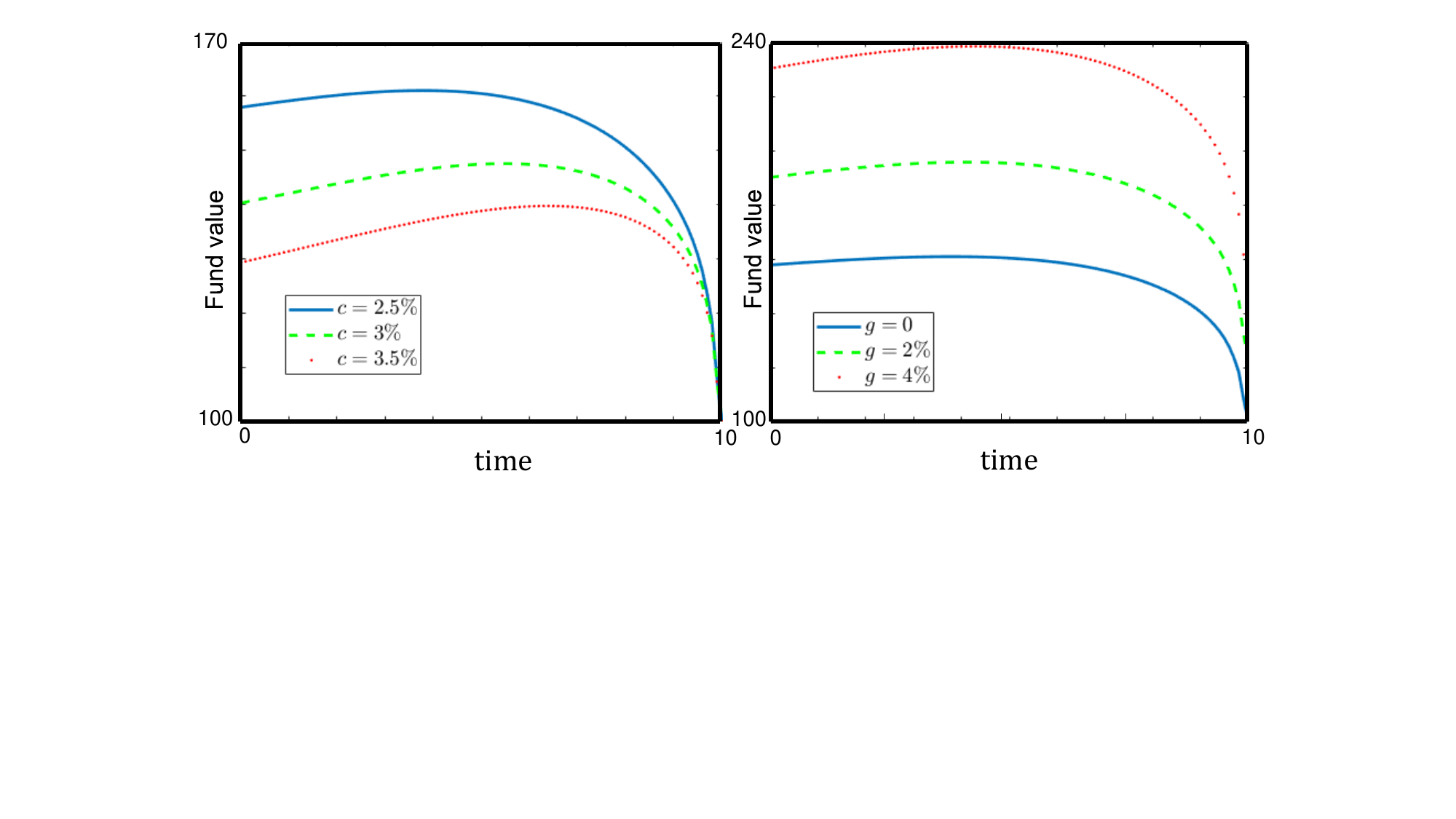}
\vspace{-4.5cm}
\caption{Sensitivity of the optimal surrender boundary $t\mapsto \ell_{x_0}(t)$ with respect to $c$ (left plot) and to $g$ (right plot).}
\label{fig4}
\end{figure}

Within this context, it is interesting to analyse the fee rate that makes the contract fairly priced. In other words, keeping all other parameters fixed and considering the contract's value $V_0=v(0,x_0)$ as a function of the fee rate $c\mapsto V_0(c)=v(0,x_0;c)$, we are after the fee rate $c^*$ for which $V_0(c^*)=x_0$. That is, $c^*$ is the fee that makes the expected present value of the cash inflows and outflows perfectly balanced. We collect in Table \ref{tabella1} the fair fee rates corresponding to the policyholder's age at the issuance of the contract (all other parameters are set as in \eqref{eq:benchmk}).
\begin{table}[h!]
	\begin{center}
		\begin{tabular}{| l | c c c|}
			\hline
			Age at issue & 50 & 60 & 70\\
			Fair fee & 2\% & 2.2\% & 2.5\%\\
			\hline
		\end{tabular}
		\vspace{8pt}
		\caption{Fair fee for different policyholder's initial ages.}
\label{tabella1} 
	\end{center}
\end{table}
The higher the policyholder's age, the higher the fee rate that makes the contract fairly priced. In other words, as the policyholder's initial age increases the demographic risk increases too (recall that in case of death the insurance pays back the full account value). On the insurer's side this risk is compensated with a higher fee in order to make the contract fair. We remark that if the fee increases, while all other parameters are fixed, the value of the VA decreases. Thus charging fees above the fair fee will make the VA's price drop below par.

\begin{table}[hb!]
	\centering
	\begin{tabular}{ |p{2.5cm} p{2cm}|p{2cm}|p{2cm}|p{2cm}|  }
		\hline
		Spread & Scenario &$V_0$&$U_0$&$V^{\text{SO}}$\\
		\hline
		$r-g=5\%$   & A    &87.2&   82.7 & 4.5\\
		& B    &89.07& 82.7    & 6.37\\
		& C   &90.97&   89.96 & 1.01\\
		& D   &92.16& 89.96 & 2.2
		\vspace{0.5cm}\\
		$r-g=3\%$   & A    &93.37&   90.56 & 2.81\\
		& B    &94.52& 90.56  & 3.96\\
		& C   &97.44&   96.75 & 0.69\\
		& D   &98.2&  96.75 & 1.45
		\vspace{0.5cm}\\
		$r-g=1\%$   & A    &103.38&   101.7 & 1.68\\
		& B    &104.04& 101.7    & 2.34\\
		& C   &107.17&   106.71 & 0.46\\
		& D   &107.64& 106.71 & 0.93 \\
		\hline
	\end{tabular}
	\vspace{0.5cm}
	\caption{The values $V_0$ of the VA contract, $U_0$ of the contract without the option of an early surrender and the value $V^{\text{SO}}$ of the SO.}
	\label{tab2}
\end{table}
\subsection{Sensitivity to risk-free and minimum guaranteed rates}
We conclude the section by analysing the impact of the spread between the risk free rate $r$ and the guaranteed minimum rate $g$ on the value of the contract as well as on the value of the embedded surrender option (SO). 
This is accomplished by comparing the value $V_0$ of the VA to the value $U_0$ of its European counterpart (i.e., without SO). We recall that $V_0$ is computed as in \eqref{eq:V0int}. Using independence of $\tau_d$ from $\cF_\infty$, we may rewrite the value $U_0$ of the European contract in \eqref{U0} as 
\begin{align*}
\begin{aligned}
U_0&= \EQ\Big[\int_0^{T}e^{-rt}{_t p_0}\mu(t)\max\{x_0e^{gt},X_t\}\ud t+e^{-r\, T}{_T p_0}\max\{x_0e^{g T},X_{T}\}\Big]\\
&=x_0 \int_0^{T}e^{-\int_0^t\mu(s)\ud s-rt}\mu(t)\Big[ e^{gt}N\Big(\big(\tfrac{\sigma}{2}+\tfrac{g-r+c}{\sigma}\big)\sqrt{t}\Big)+e^{(r-c)t}N\Big(\big(\tfrac{\sigma}{2}-\tfrac{g-r+c}{\sigma}\big)\sqrt{t}\Big)\Big]\ud t\\
&\quad +x_0 e^{-rT}{_T p_0}\Big[ e^{gT}N\Big(\big(\tfrac{\sigma}{2}+\tfrac{g-r+c}{\sigma}\big)\sqrt{T}\Big)+e^{(r-c)T}N\Big(\big(\tfrac{\sigma}{2}-\tfrac{g-r+c}{\sigma}\big)\sqrt{T}\Big)\Big],
\end{aligned}
\end{align*} 
where we used the distribution of $X$ (cf.\ \eqref{dynX0}) and $N$ denotes the cumulative standard normal distribution.

Table \ref{tab2} displays values of $V_0$, $U_0$ and $V^{\text{SO}}=V_0-U_0$. The latter is precisely the value of the SO embedded in the contract. We consider different scenarios depending on the values of the fee rate $c$ and of the surrender charge intensity $K$: Scenario A ($c=4\%$, $K=1.8\%$), Scenario B ($c=4\%$, $K=1.4\%$), Scenario C ($c=2.5\%$, $K=1.8\%$) and Scenario D ($c=2.5\%$, $K=1.4\%$). In each scenario we vary the spread $r-g$, keeping all other parameters fixed as in \eqref{eq:benchmk}.  

For each value of the spread, the value of the European contract is not affected by the surrender charge intensity $K$ and it only changes when we vary the fee rate $c$. Naturally $V_0$ is always greater than $U_0$ due to the embedded option.
In all scenarios the values $V_0$ and $U_0$ decrease as the spread $r-g$ increases. That happens because both contracts become financially less appealing to an investor when compared, for example, to bond investments. Conversely, $V^{\text{SO}}$ increases with increasing spread, hence acquiring even more value in relative terms in the overall contract's value. This is in line with the intuition that for larger spread the policyholder has greater financial incentives to exercise the SO. We additionally notice that for  $r-g=1\%$ the values of  $V_0$ and $U_0$ are above par (i.e., above the initial account value $x_0$). Instead when the spread is high enough both contract values are below par. It is hard to imagine that VA's would be traded below par, because that corresponds to situations in which the policyholder pays a price which is lower than the insurer's cost to set up the fund. However, we see that this makes sense mathematically because the fee $c$ and the surrender charge intensity, combined with the large spread $r-g$ can compensate the insurer for the initial loss. 

For any fixed value of the spread, scenario A is the most unfavourable to the policyholder as it presents relatively high values for both $c$ and $K$. More in general, contract values increase moving from scenario A to D.  The SO's value increases with $c$, i.e., the higher is the fee charged to the policyholder the more the policyholder is incentivised to abandon the contract. Instead, $V^{\text{SO}}$ decreases as the surrender charge intensity $K$ increases. 

In conclusion, Table \ref{tab2} suggests that the design of a marketable contract with mild surrender risk should follow the structure of scenario D, i.e., it should feature low values of $c$ and $K$ and, possibly, small spread $r-g$.

\section{Continuity of the value function and optimal boundary}\label{sec:contb}

In this section we begin the theoretical analysis. The state space of our problem is $[0,T]\times(0,\infty)$ and we work with the parametrisation in terms of the function $w$ from \eqref{valuew}. For an open subset $A\subseteq [0,T]\times(0,\infty)$, we denote by $C^{j,k}(A)$ the class of functions with continuous time-derivatives up to the $j$-th order and continuous space-derivatives up to $k$-th order, on $A$. If $k=j$ we simply write $C^j(A)$ and if $k=j=0$ we simply write $C(A)$. If the set $A$ is closed, the notations $C(A)$, $C^{j,k}(A)$ and $C^j(A)$ have the same meaning but with derivatives extended continuously to the closure of $A$.

Clearly, $w(T,z)=(z-1)^+$ by definition. Next we study the behaviour of $z\mapsto w(t,z)$ in the limit for $z\to 0$. Recall the definition of $t^*$ in \eqref{eq:t*}. We start with a simple observation: since $f(s)\ge 0$ for $s\in[0,t^*]$ and the law of $Z_t$ has full support on $(0,\infty)$ for every $t>0$,  then, in the case $t^*>0$, for $t\in[0,t^*)$ we have $w(t,z)>0$ for all $z\in(0,\infty)$. That is easily seen by choosing $\tau=t^*-t$ in \eqref{valuew}:
\begin{align}\label{eq:wlb}
w(t,z)&\ge \E\Big[\int_{0}^{t^*-t}e^{-cs} \,  _s p_{t}  \Big(f(t+s)+ \mu(t+s) (Z^z_s-1)^{+}  \Big)  \ud s\Big]>0,
\end{align}
where the strict inequality holds because $\mu(t\!+\!s)>0$ and $\P(Z^z_{s}>1)>0$ for all $s\in[0,t^*-t]$.
We can further refine this observation as follows.
\begin{lemma}\label{lem:stb}
Assume $t^*\in(0,T)$ and fix $t\in[0,t^*)$. Then, 
\[
w(t,z)>0\quad\text{ and }\quad w(t,z)=\sup_{t^*-t\le \tau\le T-t}\E[\cP_{t,z}(\tau)],\quad\text{ for all $z\in(0,\infty)$,}
\]
where $\cP_{t,z}(\cdot)$ was defined in \eqref{setP}.
That is, it suffices to pick stopping times in the set $[t^*-t,T-t]$. 
\end{lemma}
\begin{proof}
The lower bound $w(t,z)>0$ follows by \eqref{eq:wlb}. Setting $\tau':=(t^*-t)\ind_{\{\tau\le t^*-t\}}+\tau\ind_{\{\tau>t^*-t\}}$ for a stopping time $\tau\in[0,T-t]$, it is easy to check that $\tau'$ is a stopping time with values in $[t^*-t,T-t]$. Recalling that $f(s)\ge 0$ for $s\in[0,t^*]$ we obtain
\begin{align*}
\ind_{\{\tau\le t^*-t\}}\cP_{t,z}(\tau)&=\ind_{\{\tau\le t^*-t\}}\int_{0}^{\tau}\!\!e^{-cs}  {_s p_{t}}  \Big(f(t\!+\!s)\!+\! \mu(t\!+\!s) \big(Z^z_s\!-\!1\big)^{+}  \Big)  \ud s\\
&\le \ind_{\{\tau\le t^*-t\}}\int_{0}^{t^*-t}\!\!e^{-cs}  {_s p_{t}}  \Big(f(t\!+\!s)\!+\! \mu(t\!+\!s) \big(Z^z_s\!-\!1\big)^{+}  \Big)  \ud s=\ind_{\{\tau\le t^*-t\}}\cP_{t,z}(\tau').
\end{align*}
Moreover, $\ind_{\{\tau> t^*-t\}}\cP_{t,z}(\tau)=\ind_{\{\tau> t^*-t\}}\cP_{t,z}(\tau')$.
Then 
\[
\E[\cP_{t,z}(\tau)]=\E[\ind_{\{\tau\le t^*-t\}}\cP_{t,z}(\tau)+\ind_{\{\tau> t^*-t\}}\cP_{t,z}(\tau)]\le \E[\cP_{t,z}(\tau')].
\]
The above implies $w(t,z)\le \sup_{t^*-t\le \tau\le T-t}\E[\cP_{t,z}(\tau)]$. The reverse inequality is obvious and therefore the proof is complete.
\end{proof}

\begin{lemma}\label{lem:wincrconv}
For every $t\in[0,T]$, the function $z\mapsto w(t,z)$ is non-decreasing and convex with $\lim_{z\to\infty}w(t,z)=\infty$ and
	\begin{equation}\label{eq:w(t,0+)}
		w(t,0+):=\lim_{z\downarrow 0}w(t,z)=
		\begin{cases}
			\int_0^{t^*-t}e^{-cs}\, _s p_t f(t+s)\ud s \quad & t\in[0,t^*),\\
			0 \quad & t\in [t^*,T].
		\end{cases}
	\end{equation}
\end{lemma}
\begin{proof}
Monotonicity easily follows because $z\mapsto (Z^{z}_s(\omega)-1)^+$ is non-decreasing. Moreover, taking $\tau=T-t$ in \eqref{valuew} we obtain the lower bound 
\[
w(t,z)\ge\E\Big[\int_0^{T-t}e^{-cs}{_s}p_tf(t+s)\ud s\Big]+ e^{-c(T-t)}{_{T-t} p_t}\E[(Z^z_{T-t}-1)^+].
\] 
Since $z\mapsto Z^z_{T-t}$ is non-decreasing we conclude that $\lim_{z\to\infty}w(t,z)=\infty$, by monotone convergence. For the convexity, let $z_1,z_2\in(0,\infty)$ and $\lambda\in(0,1)$. Set $z_\lambda=\lambda z_1+(1-\lambda)z_2$. Then $Z^{z_\lambda}=\lambda Z^{z_1}+(1-\lambda)Z^{z_2}$ by linearity of the geometric Brownian motion with respect to its initial condition and $(Z^{z_\lambda}-1)^+\le \lambda (Z^{z_1}-1)^++(1-\lambda)(Z^{z_2}-1)^+$ by convexity of the positive part. Combining these inequalities we obtain for any stopping time $\tau\in[0,T-t]$
\begin{align*}
&\E\big[\cP_{t,z_\lambda}(\tau)\big]\le \E\big[\lambda\cP_{t,z_1}(\tau)+(1-\lambda)\cP_{t,z_2}(\tau)\big]\le \lambda w(t,z_1)+(1-\lambda)w(t,z_2).
\end{align*}
By arbitrariness of $\tau$ we conclude that $w(t,z_\lambda)\le \lambda w(t,z_1)+(1-\lambda)w(t,z_2)$ as needed.

It remains to calculate the limit in \eqref{eq:w(t,0+)}. Clearly the limit exists by monotonicity of $w(t,\,\cdot\,)$. Let $(z_n)_{n\in\bN}\subseteq (0,\infty)$ be such that $z_n\downarrow 0$. With no loss of generality we assume $t^*\in(0,T)$ as the cases $t^*=0$ and $t^*=T$ are easier and can be treated with the same arguments. 

First, consider $t\in[0,t^*)$. The lower bound \eqref{eq:wlb} holds for $z=z_n$. Then, letting $n\to\infty$ we have
\begin{align}\label{eq:w(t,0+)a}
w(t,0+)\geq \int_0^{t^*-t} e^{-cs} \, _sp_t f(t+s)\ud s.
\end{align}
For the opposite inequality, we recall that it suffices to take stopping times $\tau$ larger than $t^*-t$ (cf.\ Lemma \ref{lem:stb}). Since $f(s)<0$ for $s>t^*$, then for any $\tau\in[t^*-t,T-t]$ 
\begin{align}\label{eq:w(t,0+)b}
\begin{aligned}
\E[\cP_{t,z_n}(\tau)]&\le \int_0^{t^*-t} e^{-cs} \, _sp_t f(t+s)\ud s\\
&\quad+\E\bigg[\int_0^{T-t} e^{-cs} \, _sp_t\mu(t+s)\big(Z^{z_n}_s-1\big)^+\ud s+e^{-c(T-t)}\, _{T-t}p_t \big(Z^{z_n}_{T-t}-1\big)^+\Big]\\
&=:\int_0^{t^*-t} e^{-cs} \, _sp_t f(t+s)\ud s+\varphi(t,z_n).
\end{aligned}	
\end{align}	
It follows that 
\[
w(t,z_n)\le \int_0^{t^*-t} e^{-cs} \, _sp_t f(t+s)\ud s+\varphi(t,z_n)
\]
and it is easy to check by monotone convergence that $\lim_{n\to\infty}\varphi(t,z_n)=0$, because $0\le \varphi(t,z_1)<\infty$ and $Z^{z_n}_s\ge Z^{z_{n+1}}_s$ for $n\in\bN$. Combining with \eqref{eq:w(t,0+)a} we obtain \eqref{eq:w(t,0+)} for $t\in[0,t^*)$.

Now let $t\in[t^*,T]$. By the same argument as in \eqref{eq:w(t,0+)b} we have $w(t,z_n)\le \varphi(t,z_n)$, upon noticing that $f(s)\le 0$ for $s\in[t^*,T]$. Then $w(t,0+)\le 0$. Since $w(t,z)\ge 0$ for all $(t,z)\in[0,T]\times(0,\infty)$, then \eqref{eq:w(t,0+)} holds also for $t\in[t^*,T]$. 
\end{proof}

We now study continuity of the value function. First, we establish Lipschitz continuity in the $z$-variable, uniformly in $t$, and then $\frac12$-H\"older continuity in time, locally uniformly in space. Continuity of the value function is later used to obtain a free-boundary problem for $w$ and an optimal stopping time (see \eqref{eq:fbp} and \eqref{eq:tau*1}).

\begin{lemma}\label{lem:w(z)Lip}
	There exists $C>0$ such that for every $t\in[0,T]$ and $z_1,z_2\in(0,\infty)$
	\begin{equation}\label{eq:w(z)Lip}
	|w(t,z_2)-w(t,z_1)|\leq C |z_2-z_1|.
	\end{equation}
\end{lemma}
\begin{proof}
	Let $t\in[0,T]$ and $z_1,z_2\in(0,\infty)$. Without loss of generality, assume $z_1\leq z_2$. Since $z\mapsto w(t,z)$ is non-decreasing, we immediately have that
$w(t,z_2)-w(t,z_1)\geq 0$.
For the reverse inequality we notice that 
\begin{align*}
w(t,z_2)-w(t,z_1)&\le \sup_{0\le \tau\le T-t}\E\big[\cP_{t,z_2}(\tau)-\cP_{t,z_1}(\tau)\big]\\
&\le \sup_{0\le \tau\le T-t}\E\Big[ \int_{0}^{\tau}e^{-c\,s} \,  _s p_{t}  \mu(t+s) \big((Z^{z_2}_s-1)^{+}-(Z^{z_1}_s-1)^{+}\big) \ud s \notag\\
	&\qquad\qquad\qquad+\ind_{\{\tau= T-t \}}e^{-c (T-t)}  \, _{T-t} p_{t} \big((Z^{z_2}_{T-t}-1)^{+}-(Z^{z_1}_{T-t}-1)^{+}\big) \Big]\\
&\le \E\Big[\mu_M\int_{0}^{T-t} \big|Z^{z_2}_s-Z^{z_1}_s\big| \ud s+\big|Z^{z_2}_{T-t}-Z^{z_1}_{T-t}\big| \Big],	
\end{align*}	
where $\mu_M:=\max_{s\in[0,T]}\mu(s)$ (recall Assumption \ref{A1}).
For every $s\in[0,T-t]$, we have 
\begin{align*}
\E\big[|Z^{z_2}_s-Z^{z_1}_s|\big]=(z_2-z_1)\E[Z^1_s]=(z_2-z_1)e^{\alpha s},
\end{align*}
with $\alpha$ as in \eqref{eq:alpha}.
Therefore,
$0\le w(t,z_2)-w(t,z_1)\leq C|z_2-z_1|$, with $C:=e^{|\alpha|T}(1+\mu_M T)$.
\end{proof}

\begin{lemma}\label{lem:w(t)Cont}
There is $C>0$ such that for every $z\in(0,\infty)$ and $0\le t< u\le T$ we have
\begin{align}\label{eq:holder}
|w(t,z)-w(u,z)|\le C\big[(1+z)(u-t)+z\sqrt{u-t}\big].
\end{align}
\end{lemma}
\begin{proof}
Fix $z\in(0,\infty)$ and $0\le t<u\le T$. For any stopping time $\tau\in[0,T-u]$ admissible for $w(u,z)$ we define $\sigma^\tau:=\tau\ind_{\{\tau<T-u \}}+(T-t)\ind_{\{\tau=T-u \}}$. Then $\sigma^\tau\in[0,T-t]$ is a stopping time admissible for $w(t,z)$ and $w(t,z)\ge \E[\cP_{t,z}(\sigma^\tau)]$. It follows that, for any $\tau\in[0,T-u]$ 
\begin{align*}
&\E[\cP_{u,z}(\tau)]-w(t,z)\\
&\leq\E\Big[\int_0^{\tau}\!\! e^{-cs}\Big(\big( {_s p_{u}}f(u\!+\!s)\!-\! {_s p_{t}} f(t\!+\!s)\big)\!+\!\big( {_s p_{u}}\mu(u\!+\!s)\!-\! {_s p_{t}}\mu(t\!+\!s)\big)\big(Z^z_s\!-\!1\big)^+\Big)\ud s\Big]\\
&\quad-\E\Big[\ind_{\{\tau=T-u \}} \int_{T-u}^{T-t} e^{-cs} {_s p_{t}}\Big(f(t+s)+\mu(t+s)(Z^z_s-1)^+ \Big)\ud s \Big]\\
&\quad+\E\Big[\ind_{\{\tau=T-u\}}\Big(e^{-c(T-u)} {_{T-u}p_{u}}(Z^z_{T-u}-1)^+ - e^{-c(T-t)}{_{T-t}p_{t}}(Z^z_{T-t}-1)^+\Big)\Big].
\end{align*}	
Since the functions $\mu(s)$, $f(s)$ and $_s p_t$ are continuously differentiable and bounded on $[0,T]$ (Assumption \ref{A1}), we then have
\begin{align*}
&\E[\cP_{u,z}(\tau)]-w(t,z)\\
&\le c_1(u-t)\Big(1+\int_0^{T-u}\E\big[Z^z_s\big]\ud s+\E\big[Z^z_{T-u}\big]\Big)+\E\big[\big|Z^z_{T-u}-Z^z_{T-t}\big|\big],
\end{align*}
for a suitable constant $c_1>0$, that may change form line to line below, but it is independent of $u$, $t$ and $z$. By arbitrariness of $\tau$ and recalling $\E[Z^z_s]=ze^{\alpha s}$ we deduce
\begin{align}\label{eq:wt0}
\begin{aligned}
w(u,z)-w(t,z)
&\le c_1(u-t)\big(1+z\big)+ z\,\E\big[\big|Z^1_{T-u}-Z^1_{T-t}\big|^2\big]^\frac12\\
&\le c_1(u-t)\big(1+z\big)+ c_1 z\,\sqrt{u-t}.
\end{aligned}
\end{align} 
where we use H\"older inequality for the second term in the first inequality and the final expression follows by standard estimates for SDEs.

For the reverse inequality, we notice that for any stopping time $\tau\in[0,T-t]$, admissible for $w(t,z)$, the stopping time $\tau^u:=\tau\wedge (T-u)$ is admissible for $w(u,z)$ and $w(u,z)\ge \E[\cP_{u,z}(\tau^u)]$. Then, for any $\tau\in[0,T-t]$,
\begin{align*}
&\E\big[\cP_{t,z}(\tau)\big]-w(u,z)\\
&\le \E\Big[\int_0^{\tau^u}\!\!\! e^{-cs}\Big(\big( {_s p_{t}}f(t\!+\!s)\!-\! {_s p_{u}} f(u\!+\!s)\big)\!+\!\big({_s p_{t}}\mu(t\!+\!s)\!-\!{_s p_{u}}\mu(u\!+\!s)\big)(Z^z_s\!-\!1)^+\Big)\ud s\Big]\\
&\quad+\E\Big[\int_{\tau^u}^{\tau} e^{-cs}{_s p_{t}}\big(f(t+s)+\mu(t+s)(Z^z_s-1)^+ \big)\ud s\Big]\\
&\quad+\E\Big[\ind_{\{\tau=T-t \}}e^{-c(T-t)} {_{T-t}p_{t}}(Z^z_{T-t}-1)^+-\ind_{\{\tau\geq T-u \}}e^{-c(T-u)} {_{T-u}p_{u}(Z^z_{T-u}-1)^+} \Big].
\end{align*}
Using again the smoothness of $f(s)$, $\mu(s)$ and $_s p_t$, and noticing that $0\le \tau-\tau^u\le (u-t)$ we obtain
\begin{align}\label{eq:wt2}
\begin{aligned}
&\E\big[\cP_{t,z}(\tau)\big]-w(u,z)\\
&\le c_2(u-t)\Big(1+\int_0^{T-u}\E\big[Z^z_s\big]\ud s\Big)+c_2\int_{T-u}^{T-t}\E\big[Z^z_s\big]\ud s\\
&\quad+\E\Big[\ind_{\{\tau=T-t \}}e^{-c(T-t)} {_{T-t}p_{t}}(Z^z_{T-t}-1)^+-\ind_{\{\tau\geq T-u \}}e^{-c(T-u)} {_{T-u}p_{u}(Z^z_{T-u}-1)^+} \Big],
\end{aligned}
\end{align}
where $c_2>0$ is a constant that may vary from line to line below but it is independent of $u$, $t$ and $z$.
Since $\{\tau=T-t \}\subseteq\{\tau\geq T-u \}$ then $\ind_{\{\tau=T-t \}}\leq \ind_{\{\tau\geq T-u \}}$ and
\begin{align}\label{eq:wt3}
\begin{aligned}
&\E\Big[\ind_{\{\tau=T-t \}}e^{-c(T-t)} {_{T-t}p_{t}}(Z^z_{T-t}-1)^+-\ind_{\{\tau\geq T-u \}}e^{-c(T-u)} {_{T-u}p_{u}(Z^z_{T-u}-1)^+} \Big]\\
&\le \E\Big[\ind_{\{\tau\ge T-u \}}e^{-c(T-u)} \Big({_{T-t}p_{t}}(Z^z_{T-t}-1)^+-{_{T-u}p_{u}(Z^z_{T-u}-1)^+}\Big) \Big]\\
&\le c_2(u-t)z\E\big[Z^1_{T-u}\big]+z\E\big[\big|Z^1_{T-u}-Z^1_{T-t}\big|\big].
\end{aligned}
\end{align}
Combining \eqref{eq:wt2} and \eqref{eq:wt3}, recalling $\E[Z^1_s]=e^{\alpha s}$ and arbitrariness of $\tau$, we obtain
\begin{align*}
w(t,z)-w(u,z)&\le c_2(u-t)\big(1+z\big)+z\E\big[\big|Z^1_{T-u}-Z^1_{T-t}\big|\big]\\
&\le c_2(u-t)\big(1+z\big)+c_2 z\sqrt{u-t},
\end{align*}
with the same SDE estimate as in \eqref{eq:wt0}. Now \eqref{eq:holder} follows from the bound above and \eqref{eq:wt0}.
\end{proof}

As usual in optimal stopping theory, we let
\begin{align}\label{setC0}
\mathcal{C}=& \big\{ (t,z)\in [0,T)\times (0,\infty) : w(t,z)>0 \big\} 
\end{align}
and
\begin{align}\label{setS0}
\mathcal{D}=& \big\{ (t,z)\in [0,T)\times (0,\infty) : w(t,z)=0 \big\}\cup\big(\{T \}\times(0,\infty)\big) 
\end{align}
be respectively the so-called continuation and stopping regions. Thanks to Lemmas \ref{lem:w(z)Lip} and \ref{lem:w(t)Cont} and setting
\begin{align}\label{H}
H(t,z):=\mu(t)(z-1)^{+}+f(t),
\end{align}
it is well-known (cf.,\ e.g., \cite[Ch.\ III]{peskir2006optimal}) that $w\in C^{1,2}(\cC)$ and it solves the free-boundary problem 
\begin{align}\label{eq:fbp}
\begin{aligned}
w_t(t,z)+\big(\cL w\big)(t,z)-\big(c+\mu(t)\big)w (t,z)+H(t,z)&=0, \qquad\quad\qquad (t,z)\in\cC,\\
w(t,z)&=0, \qquad \quad\qquad(t,z)\in\partial\cC,\text{ with } t<T,\\
w(T,z)&=(z-1)^+,\qquad z\in(0,\infty),
\end{aligned}
\end{align}
where $\cL$ is the infinitesimal generator of the process $Z$, defined as
\begin{align}\label{eq:L}
\big(\cL \varphi\big)(t,z) = \alpha z \varphi_z(t,z)+\tfrac{\sigma^2}{2} z^2 \varphi_{zz}(t,z),
\end{align}
for sufficiently smooth functions $\varphi:[0,T]\times(0,\infty)\to\R$.
Here $\varphi_t$, $\varphi_x$ and $\varphi_{xx}$ represent the partial derivatives of the function $\varphi$ with respect to time and space.

From the continuity of $w$ and classical optimal stopping theory (cf.\ \cite[Ch.\ 3, Sec.\ 4]{shiryaev2007optimal}), we also deduce that the stopping time 
\begin{align}\label{eq:tau*1}
\tau^*_{t,z}=\inf\{s\ge 0:(t+s,Z^z_s)\in\cD\}
\end{align}
is the smallest optimal stopping time for our problem \eqref{valuew}. Moreover, the process $(U^{t,z}_{s\wedge\tau})_{s\in[0,T-t]}$ is a continuous super-martingale for any stopping time $\tau$ and a martingale for $\tau=\tau^*_{t,z}$, where
\begin{align}\label{eq:UD}
\begin{aligned}
&U^{t,z}_{s}:=D^t(s)w\big(t+s,Z^z_{s}\big)+\int_0^{s}D^t(u)H(t+u,Z^z_u)\ud u\\
&\text{ with }D^t(s):=\exp\Big(-cs-\int_0^s\mu(t+\theta)\ud\theta\Big).
\end{aligned}
\end{align}
It is therefore important to determine the form of $\cC$ and $\cD$. 

Since $z \mapsto w(t,z)$ is non-decreasing (Lemma \ref{lem:wincrconv}), then $(t,z_0)\in \mathcal{C}$ implies $(t,z)\in \mathcal{C}$, for all $z>z_0$. Hence, defining the function $b:[0,T) \to  [0,\infty]$ by
$b(t):=\inf\{z>0: w(t,z)>0 \}$, we obtain 
\begin{align}\label{setC1}
\begin{aligned}
\cC&= \big\{ (t,z)\in [0,T)\times (0,\infty) : z>b(t) \big\}\\
\cD&= \big\{ (t,z)\in [0,T)\times (0,\infty) : z\le b(t) \big\}\cup\big(\{T\}\times(0,\infty)\big).
\end{aligned}
\end{align}
Since $\lim_{z \uparrow \infty}w(t,z)=\infty$ (Lemma \ref{lem:wincrconv}), then $b(t)<\infty$ for every $t\in[0,T)$.

First of all, we recall that \eqref{eq:wlb} implies 
\begin{align}\label{eq:sC}
[0,t^*)\times(0,\infty)\subseteq\cC\quad (\text{i.e., }b(t)=0\text{ for every }t\in[0,t^*)).
\end{align}
Next, we define the set
$\mathcal{R}:=\{(t,z)\in[0,T)\times (0,\infty): H(t,z)>0\}$. For $(t,z)\in\mathcal{R}$, taking $\tau_\mathcal{R}=\inf\{s\ge 0: (t\!+\!s,Z^z_s)\notin \mathcal{R}\}\wedge(T\!-\!t)$, we obtain $w(t,z)\ge \E[\cP_{t,z}(\tau_\mathcal{R})]>0$. Hence, $\mathcal R\subset\mathcal{C}$. 

Our main result about the optimal stopping boundary $t\mapsto b(t)$ is stated in the next theorem.
\begin{theorem}\label{thm:bCont}
The optimal boundary $t\mapsto b(t)$ is continuous at any $t\in[0,T]$ for $t\neq t^*$ with $b(T):=\lim_{t\uparrow T}b(t)=1$.
\end{theorem}
Continuity of $b$ on $[0,t^*)$ is obvious by \eqref{eq:sC}. Instead, the proof for the remainder of the theorem is highly non-trivial. In our problem it is not possible to establish monotonicity of the boundary (indeed that fails as illustrated, e.g., in Figure \ref{fig2}), which is a structural assumption for known (probabilistic) results on continuity of the optimal boundary (cf.\ \cite{de2015note} and \cite{peskir2019continuity}). We must instead rely upon a change of coordinates, linked to specific bounds on the increments of $b$ (a sort of one-sided Lipschitz condition). For that we need an in-depth analysis of the properties of the value function $w$, which will be developed in the next sections. The proof of Theorem \ref{thm:bCont} will be provided in Section \ref{sect:bCont}.

Before closing this section we can state an initial property of $b$ near the maturity and an associated corollary for the optimal stopping time.
\begin{proposition}\label{prop:limsupT}
If $t^*\in[0,T)$, then $\limsup_{t\to T}b(t)=1$.
\end{proposition}
\begin{proof}
It is easy to prove that $\limsup_{t\to T}b(t)\leq 1$. Indeed, for $z>1$, Lemmas \ref{lem:w(z)Lip} and \ref{lem:w(t)Cont} and $w(T,z)=(z-1)^+>0$ imply that there exists $\eps\in(0,T)$ such that $w(t,z)>0$ for every $t\in[T-\eps,T)$. Therefore, $[T-\eps,T)\times(1,\infty)\in\cC$, which implies $\limsup_{t\to T}b(t)\leq 1$.

For the reverse inequality we need a bit more work.
Arguing by contradiction, let us assume that $\limsup_{t\to T}b(t)< 1$. Then, there exists $b_0<1$ and $\eps\in(0,T)$ such that $b(t)<b_0$ for every $t\in(T-\eps,T)$. Let $z_1<z_2$ be such that $(z_1,z_2)\subset(b_0,1)$ and fix $(t,z)\in(T-\eps,T)\times(z_1,z_2)$. With no loss of generality we also assume $t> t^*$. Define
$\sigma_{z}:=\inf\{t\geq 0: Z^z_t\notin(z_1,z_2) \}$ 	and notice that, since $b(t)<b_0<z_1$ for $t\in(T-\eps,T)$ then it must be $\tau^*_{t,z}\geq \sigma_z\wedge(T-t)$, $\P$-a.s. By the martingale property of the value function (recall \eqref{eq:UD}), we have 
\[
w(t,z)=\E\Big[\int_0^{\sigma_z\wedge(T-t)}D^t(s)H(t\!+\!s,Z^z_s)\ud s\!+\!D^t(\sigma_z\wedge(T\!-\!t))w\big(t\!+\!\sigma_z\wedge(T\!-\!t), Z^z_{\sigma_z\wedge(T-t)}\big)\Big].
\]
For every $s\leq\sigma_z$ we have $Z^z_s\le z_2<1$ and therefore
$H(t+s,Z^z_s)\leq -d_0$ for some $d_0>0$. Then, the bound $1\ge D^t(s)\ge D^t(T-t)\ge \min_{t\in[0,T]}D^t(T-t)>0$  yields
\[
w(t,z)\leq -d_1\E\big[\sigma_z\wedge(T-t)\big]+\E\big[w\big(t+\sigma_z\wedge(T-t), Z^z_{\sigma_z\wedge(T-t)}\big)\big],
\]
for some constant $d_1>0$. The constant $d_1$ may vary from line to line, it may depend on $z_1$ and $z_2$ but it is independent of $t$ and $z$.
Since on the event $\{\sigma_z>T-t \}$ we have
\[
w\big(t+\sigma_z\wedge(T-t), Z^z_{\sigma_z\wedge(T-t)}\big)=w(T,Z^z_{T-t})=(Z_{T-t}^z-1)^+=0,
\]
and $w(t,\,\cdot\,)$ is non-decreasing (Lemma \ref{lem:wincrconv}), then
\begin{align}\label{eq:w<}
\begin{aligned}
w(t,z)&\leq-d_1\E\big[\sigma_z\wedge(T-t)\big]+\E\big[\ind_{\{\sigma_z\leq T-t\}}w(t+\sigma_z,z_2)\big]\\
&\leq-d_1\big(\E\big[\sigma_z\wedge(T-t)\big]-\P\big(\sigma_z\leq T-t\big)\big),
\end{aligned}
\end{align}
where for the second inequality we used that $w$ is bounded on compacts and we have changed $d_1$ as needed.

For $\theta_z:=\min\{z_2-z,z-z_1 \}>0$, Markov's inequality yields for any $p\in[2,\infty)$
\[
\P(\sigma_z\leq T-t)= \P\Big(\sup_{0\leq s\leq T-t} |Z^z_s-z|\geq\theta_z\Big)\leq \frac{z^p}{\theta^p_z}\E\Big[\sup_{0\leq s\leq T-t} |Z^1_s-1|^p\Big]\le \frac{z^p}{\theta^p_z}d_2(T-t)^{\frac{p}{2}},
\]
where $d_2>0$ is independent of $(t,z)$ and the final inequality is by standard estimates for SDEs (see, e.g., \cite[Theorem 9.1]{baldi2017stochastic}).
Then, relabelling all constants as $d_3=d_3(z)>0$ and continuing from \eqref{eq:w<}, 
\begin{align*}
w(t,z)&\leq -d_3\big(\E\big[\sigma_z\ind_{\{\sigma_z\le T-t \}}\big]+(T-t)\P(\sigma_z> T-t)-(T-t)^{\frac{p}{2}}\big)\\
&\leq -d_3 (T-t)\big(\P(\sigma_z> T-t)-(T-t)^{\frac{p}{2}-1} \big).
\end{align*}
Since $\P(\sigma_z> T-t)\to 1$ as $t\to T$ and we can choose $p>2$, then there exists $t\in(T-\eps,T)$ such that $w(t,z)<0$, which is a contradiction. This shows that $\limsup_{t\to T}b(t)\ge 1$ which, combined with the first part of the proof yields the claim in the proposition.
\end{proof}

\begin{corollary}\label{lem:Z>1}
For every $(t,z)\in[0,T]\times[0,\infty)$, we have $\{\tau^*_{t,z}=T-t \}\subseteq\{Z^z_{T-t}\geq 1\}$.
\end{corollary}
\begin{proof}
By Proposition \ref{prop:limsupT}, there exists a sequence $(s_n)_{n\in\bN}$ such that $s_n\to T-t$ as $n\to\infty$ and $\lim_{n\to\infty} b(t+s_n)=1$. Therefore, for every $\eps>0$ there exists $\bar{n}_\eps\in\bN$ such that $b(t+s_n)\geq 1-\eps$, for every $n\geq\bar{n}_\eps$. Now let $\omega\in\{\tau^*_{t,z}=T-t \}$. Then, it must be $Z^z_{s_n}(\omega)>b(t+s_n)\geq 1-\eps$ for every $n\geq\bar{n}_\eps$.
Letting $n\to\infty$ we obtain $Z^z_{T-t}(\omega)\ge 1-\eps$ and, by arbitrariness of $\eps>0$, we obtain $Z_{T-t}(\omega)\geq 1$.
\end{proof}

\section{A monotone transformed boundary}\label{sec:beta}
In this section we construct a function $t\mapsto\beta(t)$ as a bijective, continuous transformation of the optimal boundary $t\mapsto b(t)$. The advantage of working with this transformation is that $\beta$ can be shown to be the \textit{monotone} optimal boundary of an auxiliary optimal stopping problem. That enables the use of some known ideas to establish continuity of $\beta$ and then transfer such continuity to $b$. In order to deliver our agenda, we first need to study the partial derivatives of the value function, $w_t$ and $w_z$.

The next lemma provides a formula for $w_z$. We recall the class $L^\infty((0,T)\times (0,\infty))$ of a.e.\ bounded functions on $(0,T)\times (0,\infty)$.
\begin{lemma}\label{lem:w_z}
For every $(t,z)\in\big([0,T)\times(0,\infty)\big)\setminus\partial\cC$, we have
\begin{align}\label{eq:w_z}
\begin{aligned}
w_z(t,z)=\E\Big[\int_0^{\tau^*_{t,z}}e^{-cs} {_s p_{t}}\mu(t+s)Z^1_s\ind_{\{Z^z_s\geq 1\}}\ud s+\ind_{\{\tau^*_{t,z}=T-t \}} e^{-c(T-t)} {_{T-t}p_{t}}Z^1_{T-t}\Big].
\end{aligned}
\end{align}  
In particular, $w_z\in L^\infty((0,T)\times(0,\infty))$ and $w_z(t,z)>0$ for every $(t,z)\in\cC$.
\end{lemma}
\begin{proof}
	Since $w(t,z)=0$ for $(t,z)\in\cD$ with $t<T$, then $w_z(t,z)=0$ for $(t,z)\in\cD\setminus\partial\cC$. This agrees with \eqref{eq:w_z}, since $\tau^*_{t,z}=0$ for $(t,z)\in\cD$ with $t<T$.	

Now let $(t,z)\in\cC$ and $\eps>0$. The stopping time $\tau^*_{t,z}$ is optimal for $w(t,z)$ and admissible for $w(t,z+\eps)$. Then,
\begin{align*}
w(t,z+\eps)-w(t,z)&\geq\E\bigg[\int_0^{\tau^*_{t,z}}e^{-cs} \, _s p_{t}\mu(t+s)\Big(\big(Z^{z+\eps}_s-1\big)^+-\big(Z^{z}_s-1\big)^+\Big)\ud s\\
&\hspace{28pt}+\ind_{\{\tau^*_{t,z}=T-t \}}e^{-c(T-t)}\, _{T-t}p_{t}\Big(\big(Z^{z+\eps}_{T-t}-1\big)^+-\big(Z^{z}_{T-t}-1\big)^+\Big)\bigg].
\end{align*}
For every $s\in[0,T]$, 
\[
\big(Z^{z+\eps}_s\!-\!1\big)^+\!-\!\big(Z^{z}_s\!-\!1\big)^+=\big(Z_s^{z+\eps}\!-\!Z^z_s\big)\ind_{\{Z^z_s\geq 1 \}}\!+\!\big(Z_s^{z+\eps}\!-\!1\big)\ind_{\{Z^{z+\eps}_s\geq 1>Z^z_s\}}\geq \big(Z_s^{z+\eps}\!-\!Z^z_s\big)\ind_{\{Z^z_s\geq 1 \}}.
\]
Hence, by recalling also Corollary \ref{lem:Z>1},
\begin{align*}
w(t,z+\eps)-w(t,z)\geq\E\Big[&\int_0^{\tau^*_{t,z}}e^{-cs} \, _s p_{t}\mu(t+s)\big(Z^{z+\eps}_s-Z^{z}_s\big)\ind_{\{Z^z_s\geq 1 \}}\ud s\\
&+\ind_{\{\tau^*_{t,z}=T-t \}}e^{-c(T-t)}\, _{T-t}p_{t}\big(Z^{z+\eps}_{T-t}-Z^{z}_{T-t}\big)\Big].
\end{align*}
Since $Z^{z+\eps}_s-Z^z_s=\eps Z^1_s$ for all $s\in[0,T]$, the random variable under expectation is upper bounded by $\eps\cdot C\cdot \exp(\sup_{s\in[0,T]}\sigma W_s)$ for some constant $C>0$ independent of $\eps$.
Dividing by $\eps$ and letting $\eps\to 0$, since $w\in C^{1,2}(\cC)$ we obtain
\begin{align}\label{eq:w_z>}
\begin{aligned}
w_z(t,z)\geq\E\Big[\int_0^{\tau^*_{t,z}}e^{-cs}\, _s p_{t}\mu(t+s)Z^1_s\ind_{\{Z^z_s\geq 1\}}\ud s+\ind_{\{\tau^*_{t,z}=T-t \}} e^{-c(T-t)}\, _{T-t}p_{t}Z^1_{T-t}\Big],
\end{aligned}
\end{align}
by the dominated convergence theorem. 

To show the opposite inequality, for $\eps>0$ we have
\begin{align}\label{eq:w(t,z)-w(t,z-eps)}
\begin{aligned}
w(t,z)-w(t,z-\eps)&\leq\E\Big[\int_0^{\tau^*_{t,z}}e^{-cs} {_s p_{t}}\mu(t+s)\Big(\big(Z^{z}_s-1\big)^+-\big(Z^{z-\eps}_s-1\big)^+\Big)\ud s\\
&\qquad+\ind_{\{\tau^*_{t,z}=T-t \}}e^{-c(T-t)} {_{T-t}p_{t}}\Big(\big(Z^{z}_{T-t}-1\big)^+-\big(Z^{z-\eps}_{T-t}-1\big)^+\Big)\Big].
\end{aligned}
\end{align}
Notice that we can write $Z^{z-\eps}_s=Z^z_s-\eps Z^1_s$ and therefore
\begin{align*}
\big(Z^z_s-1\big)\ind_{\{Z^z_s\geq 1> Z^{z-\eps}_s \}}&=\big(Z^z_s-1\big)\ind_{\{Z^z_s-1\geq 0\}}\ind_{\{ Z^{z}_s-\eps Z^1_s-1<0 \}}\\
&=\big(Z^z_s-1\big)\ind_{\{0\leq Z^{z}_s-1<\eps Z^1_s \}}\leq\eps Z^1_s\ind_{\{0\leq Z^{z}_s-1<\eps Z^1_s \}}.
\end{align*}
For every $s\in[0,T]$,
	\begin{align*}
		\big(Z^{z}_s-1\big)^+-\big(Z^{z-\eps}_s-1\big)^+&=\big(Z_s^{z}-Z^{z-\eps}_s\big)\ind_{\{Z^{z-\eps}_s\geq 1 \}}+\big(Z_s^{z}-1\big)\ind_{\{Z^{z}_s\geq 1>Z^{z-\eps}_s\}}\\
		&\leq \eps Z_s^{1}\ind_{\{Z^{z-\eps}_s\geq 1 \}}+\eps Z^1_s\ind_{\{0\leq Z^{z}_s-1<\eps Z^1_s \}}.
	\end{align*}
	Plugging this inequality into \eqref{eq:w(t,z)-w(t,z-eps)}, dividing by $\eps$ and letting $\eps\to 0$, we obtain
\begin{align*}
w_z(t,z)\leq\E\Big[\int_0^{\tau^*_{t,z}}e^{-cs} {_s p_{t}}\mu(t+s)Z^1_s\ind_{\{Z^z_s\geq 1\}}\ud s+\ind_{\{\tau^*_{t,z}=T-t \}} e^{-c(T-t)} {_{T-t}p_{t}}Z^1_{T-t}\Big],
\end{align*}
where we also used that $\P(Z^z_s=1)=0$ for all $s\in[0,T-t]$.
Combining the above inequality with \eqref{eq:w_z>} yields \eqref{eq:w_z}.
	
Boundedness of $w_z$ is immediate: for all $(t,z)\in ([0,T]\times(0,\infty))\setminus\partial\cC$,
\begin{align*}
\big|w_z(t,z)\big|&\le \mu_M \int_0^{T-t}\E[Z^1_s]\ud s+ \E\big[Z^1_{T-t}\big]
\le e^{|\alpha|T}\big(1+\mu_M T\big),
\end{align*}	
where $\mu_M=\max_{s\in[0,T]}\mu(s)$.
Hence $w_z\in L^\infty((0,T)\times (0,\infty))$, as claimed, because $\partial\cC$ has zero Lebesgue measure in $[0,T]\times(0,\infty)$.
	
Finally, if $(t,z)\in\cC$ then $\tau^*_{t,z}>0$, $\P$-a.s. Moreover, the fact that the law of $Z^z_s$ has full support on $(0,\infty)$ for all $s>0$ yields $\P(Z^z_s\geq 1+\delta\text{ for }s\in[t_1,t_2])>0$, for any $t< t_1<t_2<T$ and any $\delta>0$. Then, it is clear from \eqref{eq:w_z} that $w_z(t,z)>0$ for every $(t,z)\in\cC$.
\end{proof}

\begin{corollary}
For every $(t,z)\in([0,T)\times (0,\infty))\setminus\partial\cC$,
\begin{align}\label{eq:zw_z}
	zw_z(t,z)=\E\Big[\int_0^{\tau^*_{t,z}}e^{-cs}{_s p_t}\mu(t+s)\ind_{\{Z^z_s>1\}}Z^z_s\ud s+\ind_{\{\tau^*_{t,z}=T-t\}}e^{-c(T-t)}{_{T-t} p_t}Z^z_{T-t}\Big].
\end{align}
\end{corollary}

The next lemma provides an upper bound for $w_t$.
\begin{lemma}\label{lem:w_tUpBound}
For every $(t,z)\in\big([0,T)\times(0,\infty)\big)\setminus\partial\cC$, we have
\begin{align}\label{eq:w_tleq}
\begin{aligned}
w_t(t,z) &\leq \E\Big[\int_0^{\tau^*_{t,z}}e^{-cs}{_s p_t}\big(\mu(t)-\mu(t+s)\big)\big(f(t+s)+\mu(t+s)(Z^z_s-1)^+\big)\ud s\Big]\\
&\quad+\E\Big[\int_0^{\tau^*_{t,z}}e^{-cs}{_s p_t}\big(\dot f(t+s)+\dot \mu(t+s)(Z^z_s-1)^+\big)\ud s\Big]\\
&\quad+\E\Big[\ind_{\{\tau^*_{t,z}=T-t \}}e^{-c(T-t)}{_{T-t} p_t}\big(c(Z^z_{T-t}-1)^+- f(T)\big) \Big]\\
&\quad+c_0\E\Big[\ind_{\{\tau^*_{t,z}=T-t \}}e^{-c(T-t)}{_{T-t} p_t}Z^z_{T-t}\Big],
\end{aligned}
\end{align}
where $c_0>0$ is a given constant, independent of $(t,z)$.
\end{lemma}
\begin{proof}
If $(t,z)\in\cD\setminus\partial\cC$ with $t<T$, then \eqref{eq:w_tleq} trivially holds because $w_t(t,z)=w(t,z)=0$ due to $\tau^*_{t,z}=0$, $\P$-a.s.
	
Let $(t,z)\in\cC$ and $\eps\in(0,t)$. For the sake of simplicity, in this proof, we recall the notation $D^t(s)=e^{-cs}{_s p_t}$ from \eqref{eq:UD}. The stopping time $\tau=\tau^*_{t,z}\ind_{\{\tau^*_{t,z} <T-t\}}+(T-t+\eps)\ind_{\{\tau^*_{t,z}=T-t \}}$ is admissible but sub-optimal for $w(t-\eps,z)$. Then, using $w(t-\eps,z)\ge \E[\cP_{t-\eps,z}(\tau)]$ we have
\begin{align}\label{eq:w(t-eps,z)}
\begin{aligned}
&w(t,z)-w(t-\eps,z)\\
&\leq \E\Big[\int_0^{\tau^*_{t,z}}\Big(D^t(s) \big(f(t+s)+\mu(t+s)(Z^z_s-1)^+\big)\\
&\qquad\qquad\qquad-D^{t-\eps}(s) \big(f(t-\eps+s)+\mu(t-\eps+s)(Z^z_s-1)^+\big)\Big)\ud s\Big]\\
&\quad-\E\Big[\ind_{\{\tau^*_{t,z}=T-t \}}\int_{T-t}^{T-t+\eps} D^{t-\eps}(s)\Big(f(t-\eps+s)+\mu(t-\eps+s)(Z^z_s-1)^+ \Big)\ud s \Big]\\
&\quad+\E\Big[\ind_{\{\tau^*_{t,z}=T-t \}}\Big(D^t(T-t)(Z^z_{T-t}-1)^+-D^{t-\eps}(T-t+\eps)(Z^z_{T-t+\eps}-1)^+ \Big)\Big].
\end{aligned}	
\end{align}

For the last term of \eqref{eq:w(t-eps,z)}, using 
\begin{align*}
\frac{D^{t-\eps}(T-t+\eps)}{D^t(T-t)}=e^{-c\eps-\int_0^\eps\mu(t-\eps+\theta)\ud \theta}
\end{align*}
we get
\begin{align}\label{eq:w(t-eps,z)2}
\begin{aligned}
&\E\Big[\ind_{\{\tau^*_{t,z}=T-t \}}\Big(D^t(T-t)(Z^z_{T-t}-1)^+-D^{t-\eps}(T-t+\eps)(Z^z_{T-t+\eps}-1)^+ \Big)\Big]\\
&=\E\Big[\ind_{\{\tau^*_{t,z}=T-t \}}D^t(T-t)\Big((Z^z_{T-t}-1)^+-e^{-c\eps-\int_0^\eps\mu(t-\eps+\theta)\ud\theta} (Z^z_{T-t+\eps}-1)^+\Big)\Big]\\
&=\E\Big[\ind_{\{\tau^*_{t,z}=T-t \}}D^t(T-t)\Big((Z^z_{T-t}-1)^+\big(1-e^{-c\eps-\int_0^\eps\mu(t-\eps+\theta)\ud\theta}\big)\\
&\qquad\qquad\qquad+e^{-c\eps-\int_0^\eps\mu(t-\eps+\theta)\ud\theta}\big[(Z^z_{T-t}-1)^+-\big(Z^z_{T-t}e^{(\alpha-\sigma^2/2)\eps+\sigma B_\eps}-1\big)^+\big] \Big)\Big],
\end{aligned}
\end{align}
where $B_u:=W_{T-t+u}-W_{T-t}$, $u\in[0,\infty)$, is a standard Brownian motion under $\P$ for the filtration $(\cF_{T-t+u})_{u\ge 0}$. Define
\[
\varphi(t,y):=\E\big[(N^y_{T-t}-1)^+\big],\qquad (t,y)\in[0,T]\times(0,\infty),
\]
where $\ud N^y_s=\alpha N^y_s\ud s+\sigma N^y_s\ud B_s$ with $N^y_0=y$. The function $\varphi$ is analogue to a European Call option written on $(N_t)_{t\in[0,T]}$ with strike price equal to one, maturity at $T$ and zero risk-free rate. It is well-known and it can be checked directly using the log-normal density that $\varphi\in C^{1,2}([0,T)\times(0,\infty))\cap C([0,T]\times(0,\infty))$ solves the Cauchy problem
\begin{align}\label{eq:PDEvarphi}
\varphi_t(t,y)=-\tfrac{\sigma^2y^2}{2}\varphi_{yy}(t,y)-\alpha y\varphi_y(t,y),\quad  (t,y)\in[0,T)\times(0,\infty),
\end{align}
with $\varphi(T,y)=(y-1)^+$, for $y\in(0,\infty)$. It is also not difficult to check that $|\varphi_y(t,y)|\leq e^{|\alpha|T}=:c_\alpha$ and $\varphi_{yy}(t,y)\ge 0$ for $(t,y)\in[0,T)\times(0,\infty)$.
Thus, from \eqref{eq:PDEvarphi} we deduce the bound $\varphi_t(t,y)\le -\alpha y\varphi_y(t,y)\le|\alpha|c_\alpha y $. Using the latter, for the last term of \eqref{eq:w(t-eps,z)2} we have
\begin{align}\label{eq:phi(T-eps)}
\begin{aligned}
&\E\Big[\ind_{\{\tau^*_{t,z}=T-t\}}\Big( (Z^z_{T-t}-1)^+-(Z^z_{T-t}e^{(\alpha-\sigma^2/2)\eps+\sigma B_\eps}-1)^+\Big)\Big]\\
&=\E\Big[\ind_{\{\tau^*_{t,z}=T-t\}}\Big( (Z^z_{T-t}-1)^+-\E\Big[\big(Z^z_{T-t}e^{(\alpha-\sigma^2/2)\eps+\sigma B_\eps}-1\big)^+\Big|\cF_{T-t}\Big]\Big) \Big]\\
&=\E\Big[\ind_{\{\tau^*_{t,z}=T-t\}}\Big(\varphi(T,Z^z_{T-t})-\varphi(T-\eps,Z^z_{T-t}) \Big)\Big]\\
	&=\E\Big[\ind_{\{\tau^*_{t,z}=T-t\}}\int_{T-\eps}^T \varphi_t(s,Z^z_{T-t})\ud s\Big]\leq |\alpha|c_\alpha\eps\E\Big[\ind_{\{\tau^*_{t,z}=T-t\}} Z^z_{T-t}\Big].
\end{aligned}
\end{align}
Combining \eqref{eq:w(t-eps,z)}, \eqref{eq:w(t-eps,z)2} and \eqref{eq:phi(T-eps)} we  obtain
\begin{align*}
&w(t,z)-w(t-\eps,z)\notag\\
&\leq \E\Big[\int_0^{\tau^*_{t,z}} \big(D^t(s)-D^{t-\eps}(s)\big) \big(f(t+s)+\mu(t+s)(Z^z_s-1)^+\big)\ud s\Big] \notag\\
&\quad+\E\Big[\int_0^{\tau^*_{t,z}}D^{t-\eps}(s) \Big(f(t\!+\!s)\!-\!f(t\!-\!\eps\!+\!s)\!+\!\big(\mu(t\!+\!s)\!-\!\mu(t\!-\!\eps\!+\!s)\big)(Z^z_s\!-\!1)^+\Big)\ud s\Big]\\
&\quad-\E\Big[\ind_{\{\tau^*_{t,z}=T-t \}}\int_{T-t}^{T-t+\eps} D^{t-\eps}(s)\Big(f(t-\eps+s)+\mu(t-\eps+s)(Z^z_s-1)^+ \Big)\ud s \Big]\\
&\quad+\E\Big[\ind_{\{\tau^*_{t,z}=T-t\}}D^t(T-t)(Z^z_{T-t}-1)^+\Big(1-e^{-c\eps-\int_0^\eps\mu(t-\eps+\theta)\ud\theta}\Big)\Big]\\
&\quad+c_\alpha|\alpha|\eps\E\Big[ \ind_{\{\tau^*_{t,z}=T-t\}}D^t(T-t)e^{-c\eps-\int_0^\eps \mu(t-\eps+\theta)\ud\theta}Z^z_{T-t}\Big].
\end{align*}
Since $t\mapsto D^t(s)$, $t\mapsto f(t)$ and $t\mapsto \mu(t)$ are Lipschitz-continuous on $[0,T]$, and $Z$ is a geometric Brownian motion,
dividing both sides by $\eps$ and letting $\eps\to 0$ we can use the dominated convergence theorem to obtain
\begin{align*}
	w_t(t,z)&\leq \E\Big[\int_0^{\tau^*_{t,z}}D^t(s)\big(\mu(t)-\mu(t+s) \big)\big(f(t+s)+\mu(t+s)(Z^z_s-1)^+ \big)\ud s\Big]\\
	&\quad+\E\Big[\int_0^{\tau^*_{t,z}}D^t(s)\big(\dot f(t+s)+\dot\mu(t+s)(Z^z_s-1)^+ \big)\ud s\Big]\\
	&\quad-\E\Big[\ind_{\{\tau^*_{t,z}=T-t\}}D^t(T-t)\big(f(T)+\mu(T)(Z^z_{T-t}-1)^+ \big)\Big]\\
	&\quad+\E\Big[\ind_{\{\tau^*_{t,z}=T-t\}}D^t(T-t)\big([c+\mu(T)] (Z^z_{T_t}-1)^++|\alpha|c_\alpha Z^z_{T-t}\big) \Big].
\end{align*}
That leads to the desired result \eqref{eq:w_tleq}, after relabelling $|\alpha|c_\alpha$ as $c_0$.
\end{proof}

Thanks to Lemma \ref{lem:wincrconv}, we know that $w(t,b(t))=0$ for every $t\in(t^*,T]$. Moreover, defining
\begin{equation}\label{eq:b_delta}
b_\delta(t):=\inf\{z\in(0,\infty): w(t,z)>\delta \}, \qquad t\in(t^*,T], \:\delta>0,
\end{equation}
we have, by continuity and Lemma \ref{lem:wincrconv}, that $w(t,b_\delta(t))=\delta$ for every $t\in(t^*,T]$. Notice that, since $w(T,z)=(z-1)^+$, then $b_\delta(T)=1+\delta$. By monotonicity of $z\mapsto w(t,z)$, it is also clear that $\delta\mapsto b_\delta(t)$ decreases as $\delta\downarrow 0$ for every $t\in(t^*,T]$, with $b_\delta(t)>b(t)\geq 0$. Thus, $\lim_{\delta\to 0} b_\delta(t)=:b_0(t)$ exists and it satisfies $b_0(t)\geq b(t)$. Finally, by continuity of $w$ and the fact that $w(t,b_\delta(t))=\delta$, we obtain $w(t,b_0(t))=0$. Hence, by definition of $b$, we must have $b_0(t)\leq b(t)$ and so
\[
\lim_{\delta\to 0}b_\delta(t)=b(t), \qquad t\in(t^*,T].
\]
Since $b_\delta(t)>b(t)\geq 0$ for every $t\in(t^*,T]$ and $\delta>0$, then $(t,b_\delta(t))\in\cC$ for $t<T$. By Lemma \ref{lem:w_z}, we then have that $w_z(t,b_\delta(t))>0$ and thus, by the implicit function theorem, we obtain that $t\mapsto b_\delta(t)$ is continuously differentiable on $(t^*,T)$ with
\begin{equation}
	\label{bprimo}
	\dot b_{\delta}(t)=-\frac{w_t(t,b_{\delta}(t))}{w_z(t,b_{\delta}(t))}, \quad t\in (t^*,T).
\end{equation}
 
By equation \eqref{bprimo}, Lemma \ref{lem:w_z} and Lemma \ref{lem:w_tUpBound}, we obtain the next lower bound. Here, we recall $\lambda$ from Assumption \ref{A1}, $c_0=\alpha e^{|\alpha|T}$ from the proof of Lemma \ref{lem:w_tUpBound} and denote $c_1=\lambda+\mu_M$ with $\mu_M=\max_{s\in[0,T]}\mu(s)$.
\begin{lemma}\label{lem:b'_delta}
	If $t^*\in[0,T)$, for $\Lambda:=c+c_0+c_1-f(T)$ we have
	\[
	\inf_{\delta\in(0,1)}\inf_{t\in(t^*,T)}\bigg(\frac{\dot b_\delta(t)}{b_\delta(t)}\bigg)\geq -\Lambda.
	\]
	\end{lemma}
\begin{proof}
Fix $(t,z)\in\cC$ with $t\in(t^*,T)$ and recall $f(s)\le 0$ for $s\in[t^*,T)$. The first two lines on the right-hand side of \eqref{eq:w_tleq} can be bounded from above using slightly different arguments, depending on whether (i) or (ii) in Assumption \ref{A2} hold. If (i) holds we have  
\begin{align*}
\begin{aligned}
&\E\Big[\int_0^{\tau^*_{t,z}}e^{-cs}{_s p_t}\big(\mu(t)-\mu(t+s)\big)\big(f(t+s)+\mu(t+s)(Z^z_s-1)^+\big)\ud s\Big]\\
&\quad+\E\Big[\int_0^{\tau^*_{t,z}}e^{-cs}{_s p_t}\big(\dot f(t+s)+\dot \mu(t+s)(Z^z_s-1)^+\big)\ud s\Big]\\
&\le \E\Big[\int_0^{\tau^*_{t,z}}e^{-cs}{_s p_t}\big(\mu(t)-\mu(t+s)\big)\mu(t+s)(Z^z_s-1)^+\ud s\Big]\\
&\le \mu_M \E\Big[\int_0^{\tau^*_{t,z}}e^{-cs}{_s p_t}\mu(t+s)(Z^z_s-1)^+\ud s\Big],
\end{aligned}
\end{align*}
where $\mu_M=\max_{s\in[0,T]}\mu(s)$. If instead (ii) in Assumption \ref{A2} holds, we argue as follows:
\begin{align*}
\begin{aligned}
&\E\Big[\int_0^{\tau^*_{t,z}}e^{-cs}{_s p_t}\big(\mu(t)-\mu(t+s)\big)\big(f(t+s)+\mu(t+s)(Z^z_s-1)^+\big)\ud s\Big]\\
&\quad+\E\Big[\int_0^{\tau^*_{t,z}}e^{-cs}{_s p_t}\big(\dot f(t+s)+\dot \mu(t+s)(Z^z_s-1)^+\big)\ud s\Big]\\
&\le \E\Big[\int_0^{\tau^*_{t,z}}e^{-cs}{_s p_t}\big(\dot f(t+s)-\mu(t+s)f(t+s)+\dot \mu(t+s)(Z^z_s-1)^+\big)\ud s\Big]\\
&\le \lambda \E\Big[\int_0^{\tau^*_{t,z}}e^{-cs}{_s p_t}\mu(t+s)(Z^z_s-1)^+\ud s\Big],
\end{aligned}
\end{align*}
where in the first inequality we dropped the negative terms 
\[
\mu(t)f(t+s)\quad\text{ and }\quad(\mu(t)-\mu(t+s))\mu(t+s)(Z^z_s-1)^+
\]
and in the second inequality we used $\dot \mu(s)\le \lambda\mu(s)$ by Assumption \ref{A1} and (ii) in Assumption \ref{A2}. Combining the two estimates above with \eqref{eq:w_tleq} and setting $c_1=\lambda+\mu_M$, we find an upper bound for $w_t$ as
\begin{align}\label{eq:w_tleq2}
\begin{aligned}
w_t(t,z) &\leq c_1\E\Big[\int_0^{\tau^*_{t,z}}e^{-cs}{_s p_t}\mu(t+s)(Z^z_s-1)^+\ud s\Big]\\
&\quad+\E\Big[\ind_{\{\tau^*_{t,z}=T-t \}}e^{-c(T-t)}{_{T-t} p_t}\big(c(Z^z_{T-t}-1)^+ - f(T)\big) \Big]\\
&\quad+c_0\E\Big[\ind_{\{\tau^*_{t,z}=T-t \}}e^{-c(T-t)}{_{T-t} p_t} Z^z_{T-t}\Big].
\end{aligned}	
\end{align}
Next, we use $(Z^z_s-1)^+\le Z^z_s \ind_{\{Z^z_s>1\}}$ for all $s\in[0,T-t]$, to obtain
	\begin{align}\label{eq:w_tleq3}
	\begin{aligned}
		w_t(t,z)&\leq \E\Big[\int_0^{\tau^*_{t,z}}c_1 e^{-cs}{_s p_t}\mu(t+s)\ind_{\{Z^z_s>1 \}}Z_s\ud s +(c+c_0)e^{-c(T-t)}{_{T-t} p_t}\ind_{\{\tau^*_{t,z}=T-t\}}Z^z_{T-t}\Big]\\
		&\quad- f(T)e^{-c(T-t)}{_{T-t} p_t}\P\big(\tau^*_{t,z}=T-t\big).
\end{aligned}	
	\end{align}
The first two terms on the right-hand side of the expression above are the same as in \eqref{eq:zw_z} up to a multiplicative constant. For the last term above we notice
\begin{align}\label{eq:fT}
\begin{aligned}
- f(T)e^{-c(T-t)}{_{T-t} p_t}\P\big(\tau^*_{t,z}=T-t\big)
&= |f(T)|e^{-c(T-t)}{_{T-t} p_t}\E\big[\ind_{\{\tau^*_{t,z}=T-t\}\cap\{Z^z_{T-t}\ge 1\}}\big]\\
&\le |f(T)|e^{-c(T-t)}{_{T-t} p_t}\E\big[Z^z_{T-t}\ind_{\{\tau^*_{t,z}=T-t\}\cap\{Z^z_{T-t}\ge 1\}}\big]\\
&= |f(T)|e^{-c(T-t)}{_{T-t} p_t}\E\big[Z^z_{T-t}\ind_{\{\tau^*_{t,z}=T-t\}}\big],
\end{aligned}
\end{align}
thanks to Corollary \ref{lem:Z>1}. Thus, the last term in \eqref{eq:fT} is the same as the second term in \eqref{eq:zw_z}, up to a multiplicative constant. Then, combining \eqref{eq:w_tleq3} and \eqref{eq:zw_z}, we obtain
\[
\frac{\dot b_\delta(t)}{b_\delta(t)}=-\frac{ w_t(t,b_\delta(t))}{b_\delta(t) w_z(t,b_\delta(t))}\geq - c_1- c-c_0+ f(T)=:-\Lambda.
\]
The bound is uniform in $t\in[0,T]$ and $\delta>0$, hence concluding the proof.
\end{proof}

We can define a transformation of the boundary:
\begin{equation}\label{eq:beta}
	\beta(t):=e^{\Lambda t}b(t), \qquad t\in[0,T],
\end{equation}
where $\Lambda$ comes from Lemma \ref{lem:b'_delta}. 
Next we deduce some useful properties of the mapping $t\mapsto \beta(t)$.
\begin{proposition}\label{prop:betamonot}
	The transformed boundary $t\mapsto\beta(t)$ is non-decreasing and right-continuous, with $0\le \beta(t)\le e^{\Lambda T}$ for all $t\in[0,T]$. Moreover, $\beta(t)=0$ for $t\in[0,t^*)$ and $\beta(t)>0$ for $t\in(t^*,T]$. Finally, if $t^*\in[0,T)$, then $\lim_{t\uparrow T}\beta(t)=e^{\Lambda T}$.
\end{proposition}
\begin{proof}
Since $b(t)=0$ for every $t\in[0,t^*)$ (cf.\ \eqref{eq:sC}), it immediately follows from \eqref{eq:beta} that also $\beta(t)=0$ for every $t\in[0,t^*)$. For concreteness and with no loss of generality we take $t^*\in[0,T)$. For every $\delta>0$ and $t^*< s\leq t\leq T$, by Lemma \ref{lem:b'_delta}, we obtain
\[
\log b_\delta(t)-\log b_\delta(s)=\int_s^t \frac{\dot b_\delta(\theta)}{b_\delta(\theta)}\ud \theta\geq -\Lambda (t-s).
\]
Hence, $b_\delta(t)\geq b_\delta(s) \exp(-\Lambda(t-s) )$ and, by letting $\delta\to 0$, we obtain
\begin{equation}\label{eq:b>0}
	\beta(t):=e^{\Lambda t}b(t)\geq b(s)e^{\Lambda s}=:\beta(s),
\end{equation}
which yields monotonicity of $\beta$. Then, the limit $\beta(T):=\lim_{t\uparrow T}\beta(t)$ is well-defined and it must be $\beta(T)=e^{\Lambda T}$, due to Proposition \ref{prop:limsupT}. Combining with monotonicity we obtain $0\le \beta(t)\le e^{\Lambda T}$ for all $t\in[0,T]$.

Next we want to show right-continuity of $\beta$. Fix $t\in[0,T)$ and let $(t_n)_{n\in\bN}\subseteq(t,T)$ be such that $t_n\downarrow t$ as $n\to \infty$. By monotonicity of $\beta$ the right-limit $\beta(t+)=\lim_{n\to\infty}\beta(t_n)$ is well-defined. Since $b(t)=\beta(t)e^{-\Lambda t}$, then also the right-limit $b(t+)=\lim_{n\to\infty}b(t_n)$ is well-defined. Now, $\cD\ni(t_n,b(t_n))\to (t,b(t+))\in\cD$ as $n\to\infty$ because $\cD$ is closed ($w$ is continuous). Then, $(t,b(t+))\in\cD$ implies $b(t+)\leq b(t)$. That translates into $\beta(t+)\le \beta(t)$ and, since monotonicity guarantees $\beta(t+)\geq\beta(t)$, we conclude $\beta(t+)=\beta(t)$.
	
It only remains to prove that $\beta(t)>0$ for all $t\in(t^*,T]$. Assume by contradiction that there exists $t_1\in(t^*,T]$ such that $\beta(t_1)=0$. By monotonicity we then have $\beta(t)=0$ (and therefore $b(t)=0$) for every $t\in[t^*,t_1]$. The latter implies that for any $t\in[t^*,t_1)$ and $z>0$, we have $\tau^*_{t,z}\geq t_1-t$, $\P$-a.s. Thus, using $f(s)< 0$ for $s\in(t^*,T]$ we obtain
	\begin{align*}
		w(t,z)&= \E\Big[\int_0^{\tau^*_{t,z}}e^{-cs}{_s p_t} f(t+s)\ud s+\int_0^{\tau^*_{t,z}}e^{-cs} {_s p_t}\mu(t+s)\big(Z^z_s-1\big)^+\ud s\Big]\\
		&\quad+\E\Big[\ind_{\{\tau^*_{t,z}=T-t \}} e^{-c(T-t)} {_{T-t}p_t}\big(Z^z_{T-t}-1\big)^+\Big]\\
		&\leq \int_0^{t_1-t} e^{-cs}\, _sp_t f(t+s)\ud s+ \mu_M (T-t)z\E\Big[\sup_{0\leq s\leq T-t}Z^1_s\Big]+z\E\big[Z^1_{T-t}\big],
	\end{align*}
	where $\mu_M:=\sup_{s\in[0,T]}\mu(s)$. Since the first term on the right-hand side of the inequality is strictly negative and independent of $z$, we can choose $z>0$ sufficiently small to guarantee $w(t,z)<0$. That is impossible, because $w\ge 0$. Hence, we reached a contradiction and it must be $\beta(t_1)>0$ (equivalently $b(t_1)>0$).
\end{proof}

\begin{remark}
\label{remarkt*}
The above proposition implies, in particular, that $b(t)>0$ for every $t\in(t^*,T]$. The latter indicates that an insurer can totally remove the incentive for an early surrender {\em if and only if} the penalty function $k(t)$ is such that $t^*\ge T$. That corresponds to $f(t)\geq 0$ for every $t\in[0,T]$, according to \eqref{ftilde}. 

For the benchmark penalty function \eqref{kt} we have
$f(t)=\mu(t)k(t)+(K-c)e^{-K(T-t)}$. Thus, early surrender is never optimal {\em if and only if} $K\ge c$. In other words, the surrender option is worthless, unless $K<c$. 
\end{remark}

Left-continuity of $t\mapsto\beta(t)$, for $t\neq t^*$, will be obtained in Section \ref{sect:bCont}. For that, we first need to prove smoothness of the value function in the next section. Left-continuity of $t\mapsto\beta(t)$ will in turn imply its continuity and, thus, also continuity of $t\mapsto b(t)$ (Theorem \ref{thm:bCont}). In Figure \ref{fig2} we plot the boundaries $b(\cdot)$ and $\beta(\cdot)$ that correspond to $\ell_{x_0}(\cdot)$ on the left panel of Figure \ref{fig4} for $c=3.5\%$ (recall the relation $\ell_{x_0}(t)=x_0e^{gt}/b(t)$). We can see that although $b$ is not monotonic $\beta$ is increasing, in keeping with our theoretical analysis. 

\begin{figure}[ht]
\centering
\includegraphics[scale=0.5]{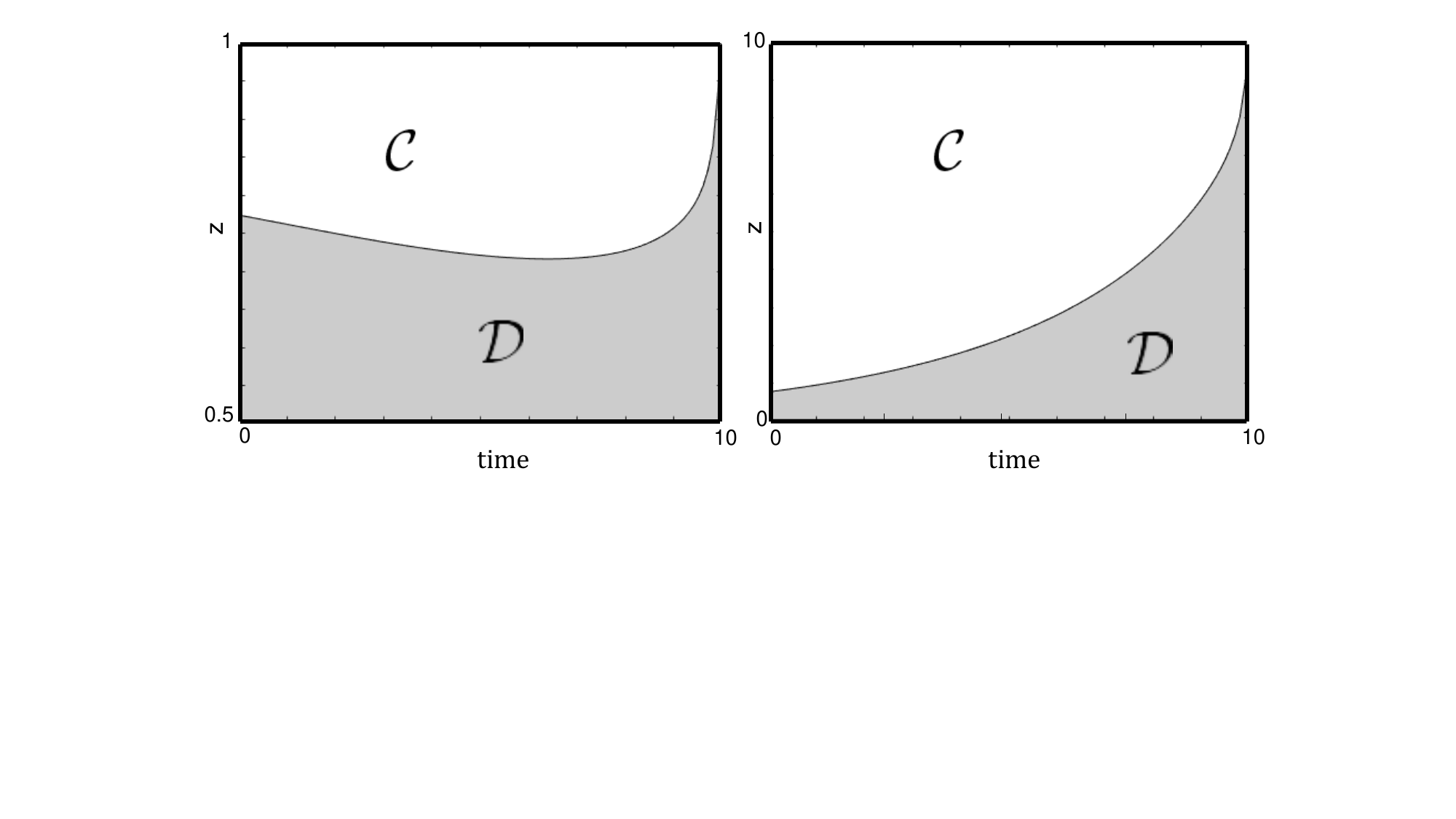}
\vspace{-4.5cm}
\caption{The optimal surrender boundaries $t\mapsto b(t)$ (left plot) and $t\mapsto \beta(t)$ (right plot).}
\label{fig2}
\end{figure}

\section{Continuous differentiability of the value function}\label{sec:wSmooth}

Given a set $A\subseteq[0,T]\times(0,\infty)$, let $\bar A=\mathrm{cl}(A)$ denote its closure. We say that a function $f:[0,T]\times(0,\infty)\to \R$ belongs to $C(\bar A)\cap C(\bar A^c)$ if the function is continuous separately on $\bar A$ and on $\bar A^c:=\mathrm{cl}(A^c)$, but it may be discontinuous at the boundary $\partial A$.

For $\eps>0$, let 
\begin{align}\label{eq:CepsDeps}
\bar{\cC_\eps}:=\bar{\cC}\cap\big([0,T-\eps]\times(0,\infty)\big)\quad\text{ and }\quad\cD_\eps:=\cD\cap\big([0,T-\eps]\times(0,\infty)\big).
\end{align}
By means of monotonicity and right-continuity of $t\mapsto\beta(t)$, in this section we show that 
\[
w\in C^1\big([0,T)\times(0,\infty)\big) \quad \text{and} \quad w_{zz}\in C\big(\bar{\cC_\eps}\big)\cap C(\cD_\eps),
\]
for every $\eps>0$ (cf.\ Proposition \ref{prop:wC1}). 
Lemma \ref{lem:w_tUpBound} provides an upper bound for $w_t$. In Lemma \ref{lem:w_tLowBound} we are going to obtain a lower bound but first we need a preliminary estimate in Lemma \ref{lem:EuAmCall}. We are going to use $\Phi(z,T-t,\cdot)$ for the probability density function of $Z^z_{T-t}$ (cf.\ \eqref{eq:Phi}) and we recall the notation for $D^t(s)$ from \eqref{eq:UD}. 

For $\eps\in(0,\frac{T-t}{2})$, let us define
\begin{align}\label{eq:theta}
\begin{aligned}
\Theta(t,z,\eps)&:=\E\Big[\ind_{\{\tau^*_{t,z}\geq T-t-\eps\}}D^{t+\eps}(T\!-\!t\!-\!\eps)\big(Z^z_{T-t-\eps}-1\big)^+\Big]\\
	&\quad-\E\Big[\ind_{\{\tau^*_{t,z}= T-t\}}D^{t}(T\!-\!t)\big(Z^z_{T-t}-1\big)^+\Big].
\end{aligned}
\end{align}	
The next lemma gives us a lower bound on $\Theta$.
\begin{lemma}\label{lem:EuAmCall}
For every $(t,z)\in[0,T)\times(0,\infty)$ and $\eps\in(0,\frac{T-t}{2})$, we have $\Theta(t,z,\eps)\ge -c(t,z)\eps$,
for some $c(t,z)\ge 0$ that can be chosen bounded in $[0,T\!-\!n^{-1}]\times[n^{-1},n]$ for any $n\ge 1/T$.
Moreover, in the limit	
	\begin{align}\label{eq:liminftheta}
	\liminf_{\eps\to 0}\big[\eps^{-1}\Theta(t,z,\eps)\big]\geq -D^{t}(T\!-\!t)\Big(\tfrac{\sigma^2}{2}\Phi(z,T-t,1)+|\alpha|e^{|\alpha|(T-t)}z\Big).
	\end{align}
\end{lemma}
\begin{proof}
	Let $\eps\in(0,\frac{T-t}{2})$ and, to simplify notation, let  
\begin{align*}
h_t(\eps):=e^{-\int_0^{T-t-\eps}\mu(t+\eps+\theta)\ud\theta} \quad\text{and}\quad t_\eps:=t+\eps.
\end{align*}

Notice that $e^{-c(T-t)}h_t(0)=D^t(T-t)$ and the second term in the expression \eqref{eq:theta} can be bounded from above as follows: 
\begin{align*}
\E\Big[\ind_{\{\tau^*_{t,z}= T-t\}}e^{-c(T-t)}h_t(0)\big(Z^z_{T-t}-1\big)^+\Big]\leq \E\Big[\ind_{\{\tau^*_{t,z}\geq T-t_\eps\}}e^{-c(T-t)}h_t(0)\big(Z^z_{T-t}-1\big)^+\Big].
\end{align*}
Further, denoting $\E_{t,z}[\cdot]=\E[\cdot|Z_t=z]$ and using the Markov property and time-homogeneity of $Z$ we obtain
\begin{align*}
&\E\Big[\ind_{\{\tau^*_{t,z}\geq T-t_\eps\}}e^{-c(T-t)}h_t(0)\big(Z^z_{T-t}-1\big)^+\Big]\\
&=\E\Big[\ind_{\{\tau^*_{t,z}\geq T-t_\eps \}}e^{-c(T-t_\eps)}h_t(0)\E\big[e^{-c\eps}\big(Z^{z}_{T-t}-1\big)^+\big|\cF_{T-t_\eps}\big]\Big]\\
&=\E\Big[\ind_{\{\tau^*_{t,z}\geq T-t_\eps \}}e^{-c(T-t_\eps)}h_t(0)\E_{T-t_\eps,Z^z_{T-t_\eps}}\big[e^{-c\eps}\big(Z_{T-t}-1\big)^+\big]\Big]\\
&=\E\Big[\ind_{\{\tau^*_{t,z}\geq T-t_\eps \}}e^{-c(T-t_\eps)}h_t(0)\E_{0,Z^z_{T-t_\eps}}\big[e^{-c\eps}\big(Z_{\eps}-1\big)^+\big]\Big]\\
&= \E\Big[\ind_{\{\tau^*_{t,z}\geq T-t_\eps \}}e^{-c(T-t_\eps)}h_t(0)v^E(T-\eps,Z^z_{T-t_\eps}) \Big],
\end{align*}
where
\[
v^E(s,\lambda):=\E\Big[e^{-c(T-s)}\big(Z^\lambda_{T-s}-1\big)^+\Big],
\]
is the price of a European call option with discount rate $c$, maturity $T$, strike price $K=1$ and underlying stock price $Z$. Using the bounds and notations above
\begin{align*}
\begin{aligned}
\Theta(t,z,\eps)&\ge \E\Big[\ind_{\{\tau^*_{t,z}\ge T-t_\eps\}}e^{-c(T-t_\eps)}h_t(\eps)\big(Z^z_{T-t_\eps}-1\big)^+\Big]\\
&\quad-\E\Big[\ind_{\{\tau^*_{t,z}\geq T-t_\eps \}}e^{-c(T-t_\eps)}h_t(0)v^E(T-\eps,Z^z_{T-t_\eps}) \Big],
\end{aligned}
\end{align*}
where we also noticed that $e^{-c(T-t_\eps)}h_t(\eps)=D^{t_\eps}(T-t_\eps)$.

We now add and subtract
\[
e^{-c(T-t_\eps)}h_t(0)\E\Big[\ind_{\{\tau^*_{t,z}\geq T-t_\eps\}}\big(Z^z_{T-t_\eps}-1\big)^+\Big],
\]
on the right-hand side above. That yields,
	\begin{align}\label{eq:EuCall}
	\begin{aligned}
\Theta(t,z,\eps)	&\ge e^{-c(T-t_\eps)}\Big(h_t(\eps)-h_t(0)\Big)\E\Big[\ind_{\{\tau^*_{t,z}\geq T-t_\eps\}}\big(Z^z_{T-t_\eps}-1\big)^+\Big]\\
	&+e^{-c(T-t_\eps)}h_t(0)\E\Big[\ind_{\{\tau^*_{t,z}\geq T-t_\eps\}}\big((Z^z_{T-t_\eps}-1)^+-v^E(T-\eps,Z^z_{T-t_\eps}) \big)\Big]\\
	&\geq e^{-c(T-t_\eps)}h_t(0)\E\Big[\ind_{\{\tau^*_{t,z}\geq T-t_\eps\}}\big((Z^z_{T-t_\eps}-1)^+-v^E(T-\eps,Z^z_{T-t_\eps}) \big)\Big],
\end{aligned}	
	\end{align}
	where for the final inequality we use the simple observation $h_t(\eps)\geq h_t(0)$. To simplify the rest of the notation we set $\bar \Theta(t,z,\eps)=e^{c(T-t_\eps)}(h_t(0))^{-1}\Theta(t,z,\eps)$.
	Since the price of the American call option is greater than or equal to the price of the European call option, we obtain
	\begin{align}\label{eq:lbt0}
	\bar \Theta(t,z,\eps)\geq \E\Big[\ind_{\{\tau^*_{t,z}\geq T-t_\eps\}}\big((Z^z_{T-t_\eps}-1)^+-v^A(T-\eps,Z^z_{T-t_\eps}) \big)\Big],
	\end{align}
	where
\[
v^A(s,\lambda):=\sup_{0\leq \tau\leq T-s}\E\Big[e^{-c\tau}\big(Z^\lambda_\tau-1\big)^+\Big].
\]
Since $v^A(s,\lambda)\geq (\lambda-1)^+$, we can drop the indicator function on the right-hand side of \eqref{eq:lbt0}. Then 
	\begin{align}\label{eq:lbt1}
		\begin{aligned}
\bar \Theta(t,z,\eps)&\geq\E\Big[(Z^z_{T-t_\eps}-1)^+-v^A(T-\eps,Z^z_{T-t_\eps}) \Big]\\
	&=\int_0^\infty \big[(\lambda-1)^+-v^A(T-\eps,\lambda)\big]\Phi(z,T-t_\eps,\lambda)\ud\lambda\\
	&=\int_0^\infty \big[v^A(T,\lambda)-v^A(T-\eps,\lambda)\big]\Phi(z,T-t_\eps,\lambda)\ud\lambda,
\end{aligned}	
	\end{align}
	where $\lambda\mapsto \Phi(z,T-t_\eps,\lambda)$ is the probability density function of $Z_{T-t_\eps}^z$, and the last equality holds because $v^A(T,\lambda)=(\lambda-1)^+$. 
	
	It is shown in \cite[Example 17]{deape2020} that the value function of the American put option has continuous derivatives on $[0,T)\times(0,\infty)$. By the same arguments, it follows that $v^A(t,z)\in C^1([0,T)\times(0,\infty))$. Then, if $s\mapsto b^A(s)$ denotes the non-increasing (possibly infinite) optimal stopping boundary for the American call option, we obtain $v_t^A(s,\lambda)=0$ for every $\lambda\geq b^A(s)$. Therefore, continuing from \eqref{eq:lbt1}, we have 
	\begin{align*}
\bar \Theta(t,z,\eps)&\geq \lim_{n\to\infty}\int_0^\infty \big[v^A(T-\eps\wedge\tfrac{1}{n},\lambda)-v^A(T-\eps,\lambda)\big]\Phi(z,T-t_\eps,\lambda)\ud\lambda\\
	&=\lim_{n\to\infty}\int_0^\infty\int_{T-\eps}^{T-\eps\wedge\tfrac{1}{n}} v^A_t(s,\lambda) \Phi(z,T-t_\eps,\lambda)\ud s\ud\lambda\\
	&=\lim_{n\to\infty}\int_{T-\eps}^{T-\eps\wedge\tfrac{1}{n}}\int_0^{b^A(s)} v^A_t(s,\lambda) \Phi(z,T-t_\eps,\lambda)\ud\lambda\ud s,
	\end{align*}
where the limit is taken out of the integral in the first line by the monotone convergence theorem, and the order of integration can be swapped by Tonelli's theorem because $v_t^A\le 0$. In particular, the monotone convergence theorem can be applied because the sequence of functions in the first integral above is non-increasing in $n$ and non-positive.
 
By the connection between optimal stopping problems and free-boundary problems (cf.\ \cite[Ch.\ III]{peskir2006optimal}) it is well-known that $v_t^A(s,\lambda)=(cv^A-\cL v^A)(s,\lambda)$, for all $(s,\lambda)\in[0,T)\times(0,\infty)$ such that  $\lambda< b^A(s)$,
where $\cL$ is given in \eqref{eq:L}. Continuing from the previous inequality, 	
\begin{align*}
	\bar \Theta(t,z,\eps)&\geq \lim_{n\to\infty}\int_{T-\eps}^{T-\eps\wedge\tfrac{1}{n}}\int_0^{b^A(s)} (cv^A-\cL v^A)(s,\lambda) \Phi(z,T-t_\eps,\lambda)\ud\lambda\ud s.
	\end{align*}
Since we make no assumption on the sign of $c-\alpha$ it may be that $b^A(s)\equiv \infty$ (i.e., $v^A=v^E$).

First we consider the case when $b^A(t)<\infty$ for $t\in[0,T]$. 	
Let $\cL^*$ be the adjoint operator of $\cL$, i.e., $\cL^*\varphi(\lambda):=\frac{\ud^2}{\ud \lambda^2}\big[\tfrac{\sigma^2\lambda^2}{2}\varphi(\lambda) \big]- \frac{\ud}{\ud\lambda}\big[\alpha\lambda\varphi(\lambda) \big]$, for $\varphi\in C^2((0,\infty))$. 	
Applying integration by parts to the integral on the right-hand side above and letting $n\to\infty$, we obtain
\begin{align}\label{eq:IntByParts}
\begin{aligned}
\bar \Theta(t,z,\eps)&\ge \int_{T-\eps}^T\big[1-b^A(s)\big]\psi(z,T-t_\eps,b^A(s))\ud s-\int_{T-\eps}^T \tfrac{\sigma^2}{2} \big[b^A(s)\big]^2 \Phi(z,T-t_\eps,b^A(s))\ud s\\
	&\quad+\int_{T-\eps}^T\int_0^{b^A(s)} v^A(s,\lambda)(c-\cL^*)\Phi(z,T-t_\eps,\lambda)\ud\lambda\ud s,
	\end{aligned}
	\end{align}
	where
\[
\psi(z,s,\lambda):=(\alpha-\sigma^2)\lambda \Phi(z,s,\lambda)-\tfrac{\sigma^2\lambda^2}{2}\Phi_\lambda(z,s,\lambda),
\]
	and we have also used the boundary conditions\footnote{In the case $b^A(s)\equiv\infty$, the first two terms on the right-hand side of  \eqref{eq:IntByParts} vanish because of the exponential decay of the log-normal density. Indeed, for any $p\ge 0$, $\lambda^p \Phi(z,s,\lambda)\to 0$ as $\lambda\to\infty$ and $|v^A_z(s,\lambda)|+|v^A(s,\lambda)|\le c(1+\lambda)$, for any $(s,z)\in [0,T)\times(0,\infty)$.} $v^A(s,b^A(s))=b^A(s)-1$ and $v^A_z(s,b^A(s))=1$. 
	For $\eps\in(0,\frac{T-t}{2})$ it is not difficult to show that the right-hand side of \eqref{eq:IntByParts} is greater than $-c(t,z)\eps$ for some positive constant $c(t,z)>0$, using the explicit form of the log-normal density and the fact that $t<T$.
It is also easy to verify that $c(t,z)\ge 0$ can be chosen bounded in any subset $[0,T-n^{-1}]\times[n^{-1},n]$, $n\ge 1/T$, thus proving the first claim in the lemma.

We now divide \eqref{eq:IntByParts} by $\eps$ and let $\eps\to 0$. Using that $b^A(T)=1$ and $v^A(T,\lambda)=(\lambda-1)^+$, we obtain
\begin{align*}
\liminf_{\eps\to 0}\tfrac{1}{\eps}\bar \Theta(t,z,\eps)\ge -\tfrac{\sigma^2}{2}\Phi(z,T-t,1).
\end{align*}

In the case $b^A(s)=+\infty$ for $s\in[0,T]$ (hence $v^A=v^E$) we use a similar argument with integration by parts combined with the observation that 
\[
\lim_{s\uparrow T}\big(cv^E-\cL v^E\big)(s,\lambda)=\big((c-\alpha)\lambda-c\big) \ind_{[1,\infty)}(\lambda)-\frac{\sigma^2 \lambda^2}{2}\delta_1(\lambda),
\]
in the sense of distributions and with $\delta_1(z)$ denoting the Dirac's delta centred in one. Then
\begin{align*}
	\liminf_{\eps\to 0}\eps^{-1}\bar \Theta(t,z,\eps)&\geq \liminf_{\eps\to 0}\frac{1}{\eps}\int_{T-\eps}^{T}\int_0^{\infty} (cv^A-\cL v^A)(s,\lambda) \Phi(z,T-t_\eps,\lambda)\ud\lambda\ud s\\
	&=\int_0^{\infty} \Big(\big[(c-\alpha)\lambda-c\big] \ind_{[1,\infty)}(\lambda)-\frac{\sigma^2 \lambda^2}{2}\delta_1(\lambda)\Big) \Phi(z,T-t,\lambda)\ud\lambda\\
	&\ge -\alpha\E\big[Z^z_{T-t}\big]-\tfrac{\sigma^2}{2}\Phi(z,T-t,1).
	\end{align*}
Recalling that $\liminf_{\eps\to 0}\eps^{-1}\bar \Theta(t,z,\eps)=e^{c(T-t)}(h_t(0))^{-1}\liminf_{\eps\to 0}\eps^{-1}\Theta(t,z,\eps)$ we conclude the proof of \eqref{eq:liminftheta}.
\end{proof}

Recalling the definitions of $H$ in \eqref{H} and $D^t(s)$ in \eqref{eq:UD}, we are now able to produce a lower bound for the time-derivative $w_t$.

\begin{lemma}\label{lem:w_tLowBound}
	For every $(t,z)\in[0,T)\times(0,\infty)\setminus\partial\cC$ we have 
	\begin{align}\label{eq:w_tgeq}
	\begin{aligned}
	w_t(t,z) &\geq \E\Big[\int_0^{\tau^*_{t,z}}D^t(s)\big(\mu(t)-\mu(t+s)\big)H(t+s,Z^z_s)\ud s\Big]\\
	&\quad+\E\Big[\int_0^{\tau^*_{t,z}}D^t(s)\big( \dot f(t+s)+\dot \mu(t+s)(Z^z_s-1)^+\big)\ud s\Big]\\
	&\quad-\E\big[D^t(T-t)|H(T,Z^z_{T-t})|\big]-D^t(T-t)\Big(\tfrac{\sigma^2}{2}\Phi(z,T-t,1)+|\alpha|e^{|\alpha|(T-t)}z\Big).
	\end{aligned}
	\end{align}
		In particular, for any $\eps>0$,
	\begin{equation}\label{eq:q}
	w_t(t,z)\geq -q(\eps)(1+z), \quad \forall \: (t,z)\in\cC\cap\big([0,T-\eps]\times(0,\infty)\big),
	\end{equation}
	for some constant $q(\eps)>0$ that may explode as $\eps\to 0$.
\end{lemma}
\begin{proof}
First we show that \eqref{eq:w_tgeq} implies \eqref{eq:q}. By definition of $H$ and continuity of $\mu(t)$, $\dot \mu(t)$, $f(t)$, $\dot f(t)$ (recall Assumption \ref{A1}), it is immediate to check that the absolute values of the integrands in the first two lines of \eqref{eq:w_tgeq} are bounded from above by $c_3(1+Z^z_s)$ for some uniform constant $c_3>0$. Similarly the third term can be bounded from above with $c_3(1+Z^z_{T-t})$. Then, taking expectations and changing $c_3>0$ to a different constant we have that the first three terms on the right-hand side of \eqref{eq:w_tgeq} are bounded from below by $-c_3(1+z)$. The explicit form of the log-normal density allows us to write $0\le \Phi(z,T-t,1)\le \tilde q(T-t)$, where $\tilde q(s)\uparrow \infty$ if $s\to 0$. Then, in conclusion
\[
w_t(t,z)\geq -\big(c_3+\tfrac{\sigma^2}{2}\tilde q(T-t)\big)(1+z)
\]
and \eqref{eq:q} follows with $q(\eps)=c_3+|\alpha|e^{|\alpha|T}+\frac{\sigma^2}{2}\tilde q(\eps)$.

It remains to prove \eqref{eq:w_tgeq}. The claim is trivial for $(t,z)\in\text{int}(\cD)$ because $\tau^*_{t,z}=0$ in that case and $w_t(t,z)=0$. Let $(t,z)\in\cC$ and $\eps\in(0,T-t)$. Set $t_\eps=t+\eps$. Since $\tau^*_{t,z}\wedge(T\!-\!t_\eps)$ is sub-optimal but admissible for the problem starting at $(t_\eps,z)$, we have
	\begin{align*}
\begin{aligned}
	&w(t+\eps,z)-w(t,z)\\
	&\geq\E\Big[\int_0^{\tau^*_{t,z}\wedge(T-t_\eps)}\Big(D^{t_\eps}(s)H(t_\eps+s,Z^z_s)-D^t(s)H(t+s,Z_s^z)\Big)\ud s \Big]\\
	&\quad-\E\Big[\ind_{\{\tau^*_{t,z}\geq T-t_\eps \}}\int_{T-t_\eps}^{\tau^*_{t,z}}e^{-cs-\int_0^s\mu(t+\theta)\ud\theta}H(t+s,Z^z_s)\ud s\Big]\\
	&\quad+\E\Big[\ind_{\{\tau^*_{t,z}\geq T-t_\eps\}}D^{t_\eps}(T-t_\eps)\big(Z^z_{T-t_\eps}-1\big)^+-\ind_{\{\tau^*_{t,z}= T-t\}}D^t(T-t)\big(Z^z_{T-t}-1\big)^+\Big].
\end{aligned}	
	\end{align*}
Notice that in the last line of the expression above we have exactly $\Theta(t,z,\eps)$ defined in \eqref{eq:theta}. Moreover, in the second line we can replace $H(t+s,Z^z_s)$ with $|H(t+s,Z^z_s)|$ and keep the same inequality by virtue of the minus sign in front of the expression. Then, dividing both sides by $\eps$, letting $\eps\to0$ and using Lemma \ref{lem:EuAmCall}, we obtain \eqref{eq:w_tgeq}. 
\end{proof}

In order to prove continuous differentiability of the value function we need the next lemma. Its proof requires the introduction of a new coordinate system and we postpone it at the end of the section for the ease of exposition.
\begin{lemma}\label{lem:tau_nTo0}
	Let $(t,z)\in\partial\cC$ with $t<T$ and let $(t_n,z_n)_{n\in\bN}\subseteq\cC$ be such that $\lim_{n\to\infty}(t_n,z_n)=(t,z)$. Then, $\lim_{n\to\infty}\tau^*_{t_n,z_n}=0$, $\P$-a.s.
\end{lemma}

Recall the closed sets $\bar{\cC_\eps}$ and $\cD_\eps$ from \eqref{eq:CepsDeps}. The next proposition is the main result of this section.
\begin{proposition}\label{prop:wC1}
For every $\eps>0$, we have $w\in C^1\big([0,T)\times(0,\infty)\big)$ and $w_{zz}\in C\big(\bar{\cC_\eps}\big)\cap C(\cD_\eps)$.
\end{proposition}
\begin{proof}
	We know, from \eqref{eq:fbp}, that $w_z$, $w_{zz}$ and $w_t$ are continuous on $\cC$. Moreover, $w(t,z)=0$ for $(t,z)\in\cD$ implies that $w_z(t,z)=w_{zz}(t,z)=w_t(t,z)=0$ for $(t,z)\in\cD\setminus\partial\cC$ with $t<T$. 
		Therefore, it remains to prove that $w_z$ and $w_t$ are continuous across the boundary $\partial\cC$. 
	
	Let $(t_0,z_0)\in\partial\cC$ with $t_0<T$ and take a sequence $(t_n,z_n)_{n\geq 1}\subseteq\cC$ converging to $(t_0,z_0)$ as $n\to\infty$. Then, by Lemma \ref{lem:tau_nTo0}, we have that $\lim_{n\to\infty} \tau^*_{t_n,z_n}=0$, $\P$-a.s. Hence, from \eqref{eq:w_z} and by using the dominated convergence theorem (which can be used by boundedness of the functions and the fact that $Z$ is a geometric Brownian motion), we obtain $\lim_{n\to\infty}w_z(t_n,z_n)=0$. Since both $(t_0,z_0)$ and the sequence $(t_n,z_n)$ were arbitrary, we conclude that $w_z\in C\big([0,T)\times(0,\infty)\big)$. For the continuity of $w_t$ across the boundary, using a similar argument as above and Lemma \ref{lem:w_tUpBound}, we obtain
	\begin{equation}\label{eq:w_t<0}
	\limsup_{n\to\infty}w_t(t_n,z_n)\leq 0.
	\end{equation}
	
To show the opposite inequality we first use Lemma \ref{lem:w_tLowBound} to obtain a refined lower bound for $w_t$. Let $(t,z)\in\cC$ with $t<T$ and let $T_1\in(t,T)$. Set $\rho^*_{t,z}:=\tau^*_{t,z}\wedge(T_1-t)$. By the martingale property of the value function (cf.\ \eqref{eq:UD}), we have
	\begin{equation*}
	w(t,z)=\E\Big[\int_0^{\rho^*_{t,z}}D^t(s)H(t+s,Z^z_s)\ud s+ D^t(\rho^*_{t,z})w\big(t+\rho^*_{t,z},Z^z_{\rho^*_{t,z}}\big)\Big].
	\end{equation*}
	Taking $\eps\in(0,\frac{T-T_1}{2})$ we have $t+\eps+\rho^*_{t,z}\leq T_1+\eps<T$ and so $\rho^*_{t,z}$ is admissible for $w(t+\eps,z)$.  By the supermartingale property of the value function (cf.\ again \eqref{eq:UD}) we have
	\begin{equation*}
	w(t+\eps,z)\ge \E\Big[\int_0^{\rho^*_{t,z}}D^{t+\eps}(s)H(t+\eps+s,Z^z_s)\ud s+ D^{t+\eps}(\rho^*_{t,z})w\big(t+\eps+\rho^*_{t,z},Z^z_{\rho^*_{t,z}}\big)\Big].
	\end{equation*}
	
	Combining the two expressions above yields
\begin{align}\label{eq:wtc-0}
\begin{aligned}
&w(t+\eps,z)-w(t,z)\\
	&\geq\E\Big[\int_0^{\rho^*_{t,z}}\Big(D^{t+\eps}(s)H(t+\eps+s,Z^z_s)-D^t(s)H(t+s,Z^z_s) \Big)\ud s\\
	&\qquad+D^{t+\eps}(\rho^*_{t,z})w\big(t+\eps+\rho^*_{t,z},Z^z_{\rho^*_{t,z}}\big)-D^t(\rho^*_{t,z})w\big(t+\rho^*_{t,z},Z^z_{\rho^*_{t,z}}\big) \Big].
\end{aligned}
\end{align}
On $\{\tau^*_{t,z}<T_1-t \}$ we have $\rho^*_{t,z}=\tau^*_{t,z}$ and 
$D^{t+\eps}(\tau^*_{t,z})w\big(t+\eps+\tau^*_{t,z},Z^z_{\tau^*_{t,z}}\big)\geq 0=D^t(\tau^*_{t,z})w\big(t+\tau^*_{t,z},Z^z_{\tau^*_{t,z}}\big)$. Then, the right-hand side in \eqref{eq:wtc-0} is further bounded from below by 
\begin{align*}
\begin{aligned}
&w(t+\eps,z)-w(t,z)\\
	&\geq\E\Big[\int_0^{\rho^*_{t,z}}\Big(D^{t+\eps}(s)H(t+\eps+s,Z^z_s)-D^t(s)H(t+s,Z^z_s) \Big)\ud s\\
	&\qquad+\ind_{\{\tau^*_{t,z}\ge T_1-t\}}\Big(D^{t+\eps}(T_1-t)w\big(T_1+\eps,Z^z_{T_1-t}\big)-D^t(T_1-t)w\big(T_1,Z^z_{T_1-t}\big) \Big)\Big]\\
	&=\E\Big[\int_0^{\rho^*_{t,z}}\Big(D^{t+\eps}(s)H(t+\eps+s,Z^z_s)-D^t(s)H(t+s,Z^z_s) \Big)\ud s\\
	&\qquad+\ind_{\{\tau^*_{t,z}\ge T_1-t\}}\Big(D^{t+\eps}(T_1-t)-D^t(T_1-t)\Big)w\big(T_1+\eps,Z^z_{T_1-t}\big)\\
	&\qquad+\ind_{\{\tau^*_{t,z}\ge T_1-t\}}D^{t}(T_1-t)\int_0^\eps w_t\big(T_1+s,Z^z_{T_1-t}\big)\ud s \Big],
\end{aligned}
\end{align*}
where the expected value of the integral in the last line is well-defined because $w_t$ is locally bounded thanks to Lemma \ref{lem:w_tUpBound} and Lemma \ref{lem:w_tLowBound} and $Z^z_{T_1-t}$ admits a density with respect to the Lebesgue measure.
Dividing by $\eps$, letting $\eps\to 0$ and using dominated convergence (since $t\mapsto H(t,z)$ and $t\mapsto D^t(s)$ are Lipschitz-continuous and $w$ is continuous)
	\begin{align}\label{eq:w_tgeqT_1}
	\begin{aligned}	
	w_t(t,z)&\geq\E\Big[\int_0^{\rho^*_{t,z}}D^t(s)\Big(\big(\mu(t)-\mu(t+s)\big) H(t+s,Z^z_s)+H_t(t+s,Z^z_s)\Big)\ud s\\
	&\qquad+\ind_{\{\tau^*_{t,z}\ge T_1-t\}}D^t(T_1-t)\big(\mu(t)-\mu(T_1) \big)w(T_1,Z^z_{T_1})\\
	&\qquad-\ind_{\{\tau^*_{t,z}\geq T_1-t\}}D^t(T_1-t)q(T-T_1)\big(1+Z^z_{T_1-t}\big)\Big],
\end{aligned}	
	\end{align}
	where in the final expression we use the lower bound from Lemma \ref{lem:w_tLowBound}. 

Now recall that $(t_0,z_0)\in\partial\cC$ with $t_0<T$ and that $(t_n,z_n)_{n\geq 1}\subseteq\cC$ converges to $(t_0,z_0)$ as $n\to\infty$. With no loss of generality, we can assume that $t_n<T_1$ for every $n\in\bN$. Let $\rho^*_n=\rho^*_{t_n,x_n}$ for every $n\in\bN$. Since by Lemma \ref{lem:tau_nTo0} $\lim_{n\to\infty}\tau^*_{t_n,z_n}=0$, $\P$-a.s., then $\lim_{n\to\infty}\rho^*_n=0$, $\P$-a.s. Therefore, from \eqref{eq:w_tgeqT_1} computed at $(t_n,z_n)\in\cC$, we obtain
	$\liminf_{n\to\infty} w_t(t_n,z_n)\geq 0$,
	which, together with \eqref{eq:w_t<0}, yields
	$\lim_{n\to\infty} w_t(t_n,z_n)=0$.
	That completes the proof  by arbitrariness of $(t_0,z_0)$ and of the sequence $(t_n,z_n)$.
\end{proof}

It remains to prove Lemma \ref{lem:tau_nTo0}. For that we need to introduce a different coordinate system, leveraging on the parametrisation of the boundary $\beta(t)$. For any fixed $t\in[0,T)$, we define the space transformation $y=\psi_t(z):=e^{\Lambda t} z$ for $z\in(0,\infty)$, with $\Lambda\geq 0$ as in \eqref{eq:beta}. Clearly, $(t,z)\in[0,T]\times(0,\infty)$ implies $(t,y)\in[0,T]\times(0,\infty)$. Since $\beta(t)=b(t)e^{\Lambda t}$, the continuation set \eqref{setC1} in the new coordinates reads
\[
\cC=\{(t,y)\in[0,T) \times(0,\infty): y>\beta(t)\},
\]
and analogously
\[
\cD=\{(t,y)\in[0,T) \times(0,\infty): y\le\beta(t)\} \cup \big(\{T\}\times(0,\infty)\big).
\]
In order to avoid confusion when parametrising the sets $\cC$ and $\cD$ with respect to the two different coordinate systems, $(t,z)$ and $(t,y)$, we write $\widehat \cC$ and $\widehat \cD$ when we refer to the coordinates $(t,y)$.

The change of coordinates has a probabilistic analogue that turns out to be very useful.
Fix $(t,z)\in[0,T]\times(0,\infty)$. 
Set $Y^{y;t}_s:=e^{\Lambda (s+t)}Z^{z}_s$ for $s\in[0,T-t]$, so that
\begin{equation}\label{eq:Ysde}
\ud Y^{y;t}_s=\gamma Y^{y;t}_s\ud s+\sigma Y^{y;t}_s\ud W_s, \quad Y^{y;t}_0=y:=ze^{\Lambda t},
\end{equation}
where $\gamma:=g+c+\Lambda-r=\alpha+\Lambda$ (cf.\ \eqref{eq:alpha}). We can re-parametrise our optimal stopping problem \eqref{valuew} in the new coordinates. First, we notice that 
\begin{align}\label{eq:hatw}
w(t,z)=w\big(t,ye^{-\Lambda t}\big)=:\hat w(t,y).
\end{align}
Then, using $Z^z_s=Y^{y;t}_se^{-\Lambda(s+t)}$ we can write
\begin{align}\label{eq:OShat}
\begin{aligned}
\hat w(t,y)=\sup_{0\le\tau\le T-t}\E\Big[&\int_{0}^{\tau}e^{-cs} {_s p_{t}}  \Big(f(t+s)+ \mu(t+s) \big(Y^{y;t}_se^{-\Lambda(s+t)}-1\big)^{+}  \Big)  \ud s\\
&+\ind_{\{\tau= T-t \}}e^{-c (T-t)}  \ _{T-t} p_{t} \big(Y^{y;t}_{T-t}e^{-\Lambda T}-1\big)^{+} \Big].
\end{aligned}
\end{align}
In the notation $Y^{y;t}$ we keep track of the time $t$ in a parametric way for the following reason: for a fixed $y\in(0,\infty)$ and any two $t_1$, $t_2$, the dynamics $(Y^{y;t_1}_s)_{s\in[0,\infty)}$ and $(Y^{y;t_2}_s)_{s\in[0,\infty)}$ are indistinguishable and therefore we could simply set $Y^{y;t_1}=Y^{y;t_2}=Y^{y}$; however, if we want to recover the original processes $Z^{z_1}$ and $Z^{z_2}$ from which $Y^{y;t_1}$ and $Y^{y;t_2}$ originate, we must be careful that $Z^{z_1}_s= e^{\Lambda(s+t_1)}Y^{y;t_1}_s$ and $Z^{z_2}_s= e^{\Lambda(s+t_2)}Y^{y;t_2}_s$ so that $Z^{z_1}_s\neq Z^{z_2}_s$.
Recalling the function $H$ from \eqref{H} and comparing \eqref{eq:OShat} with \eqref{valuew}, it is easy to see that $H(t,z)=H(t,ye^{-\Lambda t})=:\hat H(t,y)$ and therefore
\begin{align}\label{eq:hatH}
\hat H(t,y)=f(t)+\mu(t)\big(e^{-\Lambda t}y-1\big)^+.
\end{align}

Moreover, recall the optimal stopping time from \eqref{eq:tau*1}. With the notation above, we can equivalently express it as
\begin{align}
\begin{aligned}
\tau^*_{t,z}&=\inf\{s\ge 0: Z^{z}_s\le b(t+s)\}\wedge (T-t)\\
&=\inf\{s\geq 0:Y^{y;t}_s\leq \beta(t+s) \}\wedge (T-t)=:\sigma^*_{t,y}
\end{aligned}
\end{align}
Analogously, we can introduce
\begin{align*}
\tau'_{t,z}&:=\inf\{s\ge 0: Z^{z}_s< b(t+s)\}\wedge (T-t)\\
&=\inf\{s\geq 0:Y^{y;t}_s< \beta(t+s) \}\wedge (T-t)=:\sigma'_{t,y}.
\end{align*}

Then, we have the following lemma.

\begin{lemma}\label{lem:tau*=tau'}
For every $(t,z)\in[0,T)\times(0,\infty)$, we have that $\tau'_{t,b(t)}=0$ and $\tau'_{t,z}=\tau^*_{t,z}$, $\P$-a.s. Equivalently, for every $(t,y)\in[0,T)\times(0,\infty)$, we have $\sigma'_{t,\beta(t)}=0$ and $\sigma'_{t,y}=\sigma^*_{t,y}$, $\P$-a.s.
\end{lemma}
\begin{proof} 
It is enough to prove the claims only for $\sigma^*$ and $\sigma'$, by virtue of their equivalence to $\tau^*$ and $\tau'$.
We first show that $\sigma'_{t,\beta(t)}=0$, $\P$-a.s. From \eqref{eq:Ysde}, we have 
\[
Y^{y;t}_s=y\exp\Big(\sigma W_s+\big(\gamma-\tfrac{\sigma^2}{2}\big)s\Big).
\]
	Therefore, by monotonicity of $t\mapsto\beta(t)$, we obtain
	\begin{align*}
\sigma'_{t,\beta(t)}&\leq \inf\{s\geq 0:Y^{\beta(t);t}_s< \beta(t) \}=\inf\{s\geq 0:\exp\big(\sigma W_s+(\gamma-\sigma^2/2)s\big)< 1 \}\\
	&=\inf\{s\geq 0:\sigma W_s< (\sigma^2/2-\gamma)s \}=0, \quad \P\text{-a.s.,}
	\end{align*}
	where the last equality follows from the law of iterated logarithm. Hence, $\sigma'_{t,\beta(t)}=0$, $\P$-a.s.
	
We now show that $\sigma'_{t,y}=\sigma^*_{t,y}$, $\P$-a.s. The pair $(t,y)$ is fixed in the rest of the proof and we will use the strong Markov property of the process $Y^{y;t}$. With no loss of generality here we assume $(\Omega,\cF)$ as the canonical space of continuous trajectories. We denote $\sigma'_{t,y}$ and $\sigma^*_{t,y}$ under $\P$, respectively, by $\sigma'$ and $\sigma^*$ under $\P^t_{0,y}$, where $\P^t_{0,y}(A)=\P(A|Y^{y;t}_0=y)$ for $A\in\cF$. The notation is slightly cumbersome because we must keep track of $t$ which enters the dynamics $Y^{y;t}$ via the initial condition $y=e^{\Lambda t}z$ and it enters $\sigma'$ and $\sigma^*$ via the boundary $s\mapsto \beta(t+s)$. In particular, we will use that for any $(s,\lambda)\in[0,T-t)\times[0,\infty)$ the stopping time $\sigma'$ under the measure $\P^t_{s,\lambda}(\cdot)=\P(\cdot | Y^{y;t}_s=\lambda)$ reads
\[
\sigma'=\inf\{u\ge 0:Y^{y;t}_{s+u}<\beta(t+s+u)\}\wedge(T-t-s),\quad\text{with $Y^{y;t}_s=\lambda$.}
\]

Clearly, we have $\sigma'\geq\sigma^*$, $\P^t_{0,y}$-a.s. Assume by contradiction that $\P^t_{0,y}(\sigma'>\sigma^*)>0$. Then, there exists $\delta>0$ such that $\P^t_{0,y}(\sigma'\geq\sigma^*+\delta)>0$. Denote by $\{\theta_s,\,s\ge 0\}$ the shift operator. For any functional $F:[0,\infty)\times C([0,\infty);\R)\to \R$ and any $\varphi\in C([0,\infty);\R)$ it holds $ (F\circ \theta_u) (s,\varphi(\cdot))=F(u+s,\varphi(u+\cdot))$. Then 
	\begin{align*}
	0&< \P^t_{0,y}\big(\sigma'\geq \sigma^*+\delta\big)=\E^t_{0,y}\Big[\P^t_{0,y}\big(\sigma'\geq \sigma^*+\delta\big|\cF_{\sigma^*}\big)\Big]\\
	&=\E^t_{0,y}\Big[\P^t_{0,y}\big(\sigma'\circ\theta_{\sigma^*}+\sigma^*\geq \sigma^*+\delta\big|\cF_{\sigma^*}\big)\Big]=\E^t_{0,y}\Big[\P^t_{0,y}\big(\sigma'\circ\theta_{\sigma^*}\geq\delta\big|\cF_{\sigma^*}\big)\Big]\\
	&=\E^t_{0,y}\Big[\P^t_{\sigma^*,Y_{\sigma^*}}\big(\sigma'\geq\delta\big)\Big]=0,
	\end{align*}
	where the penultimate equality holds by the strong Markov property and the last equality holds by the first part of the proof because $Y_{\sigma^*}\leq\beta(\sigma^*)$, $\P^t_{0,y}$-a.s. Thus we have reached a contradiction and it must be $\sigma'_{t,y}=\sigma^*_{t,y}$, $\P$-a.s.
\end{proof}

As a consequence of Lemma \ref{lem:tau*=tau'} and monotonicity of the boundary $t\mapsto \beta(t)$, we obtain that
\[
\sigma^*_{t,y}=\sigma'_{t,y}=\zeta^*_{t,y}:=\inf\{s\geq 0: (t+s,Y_s^{y;t})\in\text{int}(\widehat\cD) \}\wedge(T-t), \quad \P\text{-a.s.},
\]
where $\text{int}(\widehat\cD):=\widehat\cD\setminus\partial\widehat\cD$ is the interior of $\widehat\cD$. 
Since 
\begin{align*}
&\inf\{s\geq 0: (t+s,Y_s^{y;t})\in\text{int}(\widehat\cD) \}\wedge(T-t)=\inf\{s\geq 0: (t+s,Z_s^{z})\in\text{int}(\cD) \}\wedge(T-t),
\end{align*}
then
\begin{equation}\label{eq:tauIntD}
\tau^*_{t,z}=\tau'_{t,z}=\rho^*_{t,z}:=\inf\{s\geq 0: (t+s,Z_s^z)\in\text{int}(\cD) \}\wedge(T-t), \quad \P\text{-a.s.,}
\end{equation}
This allows us to prove Lemma \ref{lem:tau_nTo0}.

\begin{proof}[{\bf Proof of Lemma \ref{lem:tau_nTo0}}]
Clearly, we have $\liminf_{n\to\infty}\tau^*_{t_n,z_n}\geq 0$, $\P$-a.s. We now want to prove $\limsup_{n\to\infty}\tau^*_{t_n,z_n}\leq 0$, $\P$-a.s. Since $(t,z)\in\partial\cC$ with $t<T$, from \eqref{eq:tauIntD} and Lemma \ref{lem:tau*=tau'}, we have $\rho^*_{t,z}(\omega)=\tau^*_{t,z}(\omega)=0$, for $\omega\in\Omega_0$ with $\P(\Omega_0)=1$. Fix $\omega\in\Omega_0$. Then, for every $\delta>0$, there exists $s\in[0,\delta)$ such that $(t+s,Z^z_s(\omega))\in\text{int}(\cD)$. Since $\text{int}(\cD)$ is an open set, 
then also $(t_n+s,Z^{z_n}_s(\omega))\in\text{int}(\cD)$ for $n$ sufficiently large. Hence, $\tau^*_{t_n,z_n}(\omega)\leq\delta$ for $n$ sufficiently large and, thus, we obtain $\limsup_{n\to\infty}\tau^*_{t_n,z_n}(\omega)\leq\delta$. Since $\delta>0$ was arbitrary and $\omega\in\Omega_0$, by letting $\delta\to 0$ we have the desired result.
\end{proof}

\section{Continuity of the optimal boundary and integral equations}\label{sect:bCont}

In this section we prove Theorem \ref{thm:bCont}. We notice that $t\mapsto b(t)$ is right-continuous because of Proposition \ref{prop:betamonot}. Thus, it remains to show its left-continuity at any $t\neq t^*$. After obtaining continuity of $b$, we will use it to derive an integral equation that allows an efficient numerical computation of the boundary.

Recall the functions $\hat w$ and $\hat H$ defined in \eqref{eq:hatw} and \eqref{eq:hatH}, respectively, and the free-boundary problem \eqref{eq:fbp} solved by $w$. Then, it is not hard to check that $\hat w$ solves
\begin{equation}\label{eq:hatwODE}
	\hat{w}_t(t,y)+\big(\cG\hat w\big)(t,y)-(c+\mu(t))\hat{w}(t,y)+\hat{H}(t,y)=0, \qquad (t,y)\in\widehat \cC,
\end{equation}
where $\cG$ is the generator of the process $Y$ in \eqref{eq:Ysde}, which is defined on sufficiently smooth functions $\varphi$ as $(\cG \varphi)(t,y) := (\alpha+\Lambda) y \varphi_y(t,y)+\tfrac{\sigma^2}{2} y^2 \varphi_{yy}(t,y)$, and we also recall that $\widehat \cC$ corresponds to the set $\cC$ expressed with the $(t,y)$ coordinates.

\begin{proposition}\label{prop:betalc}
The transformed boundary $t\mapsto \beta(t)$ is continuous on $[0,T]\setminus\{t^*\}$.
\end{proposition}
\begin{proof}
Continuity on $[0,t^*)$ is trivial because $\beta(t)=0$ for $t\in[0,t^*)$ (cf.\ Proposition \ref{prop:betamonot}). With no loss of generality we then assume $t^*\in[0,T)$.
Right-continuity on $[0,T]$ was proven in Proposition \ref{prop:betamonot} and it only remains to prove left-continuity on $(t^*,T]$.
Combining Propositions \ref{prop:limsupT} and \ref{prop:betamonot} we know the limits $\lim_{t\uparrow T}b(t)=1$ and $\lim_{t\uparrow T}\beta(t)=e^{\Lambda T}$. Then, left-continuity at $T$ follows by simply setting $b(T):=1$ (cf.\ Theorem \ref{thm:bCont}) and $\beta(T):=e^{\Lambda T}$.	
	
By monotonicity, the left-limit of $\beta$ exists at every point $t\in(t^*,T]$. We denote it by $\beta(t-)$. Arguing as in \cite{de2015note}, let us assume by contradiction that there exists $t_0\in(t^*,T)$ such that $\beta(t_0-)<\beta(t_0)$. Let $0<y_1<y_2<\infty$ be such that $[y_1,y_2]\subset(\beta(t_0-),\beta(t_0))$. Then, by monotonicity of $\beta$, we have that $[0,t_0)\times [y_1,y_2]\subset\widehat\cC$. Let $\varphi$ be an arbitrary non-negative, infinitely differentiable function with compact support on $[y_1,y_2]$. By \eqref{eq:hatwODE}, for every $t\in[0,t_0)$, we obtain
	\begin{align}\label{eq:phi}
		\begin{aligned}
			-\int_{y_1}^{y_2}\hat H(t,y)\varphi(y)\ud y&=\int_{y_1}^{y_2}\Big(\hat{w}_t(t,y)+\big(\cG\hat{w}\big)(t,y)-(c+\mu(t)) \hat{w}(t,y)\Big)\varphi(y)\ud y\\
			&=\int_{y_1}^{y_2}\Big(\big(\hat{w}_t(t,y)-(c+\mu(t))\hat{w}(t,y)\big)\varphi(y)+\hat{w}(t,y)\big(\cG^*\varphi\big)(y)\Big)\ud y,
		\end{aligned}
	\end{align}
	where $\cG^*$ is the adjoint operator of $\cG$.
Since $(t_0,y)\in\widehat\cD$ implies $(t_0,z):= (t_0, ye^{-\Lambda t_0})\in\cD$ for every $y\in[y_1,y_2]$, and since $w\in C^1([0,T)\times(0,\infty))$ by Proposition \ref{prop:wC1}, then
\begin{align*}
&\lim_{t\uparrow t_0}\hat w(t,y)=\lim_{t\uparrow t_0} w\big(t,y e^{-\Lambda t}\big)=0\quad\text{and}\\
&\lim_{t\uparrow t_0}\hat w_t(t,y)=\lim_{t\uparrow t_0}\Big(w_t\big(t,y e^{-\Lambda t}\big)-\Lambda y e^{-\Lambda t} w_z\big(t,y e^{-\Lambda t}\big)\Big)=0,
\end{align*}	
for $y\in[y_1,y_2]$. Hence, letting $t\to t_0$ in \eqref{eq:phi} and using dominated convergence (thanks to the continuity of $\hat H$,  $\mu$, $\hat w$ and $\hat w_t$), we obtain
$\int_{y_1}^{y_2}\hat H(t_0,y) \varphi(y)\ud y=0$.
Since $\varphi$ is arbitrary, then it must be $\hat H(t_0,y)=0$ for all $y\in(y_1,y_2)$, which is equivalent to 
$H(t_0,z)=0$, for $z\in(e^{-\Lambda t_0}y_1,e^{-\Lambda t_0}y_2)$.
The latter is impossible and we have reached a contradiction which implies left-continuity of $\beta$.
\end{proof}
Now we can easily prove Theorem \ref{thm:bCont}.

\begin{proof}[\bf Proof of Theorem \ref{thm:bCont}]
Since $b(t)=e^{-\Lambda t}\beta(t)$, thanks to Proposition \ref{prop:limsupT} and Proposition \ref{prop:betalc} we obtain continuity of $b$ on $[0,T]\setminus\{t^*\}$ and its limit at $T$.
\end{proof}

Finally, we are ready to prove Theorem \ref{thm:main}. We will give a general statement concerning a probabilistic representation for the function $w$ and an integral equation for the boundary $b$. Then, the proof of Theorem \ref{thm:main} will follow via a simple re-parametrisation. It is useful to recall the set $\mathcal R$ introduced after \eqref{eq:sC} and notice that $(t,z)\in\mathcal R\iff z>h(t)$, where $h(t)=(1-f(t)/\mu(t))^+$. Notice that $h(t)$ is the unique solution of $H(t,h(t))=0$ for $t\in[t^*,T)$.
\begin{proposition}\label{prop:inteq}
For every $(t,z)\in[0,T]\times(0,\infty)$ the function $w$ can be represented as 
\begin{align}\label{eq:wint}
\begin{aligned}
w(t,z)=e^{-c(T-t)}{_{T-t} p_t}\E\big[\big(Z^z_{T-t}-1\big)^+\big]+\int_0^{T-t}e^{-cs}{_s p_t}\E\big[H(t+s,Z^z_s)\ind_{\{Z^z_s> b(t+s)\}}\big]\ud s.
\end{aligned}
\end{align}
Moreover, the boundary $b(t)$ is the unique continuous solution of the integral equation, 
\begin{align}\label{eq:bint}
\begin{aligned}
e^{-c(T-t)}{_{T-t} p_t}\E\big[\big(Z^{b(t)}_{T-t}-1\big)^+\big]+\int_0^{T-t}e^{-cs}{_s p_t}\E\big[H(t+s,Z^{b(t)}_s)\ind_{\{Z^{b(t)}_s > b(t+s)\}}\big]\ud s=0,
\end{aligned}
\end{align}
for\footnote{It is important to notice here that $t\in(t^*,T)$ should be understood as $t\in\Gamma_b:=\{s\in[0,T):b(s)>0\}$, because the equation only holds when $b$ is strictly separated from zero. This observation is also important when proving uniqueness.} $t\in(t^*,T)$,
with $b(t)\le h(t)$ for all $t\in(t^*,T)$ and $b(T)=1$.
\end{proposition}
\begin{proof}
Since $w$ is continuously differentiable in $[0,T)\times(0,\infty)$ with $w_{zz}$ continuous on $\bar \cC_\eps$ and on $\cD_\eps$ (cf.\ Proposition \ref{prop:wC1}), then $\hat w$ inherits the same regularity. Moreover, $t\mapsto \beta(t)$ is continuous and non-decreasing and therefore we can apply \cite[Thm.\ 3.1]{peskir2005change} to obtain, for $t\in [0,T)$ and any $u\in[0,T-t)$
\begin{align*}
\begin{aligned}
&\E\Big[e^{-c u+\int_0^u\mu(t+s)\ud s}\hat w\big(t+u,Y^{y;t}_u\big)\Big]\\
&=\hat w(t,y)+\E\Big[\int_0^u e^{-cs+\int_0^s\mu(t+r)\ud r}\big(\hat w_t+\cG \hat w-(c+\mu(\cdot))\hat w\big)(t+s,Y^{y;t}_s)\ind_{\{Y^{y;t}_s > \beta(t+s)\}}\ud s\Big]\\
&=\hat w(t,y)-\E\Big[\int_0^u e^{-cs+\int_0^s\mu(t+r)\ud r}\hat H(t+s,Y^{y;t}_s)\ind_{\{Y^{y;t}_s > \beta(t+s)\}}\ud s\Big],
\end{aligned}
\end{align*}
where in the second equality we used \eqref{eq:hatwODE}. Then, changing variables and expressing $y=e^{\Lambda t}z$, $\hat w(t,y)=w(t,z)$, $\hat H(t,y)=H(t,z)$ and $Y^{y;t}_s=e^{\Lambda(t+s)}Z^z_s$ (cf.\ \eqref{eq:Ysde}, \eqref{eq:hatw} and \eqref{eq:hatH}) we obtain
\begin{align*}
\begin{aligned}
&\E\Big[e^{-c u+\int_0^u\mu(t+s)\ud s}w\big(t+u,Z^z_u\big)\Big]=w(t,z)-\E\Big[\int_0^u e^{-cs+\int_0^s\mu(t+r)\ud r} H(t+s,Z^z_s)\ind_{\{Z^z_s > b(t+s)\}}\ud s\Big].
\end{aligned}
\end{align*}
Letting $u\to T-t$ we can use dominated convergence, continuity of $w$ and $w(T,z)=(z-1)^+$ to conclude 
\begin{align*}
\begin{aligned}
\E\Big[e^{-c(T-t)+\int_0^{T-t}\mu(t+s)\ud s}\big(Z^z_{T-t}-1\big)^+\Big]=w(t,z)\!-\!\E\Big[\!\int_0^{T-t}\!\! e^{-cs+\int_0^s\mu(t+r)\ud r} H(t\!+\!s,Z^z_s)\ind_{\{Z^z_s > b(t+s)\}}\ud s\Big].
\end{aligned}
\end{align*}
Rearranging terms yields \eqref{eq:wint}.

Taking $z=b(t)$ and $t\in(t^*,T]$ in \eqref{eq:wint} immediately yields that $b$ solves \eqref{eq:bint} because $w(t,b(t))=0$ for $t\in(t^*,T]$ due to the fact that $b$ is continuous on $(t^*,T]$ with $b(T)=1$ (Theorem \ref{thm:bCont}). It only remains to show that \eqref{eq:bint} admits at most one continuous solution with terminal condition equal to one. This fact can be proven using a well-consolidated procedure which was first established in \cite{peskir2005american} and subsequently refined in several other papers (see, e.g., numerous examples in \cite{peskir2006optimal}). Here we only illustrate the main ideas and we refer the interested reader to \cite{peskir2006optimal} for fuller details.

First one assumes that there is a function $d:[0,T]\to [0,\infty)$ from the class $\Sigma$ (cf.\ \eqref{eq:Sigma}), such that
\begin{equation*}
\begin{aligned}
e^{-c(T-t)}{_{T-t} p_t}\E\big[\big(Z^{d(t)}_{T-t}-1\big)^+\big]+\int_0^{T-t}e^{-cs}{_s p_t}\E\big[H(t+s,Z^{d(t)}_s)\ind_{\{Z^{d(t)}_s > d(t+s)\}}\big]\ud s=0,
\end{aligned}
\end{equation*}
for all $t\in\Gamma_d$ where $\Gamma_d:=\{s\in[0,T]:d(s)>0\}$. Then, setting
\begin{align*}
\begin{aligned}
w^d(t,z)=e^{-c(T-t)}{_{T-t} p_t}\E\big[\big(Z^z_{T-t}-1\big)^+\big]+\int_0^{T-t}e^{-cs}{_s p_t}\E\big[H(t+s,Z^z_s)\ind_{\{Z^z_s> d(t+s)\}}\big]\ud s,
\end{aligned}
\end{align*}
one can show that $w^d$ is a continuous function (even though $d$ is not continuous at $t^*$). Using the strong Markov property, it is easy to show that 
\begin{align*}
\begin{aligned}
u\mapsto e^{-c u}{_u p_t}w^d\big(t+u,Z^z_{u}\big)+\int_0^{u}e^{-cs}{_s p_t}H(t+s,Z^z_s)\ind_{\{Z^z_s> d(t+s)\}}\ud s,
\end{aligned}
\end{align*}
is a continuous martingale on $[0,T-t]$. The rest of the proof follows four steps, all of which use the martingale property above. First, one shows that $w^d(t,z)=0$ for $z\le d(t)$ and $t\in\Gamma_d$. Second, one shows that $w(t,z)\ge w^d(t,z)$ for $(t,z)\in[0,T)\times(0,\infty)$. Using those two results one obtains, in two separate steps, the two inequalities  $d(t)\ge b(t)$ and $d(t)\le b(t)$ for all $t\in[0,T]$.
\end{proof}

We can finally prove Theorem \ref{thm:main}.
\begin{proof}[\bf Proof of Theorem \ref{thm:main}]
Recall from Section \ref{subs:markovian} that $u(t,z)=w(t,z)+(1-k(t))$ and that the pairs $(t,z)$ and $(t,x)$ are linked by the transformation $z=x_0 e^{gt}/x$. Moreover, \eqref{eq:uv} gives $v(t,x)=x u(t,x_0 e^{gt}/x)$ and $V_0=v(0,x_0)=x_0 u(0,1)$. Then, by \eqref{eq:wint} we obtain
\begin{align*}
\begin{aligned}
&v(t,x)\\
&=x(1-k(t))\\
&\quad+x\Big(e^{-c(T-t)}{_{T-t} p_t}\E\big[\big(Z^z_{T-t}\!-\!1\big)^+\big]\!+\!\int_0^{T-t}\!e^{-cs}{_{s} p_t}\E\Big[H(t\!+\!s,Z^z_s)\ind_{\{Z^z_s> b(t+s)\}}\Big]\ud s\Big)\\
&=x e^{-c(T-t)}{_{T-t} p_t}\E\big[\mathrm{max}\big\{Z^z_{T},1\big\}\big]\\
&\quad+x\int_0^{T-t}\!e^{-cs}{_{s} p_t}\E\Big[\mu(t\!+\!s)\mathrm{max}\big\{Z^z_{s},1\big\}\ind_{\{Z^z_s> b(t+s)\}}\! +\! (\mu(t\!+\!s)\!-\!f(t\!+\!s))\ind_{\{Z^z_s\le b(t+s)\}}\Big]\ud s,
\end{aligned}
\end{align*} 
where the second equality follows from \eqref{eq:kito} and recalling $k(T)=0$.  

Now we take $t=0$ and $x=x_0$ so that $z=1$ and compute
\begin{align*}
\begin{aligned}
&V_0=x_0 e^{-cT}{_{T} p_0}\E\big[\mathrm{max}\big\{Z^1_{T},1\big\}\big]\\
&\quad+x_0\int_0^{T}\!e^{-cs}{_{s} p_0}\E\Big[\mu(s)\mathrm{max}\big\{Z^1_{s},1\big\}\ind_{\{Z^1_s> b(s)\}}\! +\! (\mu(s)\!-\!f(s))\ind_{\{Z^1_s\le b(s)\}}\Big]\ud s,
\end{aligned}
\end{align*} 
Recall that $Z^1_s=x_0 e^{gs}/X_s$ for any $s\in[0,T]$, where the process $(X_s)_{s\in[0,T]}$ starts from $X_0=x_0$ at time zero. Recall also that for any stopping time $\tau\in[0,T]$ and any $\cF_\tau$-measurable random variable $\Psi$ the change of measure from $\P$ to $\Q$ reads $x_0\E[e^{-c\tau}\Psi]=\E^{\Q}[e^{-r \tau} X_\tau\Psi]$ (cf.\ \eqref{eq:PQ}). Substituting in the expression above yields
\begin{align*}
\begin{aligned}
&V_0= e^{-rT}{_{T} p_0}\E^\Q\big[\mathrm{max}\big\{x_0e^{g T},X_{T}\big\}\big]\\
&\quad+\int_0^{T}\!\!e^{-rs}{_s p_0}\E^\Q\Big[\mu(s)\mathrm{max}\big\{x_0e^{gs},X_s\big\}\ind_{\{\ell_{x_0}(s)> X_s\}}\! +\! X_s(\mu(s)\!-\!f(s))\ind_{\{\ell_{x_0}(s)\le  X_s\}}\Big]\ud s,
\end{aligned}
\end{align*} 
where in the indicator functions we have introduced $\ell_{x_0}(s):=x_0e^{gs}/b(s)$. This completes the proof of \eqref{eq:V0int} and by definition we also have $\ell_{x_0}(T)=x_0 e^{gT}$ because $b(T)=1$.

For \eqref{eq:ellint}, first we rewrite \eqref{eq:bint} in terms of the density $\Phi(z,s,y)$ of $Z^z_s$ under $\P$. That yields
\begin{align*}
\begin{aligned}
e^{-c(T-t)}{_{T-t} p_t}\int_1^\infty \!\Phi\big(b(t),T\!-\!t,y\big)\big(y\!-\!1\big)\ud y\!+\!\int_0^{T-t}\!\!e^{-cs}{_s p_t}\int_{b(t+s)}^\infty H(t\!+\!s,y)\Phi\big(b(t),s,y\big)\ud y\,\ud s=0.
\end{aligned}
\end{align*}
Changing variable of integration we obtain
\begin{align*}
\begin{aligned}
e^{-c(T-t)}{_{T-t} p_t}\int_1^\infty \!\Phi\big(b(t),T\!-\!t,y\big)\big(y\!-\!1\big)\ud y\!+\!\int_t^{T}\!\!e^{-c(s-t)}{_{s-t} p_t}\int_{b(s)}^\infty H(s,y)\Phi\big(b(t),s-t,y\big)\ud y\,\ud s=0.
\end{aligned}
\end{align*}
The proof is completed by recalling $\ell_{x_0}(s)=x_0e^{gs}/b(s)$.

Optimality of the stopping time $\gamma_*$ follows easily because
\begin{align*}
\gamma_*&=\inf\{s\in[0,T): X_s\ge \ell_{x_0}(s)\}\wedge T=\inf\{s\in[0,T): X_s\ge x_0 e^{gs}/b(s)\}\wedge T\\
&=\inf\{s\in[0,T): Z^1_s\le b(s)\}\wedge T=\tau^*_{0,1},
\end{align*}
and $\tau^*_{0,1}$ (cf.\ \eqref{eq:tau*1}) is optimal for $u(0,1)$.
\end{proof}

\section{Some comments on one-sided Lipschitz bounds for optimal boundaries}\label{sec:meth}

In this section we first draw a detailed comparison between our framework and the one used in \cite{de2019lip} for locally Lipschitz optimal boundaries. We then provide a general statement with the flavour of Lemma \ref{lem:b'_delta} that leads from an initial optimal stopping problem to a reparametrisation thereof in the form of a problem with a monotone boundary. We emphasise that we are not aiming for the fullest generality, but rather we want to illustrate methods that can be developed in future research.  

We consider a dynamics on $\cI\subset\R$ of the form
\[
Z_t=z+\int_0^t\alpha(Z_s)\ud s+\int_0^t\sigma(Z_s)\ud W_s,
\]
where $W$ is a Brownian motion and for simplicity we assume $\alpha,\sigma\in C^1(\cI)$. We use the notation $Z^z_t$ to keep track of the initial condition $Z^z_0=z$ and we assume that the SDE admits a unique strong solution. The optimal stopping problem of interest reads
\begin{equation}\label{eq:V}
V(t,z)=\sup_{0\le \tau\le T-t}\E\Big[\int_0^\tau H(t+s,Z^z_s)\ud s+\ind_{\{\tau=T-t\}}G(Z^z_{T-t})\Big],
\end{equation}
where specific assumptions on functions $H$ and $G$ will follow later. Notice that in the special case $\alpha(z)=\alpha z$, $\sigma(z)=\sigma z$, 
\begin{equation}\label{eq:HGVA}
H(t,z)=e^{-\int_0^t(c+\mu(u))\ud u}\Big(f(t)+\mu(t)(z-1)^+\Big)\quad\text{and}\quad G(z)=e^{-\int_0^T(c+\mu(u))\ud u}(z-1)^+
\end{equation} 
we cover the framework of the Variable Annuity (VA) problem from Section \ref{subs:markovian}, with $\cI=[0,\infty)$ and
\[
V(t,z)=e^{-\int_0^t(c+\mu(u))\ud u}w(t,z),
\]
where $w(t,z)$ is the value function in \eqref{valuew}.

We denote the continuation and stopping sets by
\begin{equation*}
\cC:=\big\{(t,z)\in[0,T)\times\cI:V(t,z)>0\big\}\quad\text{and}\quad\cD:=\big\{(t,z)\in[0,T)\times\cI:V(t,z)=0\big\}\cup\big(\{T\}\times\cI\big).
\end{equation*}
If $z\mapsto V(t,z)$ is non-decreasing, then setting $b(t):=\inf\{z\in\cI:V(t,z)>0\}$ we obtain
\begin{equation}\label{eq:Cb}
\cC=\big\{(t,z)\in[0,T)\times\cI:z>b(t)\big\},
\end{equation}
and the study of $t\mapsto b(t)$ becomes a key objective of the theory. For simplicity and with no loss of generality we assume throughout that 
\[
b(t)\in\cI\ \text{for}\ t\in[0,T].
\] 
Otherwise we can restrict the study to the set of times where $b(t)\in\cI$ as in the VA problem. Notice that, by pathwise uniqueness for the SDE, it is not hard to show 
\[
H(t,\cdot),\,G(\cdot)\ \text{increasing}\implies V(t,\cdot)\ \text{increasing}.
\]
Throughout the section we assume that $V\in C^{1,2}(\cC)$ and $V\in C([0,T]\times\cI)$. The above conditions are met by the VA problem. 

Besides some technical assumptions on the coefficients $\alpha$, $\sigma$ and on the function $G$, the main structural assumptions in \cite[Thm.\ 4.3]{de2019lip} on the function $H$ concern the existence of two constants $c_1,c_2>0$ such that 
\begin{equation}\label{eq:asslip}
\partial_z H(t,z)\ge c_1\quad \text{and}\quad\partial_t H(t,z)\le c_2\big(1+\partial_z H(t,z)\big),\quad \text{for}\, (t,z)\in \cC.
\end{equation}
Under the assumption that the law of the process $Z$ is absolutely continuous with respect to the Lebesgue measure we can soften the requirement in \eqref{eq:asslip} to hold for a.e.\ $(t,z)\in \cC$. That is better suited for the case of the VA problem because $H$ is not smooth.
Even then, the first condition in \eqref{eq:asslip} is too restrictive in the VA problem where for a.e.\ $(t,z)\in[0,T)\times[0,\infty)$ we have (cf.\ \eqref{eq:HGVA})
\[
\partial_z H(t,z)=e^{-\int_0^t(c+\mu(u))\ud u}\mu(t)\ind_{\{z > 1\}}.
\] 
The right-hand side is allowed to vanish in $\cC$ because the optimal boundary $b(t)$ takes values smaller than one. The constant $c_1>0$ from \eqref{eq:asslip} appears in several places in the Lipschitz bound of \cite[Thm.\ 4.3]{de2019lip} and the argument of proof therein breaks down somewhat irreparably. While the full strength of Lipschitz regularity appears difficult to recover, the next theorem summarises the main methodological message in our paper: a one-sided Lipschitz bound of the boundary $b(t)$ allows a reparametrisation of the problem in \eqref{eq:V} in terms of another optimal stopping problem with monotonic boundary (hence amenable to further analysis). Following that approach, we provide sufficient conditions for the theorem to be applicable.

\begin{theorem}
Suppose that $\cC$ is of the form in \eqref{eq:Cb} with $b(t)\in\cI$ for $t\in[0,T]$. Suppose also that $\partial_z V(t,z)>0$ for $(t,z)\in\cC$ and\footnote{We slightly abuse our notation as we should specify $\lambda\in[0,\infty)$ such that $b(t)+\lambda\in\cI$.} 
\begin{equation}\label{eq:LB}
\inf_{\lambda\in[0,\infty)}\Big(-\frac{\partial_t V(t,b(t)+\lambda)}{\partial_z V(t,b(t)+\lambda)}\Big)\ge -\zeta(t),
\end{equation}
for a non-negative function $\zeta\in L^1(0,T)$. Then, the function $\beta(t):=b(t)+\int_0^t\zeta(s)\ud s$ is non-decreasing and $\hat \tau_*:=\inf\{t\in[0,T]:Y^{0,z}_t\le \beta(t)\}\wedge T$ is optimal in the problem with value 
\begin{equation}\label{eq:hatV}
\hat V(0,z)=\sup_{0\le \tau\le T}\E\Big[\int_0^\tau\hat H(s,Y^{0,z}_{s})\ud s+\ind_{\{\tau=T\}}\hat G(Y^{0,z}_{T})\Big],
\end{equation}
where
\begin{equation}\label{eq:dynY}
Y^{0,z}_t=z+\int_0^t\hat \alpha\big(s,Y^{0,z}_s\big)\ud s+\int_0^t\hat \sigma\big(s,Y^{0,z}_s\big)\ud W_s,
\end{equation}
and, setting $\chi(t):=\int_0^t\zeta(s)\ud s$, functions $\hat \alpha$, $\hat \beta$, $\hat H$ and $\hat G$ read as
\begin{equation*}
\begin{aligned}
&\hat H(t,y):=H\big(t,y-\chi(t)\big),\quad \hat G(y):=G(y-\chi(T)),\\
&\hat \alpha(t,y):=\alpha\big(y-\chi(t)\big)+\zeta(t),\quad \hat \sigma(t,y):=\sigma\big(y-\chi(t)\big).
\end{aligned}
\end{equation*}
\end{theorem} 
\begin{proof}
The proof follows the construction performed in Sections \ref{sec:beta} and \ref{sec:wSmooth}, therefore we try to avoid repetitions. For each $\delta>0$ and $t\in[0,T)$ we define $b_\delta(t):=\inf\{z\in\cI: V(t,z)>\delta\}$. Then $b_\delta(t)\downarrow b(t)$ for $t\in[0,T)$ when $\delta\to 0$. Given that
\[
\dot b_\delta(t)=-\frac{\partial_t V(t,b_\delta(t))}{\partial_x V(t,b_\delta(t))},
\]
the condition in \eqref{eq:LB} yields
\[
b(t_2)-b(t_1)\ge -\int_{t_1}^{t_2}\zeta(s)\ud s,\quad 0\le t_1<t_2\le T.
\]
It is therefore obvious that $\beta(t)=b(t)+\int_0^t \zeta(s)\ud s$ is non-decreasing. Moreover, defining
\[
Y^{0,z}_t:= Z^z_t+\int_0^t\zeta(s)\ud s=Z^z_t+\chi(t),
\]
an application of It\^o's formula and simple calculations yield \eqref{eq:dynY}. Substituting $Z^z_t=Y^{0,z}_t-\chi(t)$ inside the objective function for $V(0,z)$ yields the equivalence $\hat V(0,z)=V(0,z)$. The stopping time $\tau_*=\inf\{t\in[0,T]: Z^z_t\le b(t)\}\wedge T$ is optimal for $V(0,z)$ and it is clear that 
\[
\hat \tau_*=\inf\{t\in[0,T]:Y^{0,z}_t\le \beta(t)\}\wedge T=\inf\{t\in[0,T]:Z^z_t+\chi(t)\le b(t)+\chi(t)\}\wedge T=\tau_*.
\] 
Hence, optimality of $\hat \tau_*$ also holds. 
\end{proof}

It is shown in the proof above that $\hat V(0,z)=V(0,z)$. However, to fully exploit the reparametrisation of the optimal stopping problem and the monotonicity of the new boundary $\beta(t)$ it is useful to notice that the change of coordinates is expressed by a mapping 
\[
(t,y)=\cM[(t,z)]=\big(t,z+\chi(t)\big)\quad\text{with inverse}\quad(t,z)=\cM^{-1}[(t,y)]=\big(t,y-\chi(t)\big).
\]
In particular, $\cM[(0,z)]=(0,z)$ and the equivalence between value functions holds in the forms $V=\hat V\circ\cM$ and $\hat V=V\circ\cM^{-1}$. That is, $V(t,z)=\hat V(t,z+\chi(t))=\hat V(t,y)$ and $\hat V(t,y)=V(t,y-\chi(t))=V(t,z)$. It then becomes clear that denoting $\hat \cC=\{(t,y):y>\beta(t)\}$, 
\[
(t,z)\in\cC\iff z>b(t)\iff y>\beta(t)\iff (t,y)\in\hat \cC, 
\]
and indeed $\hat \cC=\{(t,y):\hat V(t,y)>0\}=\{(t,y):y>\beta(t)\}$
This is the sort of transformations we leveraged on in the second half of Section \ref{sec:wSmooth}. In that section we additionally used the exponential structure of geometric Brownian motion and therefore the reparametrisation in \eqref{eq:beta} appears in a slightly different form to the one in the theorem above. More specifically, the analogue of \eqref{eq:LB} obtained in Lemma \ref{lem:b'_delta} can be interpreted as saying that $\zeta(t)=\Lambda b(t)$. At that point in the paper we know that $b$ is non-negative, upper semi-continuous and finite on $[0,T]$, hence bounded.

As a final point we determine sufficient conditions for \eqref{eq:LB} that overcome the limitations of the framework from \cite{de2019lip} described before the theorem. The discussion that precedes the proposition is a little informal, in order to avoid an overly technical exposition. Precise statements are provided in the proposition itself. 

We need a probabilistic representation for the spatial derivative of $V$ in the form
\begin{equation}\label{eq:Vz}
\partial_z V(t,z)=\E\Big[\int_0^{\tau_*^{t,z}}\partial_z H(t+s,Z^z_s)\partial_z Z^z_s\ud s+\ind_{\{\tau^{t,z}_*=T-t\}}\partial_z G(Z^z_{T-t})\partial_z Z^z_{T-t}\Big],
\end{equation} 
where $\tau^{t,z}_*$ is optimal in $V(t,z)$, $\partial_z H$ is sufficiently integrable, and $\partial_z Z^z$ is the derivative of the stochastic flow and it satisfies 
\begin{equation}\label{eq:Zchange}
\partial_z Z^z_t=\exp\Big(\int_0^t\big(\alpha'(Z^z_s)-\tfrac12[\sigma'(Z^z_s)]^2\big)\ud s+\int_0^t\sigma'(Z^z_s)\ud W_s\Big),
\end{equation}
with $\alpha'(z)$ and $\sigma'(z)$ the derivatives of the coefficients (cf.\ \cite[Ch.\ V.7]{protter2012stochastic}). General sufficient conditions for the representation in \eqref{eq:Vz} can be found in \cite[Thm.\ 3.1]{de2019lip}. However, those conditions can be significantly relaxed in specific examples as in the VA problem, where \eqref{eq:Vz} holds thanks to Lemma \ref{lem:w_z}. Since $\partial_z V$ appears in the denominator of \eqref{eq:LB}, we notice that if 
\[
\int_0^{T-t}\P\Big(\partial_z H(t+s,Z^z_s)>0,\,s<\tau^{t,z}_*\Big)\ud s>0\quad\text{or}\quad\P\big(\partial_z G(Z^z_{T-t})>0,\,\tau^{t,z}_*=T-t\big)>0,
\]
then $\partial_z V(t,z)>0$ for $(t,z)\in\cC$. Both conditions are satisfied in the VA problem (cf.\ \eqref{eq:w_z}).

An upper bound for $\partial_t V$ can be obtained as follows: for $\eps>0$, let
\[
\tau_\eps:=\tau^{t,z}_*\ind_{\{\tau^{t,z}_*<T-t\}}+(T-t+\eps)\ind_{\{\tau^{t,z}_*=T-t\}}.
\]
Choosing $\tau_\eps$ in the problem with value $V(t-\eps,z)$ yields
\begin{equation*}
\begin{aligned}
&V(t,z)-V(t-\eps,z)\\
&\le\E\Big[\int_0^{\tau^{t,z}_*}\big(H(t+s,Z^z_s)-H(t-\eps+s,Z^z_s)\big)\ud s+\ind_{\{\tau^{t,z}_*=T-t\}}\Big(G(Z^z_{T-t})-G(Z^z_{T-t+\eps})\Big)\Big] \\
&\quad-\E\Big[\ind_{\{\tau^{t,z}_*=T-t\}}\int_{T-t}^{T-t+\eps} H(t-\eps+s,Z^z_s)\ud s\Big].
\end{aligned}
\end{equation*}
Dividing by $\eps$ and letting $\eps\to 0$ the first and third term on the right-hand side converge to 
\begin{equation*}
\begin{aligned}
&\frac1\eps\E\Big[\int_0^{\tau^{t,z}_*}\big(H(t+s,Z^z_s)-H(t-\eps+s,Z^z_s)\big)\ud s\Big]\xrightarrow{\eps\to 0} \E\Big[\int_0^{\tau^{t,z}_*}\partial_t H(t+s,Z^z_s)\ud s\Big],\\
&\frac1\eps\E\Big[\ind_{\{\tau^{t,z}_*=T-t\}}\int_{T-t}^{T-t+\eps} H(t-\eps+s,Z^z_s)\ud s\Big]\xrightarrow{\eps\to 0}\E\Big[\ind_{\{\tau^{t,z}_*=T-t\}}H\big(T,Z^z_{T-t}\big)\Big],
\end{aligned}
\end{equation*}
where the limits are justified by assuming, e.g., that $H$ and $\partial_t H$ are sufficiently integrable ($H$ continuous). For the remaining term, by tower property and Markov property we get 
\begin{equation*}
\begin{aligned}
\E\Big[\ind_{\{\tau^{t,z}_*=T-t\}}\Big(G(Z^z_{T-t})-G(Z^z_{T-t+\eps})\Big)\Big]&=\E\Big[\ind_{\{\tau^{t,z}_*=T-t\}}\Big(G(Z^z_{T-t})-\E\big[G(Z^z_{T-t+\eps})\big|\cF_{T-t}\big]\Big)\Big]\\
&=\E\Big[\ind_{\{\tau^{t,z}_*=T-t\}}\Big(G(Z^z_{T-t})-\Psi(\eps,Z^z_{T-t})\Big)\Big],
\end{aligned}
\end{equation*}
where $\Psi(t,z):=\E\big[G(Z^z_t)\big]$ solves the forward PDE
\[
-\partial_t\Psi(t,z)+\cL \Psi(t,z)=0,\quad (t,z)\in[0,\infty)\times\cI,
\]
with $\psi(0,z)=	G(z)$. More precisely, if $G$ is continuous on $\overline \cI$ and the semigroup associated to the process $Z$ has good smoothing properties\footnote{For example, if $\alpha,\sigma\in C^1(\cI)$ and $\sigma(z)>0$ for $z\in\cI$.} then $\Psi\in C^{1,2}((0,\infty)\times\cI)\cap C([0,\infty)\times\overline\cI)$. Notice that for $z\in\cI$,
\[
G(z)-\Psi(\eps,z)=-\int_0^\eps\partial_t\Psi(s,z)\ud s=-\int_0^\eps\cL\Psi(s,z)\ud s.
\]
Motivated by the equation above and in analogy with Lemma \ref{lem:EuAmCall}, we assume that there is a function $\varphi\in C(\cI)$ such that
\begin{equation}\label{eq:linfG}
\limsup_{\eps\to 0}\frac1\eps\E\Big[\ind_{\{\tau^{t,z}_*=T-t\}}\Big(G(Z^z_{T-t})-\Psi(\eps,Z^z_{T-t})\Big)\Big]\le\E\Big[\varphi\big(Z^z_{T-t}\big)\Big].
\end{equation}
If for example $G\in C^2(\cI)$, then $\varphi(z)=-\cL G(z)$ (cf.\ the proof of Lemma \ref{lem:last}) and we recover the same result as in \cite[Thm.\ 3.1]{de2019lip}. Instead, in the proof of Lemma \ref{lem:w_tUpBound} (cf.\ \eqref{eq:phi(T-eps)}) we obtain the analogue of \eqref{eq:linfG} for $G(z)=(z-1)^+$.
  
Combining the above calculations we obtain the upper bound,
\begin{equation}\label{eq:ubut}
\partial_t V(t,z)\le \E\Big[\int_0^{\tau^{t,z}_*}\partial_t H(t+s,Z^z_s)\ud s-\ind_{\{\tau^{t,z}_*=T-t\}}\Big(H\big(T,Z^z_{T-t}\big)-\varphi(Z^z_{T-t})\Big)\Big].
\end{equation}  
Using \eqref{eq:Vz} and the above bound we can state sufficient conditions for \eqref{eq:LB}.
\begin{proposition}\label{prop:suff}
Suppose that \eqref{eq:Vz} and \eqref{eq:ubut} hold with $\partial_z V(t,z)>0$ for $(t,z)\in\cC$. Suppose that the process
\begin{equation}\label{eq:M}
M_t:=\exp\Big(\int_0^t\sigma'(Z^z_s)\ud W_s-\tfrac12\int_0^t[\sigma'(Z^z_s)]^2\ud s\Big),\quad t\in[0,T],
\end{equation}
is a martingale. Suppose that the following properties hold for the functions $\alpha$, $\sigma$, $H$ and $G$:
\begin{itemize}
\item[(i)] for all $z\in\cI$, we have $\sigma'(z)\ge 0$ and $|\alpha'(z)|\le p$ for some $p>0$;
\item[(ii)] $\partial_z H (t,z)\ge 0$ for $(t,z)\in\cC$ and $\partial_z G(z)\ge 0$ for $z\in\cI$;
\item[(iii)] there is $c_1>0$ such that $\partial_t H(t,z)\le c_1\partial_x H(t,z)$ for $(t,z)\in\cC$, and either $\sigma'(z)\equiv 0$ or $z\mapsto H(t,z)$ is convex for $(t,z)\in\cC$.
\end{itemize}
Then, \eqref{eq:LB} holds whenever one of the conditions below is satisfied:
\begin{itemize}
\item[(iv)] For $(t,z)\in\cC$, we have 
\[
H(T,Z^z_{T-t})-\varphi(Z^z_{T-t})\ge 0,\quad \text{on}\ \{\tau^{t,z}_*=T-t\}.
\]
\item[(v)] For $(t,z)\in\cC$, we have
\[
H(T,Z^z_{T-t})\ge 0,\quad \text{on}\ \{\tau^{t,z}_*=T-t\}
\]
and there is $c_2>0$ such that
\begin{equation}\label{eq:Gphi}
\E\big[\ind_{\{\tau^{t,z}_*=T-t\}}\varphi\big(Z^z_{T-t}\big)\big]\le c_2\E\big[\ind_{\{\tau^{t,z}_*=T-t\}}\partial_z G(Z^z_{T-t})\partial_z Z^z_{T-t}\big]. 
\end{equation}
\end{itemize}
\end{proposition}
\begin{proof}
Using \eqref{eq:Vz} and \eqref{eq:ubut}, with (ii) and (iii), we obtain for $(t,z)\in\cC$
\begin{equation*}
\begin{aligned}
\frac{\partial_t V(t,z)}{\partial_z V(t,z)}&\le \frac{\E\Big[\int_0^{\tau^{t,z}_*}\partial_t H(t+s,Z^z_s)\ud s-\ind_{\{\tau^{t,z}_*=T-t\}}\Big(H\big(T,Z^z_{T-t}\big)-\varphi(Z^z_{T-t})\Big)\Big]}{\E\Big[\int_0^{\tau_*^{t,z}}\partial_z H(t+s,Z^z_s)\partial_z Z^z_s\ud s+\ind_{\{\tau^{t,z}_*=T-t\}}\partial_z G(Z^z_{T-t})\partial_z Z^z_{T-t}\Big]}\\
&\le \frac{\E\Big[\int_0^{\tau^{t,z}_*}c_1\partial_z H(t+s,Z^z_s)\ud s-\ind_{\{\tau^{t,z}_*=T-t\}}\Big(H\big(T,Z^z_{T-t}\big)-\varphi(Z^z_{T-t})\Big)\Big]}{e^{-pT}\widetilde\E\Big[\int_0^{\tau_*^{t,z}}\partial_z H(t+s,Z^z_s)\ud s+\ind_{\{\tau^{t,z}_*=T-t\}}\partial_z G(Z^z_{T-t})\Big]},
\end{aligned}
\end{equation*}
where in the denominator of the final expression we have introduced the probability measure $\ud \widetilde\P=M_{T-t}\ud \P$ and also used that
\begin{equation}\label{eq:dZM}
	\partial_z Z^z_t=\exp\Big(\int_0^t\alpha'(Z^z_s)\ud s\Big)M_t\ge e^{-p T}M_t,
\end{equation}
thanks to (i).

Under condition (iv), using $\partial_z G\ge 0$ we get
\begin{equation*}
\begin{aligned}
\frac{\partial_t V(t,z)}{\partial_z V(t,z)}&\le \frac{\E\Big[\int_0^{\tau^{t,z}_*}c_1\partial_z H(t+s,Z^z_s)\ud s\Big]}{e^{-pT}\widetilde\E\Big[\int_0^{\tau_*^{t,z}}\partial_z H(t+s,Z^z_s)\ud s\Big]}.
\end{aligned}
\end{equation*}
Instead, under condition (v) we have
\begin{equation*}
\begin{aligned}
\frac{\partial_t V(t,z)}{\partial_z V(t,z)}&\le \frac{\E\Big[\int_0^{\tau^{t,z}_*}c_1\partial_z H(t+s,Z^z_s)\ud s+\ind_{\{\tau^{t,z}_*=T-t\}}c_2 \partial_z G(Z^z_{T-t})\partial_z Z^z_{T-t}\Big]}{e^{-pT}\widetilde\E\Big[\int_0^{\tau_*^{t,z}}\partial_z H(t+s,Z^z_s)\ud s+\ind_{\{\tau^{t,z}_*=T-t\}}\partial_z G(Z^z_{T-t})\Big]}\\
&\le \frac{\E\Big[\int_0^{\tau^{t,z}_*}c_1\partial_z H(t+s,Z^z_s)\ud s\Big]}{e^{-pT}\widetilde\E\Big[\int_0^{\tau_*^{t,z}}\partial_z H(t+s,Z^z_s)\ud s\Big]}+c_2 e^{2p T},
\end{aligned}
\end{equation*}
where the final term in the expression above is obtained by changing measure in the numerator on the event $\{\tau^{t,z}_*=T-t\}$.

Clearly, if $\sigma'(z)\equiv 0$ in (iii), then $\widetilde\P=\P$ and 
\begin{equation}\label{eq:ratioH}
\frac{\E\Big[\int_0^{\tau^{t,z}_*}c_1\partial_z H(t+s,Z^z_s)\ud s\Big]}{e^{-pT}\widetilde\E\Big[\int_0^{\tau_*^{t,z}}\partial_z H(t+s,Z^z_s)\ud s\Big]}\le c_1 e^{pT}.
\end{equation}
In that case \eqref{eq:LB} holds with $\zeta(t)=c_1 e^{pT}+c_2 e^{2pT}$, under both (iv) and (v).

Let us now consider the case when the measure change is not trivial, i.e., when $\sigma'(z)\neq 0$ but $z\mapsto H(t,z)$ is convex for $(t,z)\in\cC$. Under the measure $\widetilde \P$ the dynamics of $Z$ reads as
\[
\ud Z^z_s=\big(\alpha(Z^z_s)+\sigma(Z^z_s)\sigma'(Z^z_s)\big)\ud s+\sigma(Z^z_s)\ud \widetilde W_s,
\]
where $\widetilde W_s=W_s-\int_0^s\sigma'(Z^z_r)\ud r$ is a Brownian motion by Girsanov theorem. Let us introduce an auxiliary process $\hat Z$ with dynamics
\[
\ud \hat Z^z_s=\big(\alpha(\hat Z^z_s)+\sigma(\hat Z^z_s)\sigma'(\hat Z^z_s)\big)\ud s+\sigma(\hat Z^z_s)\ud W_s,
\]
with respect to the original Brownian motion. Let us also introduce
\[
\hat \tau^{t,z}_*:=\inf\{s\in[0,T-t]: \hat Z^z_s\le b(t+s)\}\wedge(T-t).
\]
Then, it is clear that the law of the pair $(Z^z_\cdot,\tau^{t,z}_*)$ under $\widetilde \P$ is the same as the law of the pair $(\hat Z^z_\cdot,\hat \tau^{t,x}_*)$ under $\P$. Moreover, using $\sigma'(z)\ge 0$ from (i), we get the pathwise comparison, $\P$-a.s.
\[
\hat Z^z_s\ge Z^z_s,\quad \forall s\in[0,\infty).
\]
That in turn implies $\hat \tau^{t,z}_*\ge \tau^{t,z}_*$, $\P$-a.s. Then,
\begin{equation*}
\begin{aligned}
\widetilde\E\Big[\int_0^{\tau_*^{t,z}}\partial_z H(t+s,Z^z_s)\ud s\Big]&= \E\Big[\int_0^{\hat \tau_*^{t,z}}\partial_z H(t+s,\hat Z^z_s)\ud s\Big]\\
&\ge \E\Big[\int_0^{\tau_*^{t,z}}\partial_z H(t+s,\hat Z^z_s)\ud s\Big]\ge \E\Big[\int_0^{\tau_*^{t,z}}\partial_z H(t+s,Z^z_s)\ud s\Big],
\end{aligned}
\end{equation*}  
where the first inequality holds because $\partial_z H\ge 0$ and the second one because $z\mapsto \partial_z H(t,z)$ is nondecreasing. Using the inequality above yields once again \eqref{eq:ratioH}. Then \eqref{eq:LB} holds with $\zeta(t)=c_1 e^{pT}+c_2 e^{2pT}$, under both (iv) and (v).
\end{proof}
In the VA problem we have that (i)--(iii) and (v) hold. In particular, the first requirement in (iii) is related to Assumption \ref{A2}. The first assumption in (v) is rather natural. In the VA problem it is verified in Corollary \ref{lem:Z>1}. On an intuitive level it models a situation in which it is not optimal to continue near maturity while incurring a strictly negative cost $H(t+s,Z^z_s)<0$. One should notice that the assumption implicitly links the behaviour of the function $H$ near maturity with the sign/growth of the function $G$. The second assumption in (v) relates to the comparison of the last two terms on the right-hand side of \eqref{eq:w_tleq} with the last term on the right-hand side of \eqref{eq:zw_z}.

To conclude, it is worth elucidating condition \eqref{eq:Gphi} in (v) a bit further as it may appear rather abstract.
\begin{lemma}\label{lem:last}
Condition \eqref{eq:Gphi} in Proposition \ref{prop:suff} is satisfied if $M_t$ in \eqref{eq:M} is a martingale, $\sigma'(z)\ge 0$, $|\alpha'(z)|\le p$ for some $p>0$, $\partial_z G\ge 0$, $\varphi(z)\le c_3 \partial_z G(z)$ for some $c_3>0$, and either $\sigma'(z)\equiv 0$ or $G$ is convex.

In particular, if $G\in C^2(\cI)$ is such that $\partial_z G\ge 0$ and
\[
\sup_{0\le s\le T}|\E[\cL G(Z^z_s)]|<\infty,\quad z\in\cI, 
\]
and either $\alpha(z)\ge 0$ or $|\alpha(z)|\le c_3$ for $z\in\cI$, then $\varphi(z)\le c_3 \partial_z G(z)$.
\end{lemma}
\begin{proof}
We start by proving the final claim. Recall that $\Psi(t,z)=\E[G(Z^z_t)]$. For $G\in C^2(\cI)$, using It\^o's formula we have 
\[
\Psi(\eps,z)-G(z)=\int_0^\eps\E\big[\cL G(Z^z_s)\big]\ud s.
\]
Therefore, 
\[
\varphi(z)=-\lim_{\eps\to 0}\frac{\Psi(\eps,z)-G(z)}{\eps}=-\cL G(z).
\]
Using convexity of $G$ we get $\varphi(z)\le -\alpha(z)\partial_z G(z)$ and the claim easily follows from the assumptions on $\alpha(z)$.

Let us now prove the first claim in the lemma. By the assumed bound $\varphi(z)\le c_3\partial_z G(z)$ we get
\begin{equation*}
\begin{aligned}
\E\big[\ind_{\{\tau^{t,z}_*=T-t\}}\varphi\big(Z^z_{T-t}\big)\big]\le c_3\E\big[\ind_{\{\tau^{t,z}_*=T-t\}}\partial_z G\big(Z^z_{T-t}\big)\big]
\end{aligned}
\end{equation*}
Recall from the proof of Proposition \ref{prop:suff} that $\hat Z^z_s\ge Z^z_s$ and $\hat \tau^{t,z}_*\ge \tau^{t,z}_*$. Then,
\begin{equation*}
\begin{aligned}
\E\big[\ind_{\{\tau^{t,z}_*=T-t\}}\partial_z G\big(Z^z_{T-t}\big)\big]&\le \E\big[\ind_{\{\hat \tau^{t,z}_*=T-t\}}\partial_z G\big(\hat Z^z_{T-t}\big)\big]\\
&=\widetilde \E\big[\ind_{\{\tau^{t,z}_*=T-t\}}\partial_z G\big(Z^z_{T-t}\big)\big]\\
&=\E\big[\ind_{\{\tau^{t,z}_*=T-t\}}\partial_z G\big(Z^z_{T-t}\big)M_{T-t}\big]\\
&\le e^{pT}\E\big[\ind_{\{\tau^{t,z}_*=T-t\}}\partial_z G\big(Z^z_{T-t}\big)\partial_z Z^z_{T-t}\big],
\end{aligned}
\end{equation*}
where the first inequality is by convexity of $G$ and $\partial_z G\ge 0$, the first equality is by the equivalence in law of $(\hat Z^z_\cdot,\hat \tau^{t,z}_*)$ and $(Z^z_\cdot,\tau^{t,z}_*)$, the second equality is by definition of $\widetilde \P$ and the final inequality uses \eqref{eq:Zchange} along with the boundedness of $\alpha'$ (recall \eqref{eq:dZM}). Combining the bounds above, yields the claim with $c_2=c_3\exp (pT)$.
\end{proof}

\medskip
\noindent{\bf Funding}: T.\ De Angelis received partial financial support from EU -- Next Generation EU -- PRIN2022 (2022BEMMLZ) CUP: D53D23005780006 and PRIN-PNRR2022 (P20224TM7Z) CUP: D53D23018780001. G. Stabile received financial support from EU -- Next Generation EU -- PRIN2022 (2022FWZ2CR) CUP:
B53D23009970006. G. Stabile was also partially supported by Sapienza University of Rome, research project ``\emph{L'opzione di riscatto nelle polizze rivalutabili: l'impatto delle penalit\`a sulle scelte dei sottoscrittori}'', grant no.~RP12218167E24F87.
\medskip

\bibliography{bibfile}{}
\bibliographystyle{abbrv}
\end{document}